\documentclass[12pt]{amsart}
\usepackage[top=1.25in,bottom=1in,left=1in,right=1in]{geometry}
\usepackage[utf8]{inputenc} 
\usepackage{amsmath}
\usepackage{amssymb}
\usepackage{amsthm}
\usepackage{bbm}
\usepackage{xcolor}
\usepackage{graphicx}
\usepackage{float}
\usepackage{enumerate}
\usepackage{natbib}
\usepackage{setspace,lipsum}
\usepackage{bm}

\usepackage{mathrsfs}
\usepackage{comment}
\usepackage{apptools}
\usepackage{url}
\newcommand{\RN}[1]{  \textup{\uppercase\expandafter{\romannumeral#1}}}

\usepackage{booktabs}

\usepackage[subrefformat=parens]{subcaption}
\usepackage{mathtools}

\newtheorem{claim}{Claim}

\newtheorem{proposition}{Proposition}
\newtheorem{theorem}{Theorem}
\newtheorem{corollary}{Corollary}
\newtheorem{lemma}{Lemma}

\newtheorem{observation}{Observation}
\theoremstyle{remark}
\newtheorem{example}{Example}
\theoremstyle{remark}

\newtheorem{definition}{Definition}

\usepackage{etoolbox}
\makeatletter
\patchcmd{\@maketitle}
  {\ifx\@empty\@dedicatory}
  {\ifx\@empty\@date \else {\vskip3ex \centering\footnotesize\@date\par\vskip1ex}\fi
   \ifx\@empty\@dedicatory}
  {}{}
\patchcmd{\@adminfootnotes}
  {\ifx\@empty\@date\else \@footnotetext{\@setdate}\fi}
  {}{}{}
\makeatother

\makeatother

\AtAppendix{
\numberwithin{equation}{section}
\numberwithin{definition}{section}
\numberwithin{theorem}{section}
\numberwithin{proposition}{section}
\numberwithin{lemma}{section}
\numberwithin{corollary}{section}}

\DeclareMathOperator*{\argmin}{arg\,min}
\newcommand{\indep}{\perp \!\!\! \perp}

\allowdisplaybreaks

\title[Algorithm Design: A Fairness-Accuracy Frontier]{Algorithm Design:\\ A Fairness-Accuracy Frontier}
\thanks{We thank Nageeb Ali, Eric Auerbach, Simon Board, Krishna Dasaratha, Will Dobbie, Alex Frankel, Ben Golub, Sergiu Hart, Peter Hull, Navin Kartik, Yair Livne, Sendhil Mullainathan, Derek Neal, Jose Montiel Olea, Larry Samuelson, and Max Tabord-Meehan for helpful comments, and National Science Foundation Grant SES-1851629 for financial support. We also thank Andrei Iakovlev and Aristotle Magganas for valuable research assistance on this project.}
\author[Annie Liang]{Annie Liang$^\dag$}
\thanks{$^\dag$Northwestern University}
\author[Jay Lu]{Jay Lu$^\S$}
\thanks{$^\S$UCLA}
\author[Xiaosheng Mu]{Xiaosheng Mu$^\ddag$}
\thanks{$^\ddag$Princeton University}
\author[Kyohei Okumura]{Kyohei Okumura$^\dag$}
\date{\today} 
\begin{document}

\maketitle

\begin{abstract}
\singlespacing

Algorithm designers increasingly optimize not only for accuracy, but also for the fairness of the algorithm across pre-defined groups. We study the tradeoff between fairness and accuracy for any given set of inputs to the algorithm. We propose and characterize a fairness-accuracy frontier, which consists of the optimal points across a broad range of preferences over fairness and accuracy. Our results identify a simple property of the inputs, \emph{group-balance}, which determines the shape of the frontier. We further study an information design problem where the designer flexibly regulates the inputs (e.g., by coarsening an input or banning its use) but the algorithm is chosen by another agent. Whether it is optimal to ban an input generally depends on the designer's preferences. But when inputs are group-balanced, then excluding group identity is strictly suboptimal for all designers, and when the designer has access to group identity, then it is strictly suboptimal to exclude any informative input.

\end{abstract}

\section{Introduction}

Suppose a hospital uses a  machine learning algorithm to aid in the diagnosis of a  medical condition, where the algorithm makes the correct diagnosis 90\% of the time for Red patients but only 50\% of the time for Blue patients. Such an outcome---where the consequences of a policy differ systematically across two groups---is known as \emph{disparate impact}. Across a wide range of applications, algorithms have been shown to have disparate impact \citep*{ArnoldDobbieHull,Pinkham}. For example, patients assigned the same risk score by a healthcare algorithm were shown to have substantially different actual health risks depending on their race \citep*{ObermeyerMullainathan}; the false-positive rate of an algorithm used to predict criminal reoffense was shown to be twice as high for Black defendants as for White defendants \citep{ProPublica}; and the accuracies of facial-recognition technologies  vary substantially across demographic groups \citep{Klareetal}. These findings have led  designers to impose group-based fairness constraints on algorithms in settings ranging from healthcare to bail evaluation to lending \citep{EthicalAlg,HardtPriceSrebro}.

Policymakers prefer algorithms that have lower disparate impact but also prefer algorithms that are more accurate. In an ideal world, both goals could be achieved together; in practice, there may be an intrinsic tradeoff between accuracy (the overall error rate of the algorithm) and fairness (how similar the error rates are across pre-defined groups).\footnote{Equity-efficiency tradeoffs such as this have been studied in  settings as diverse as taxation \citep*{SaezStantcheva16,DworczakKominersAkbarpour21}, policing \citep*{Persico,VohraLeePaiRoth}, and college admissions \citep{ChanEyster,EllisonPathak}.} What this tradeoff looks like depends on the inputs to the algorithm, which can be observed, manipulated, and regulated---raising the following  questions: How does the tradeoff between fairness and accuracy depend on the information available for prediction? Which informational environments create a tension between fairness and accuracy, and which ameliorate it? While the tradeoff between fairness and accuracy has been empirically computed in specific applications \citep{WeiNiethammer,Goel,LittleWylandtAllen}, we know substantially less about how the available information shapes the tension between these two objectives in general.

In this paper, we address these questions  by defining and studying a \emph{fairness-accuracy frontier}. This frontier consists of all outcomes that are optimal for some objective function in a  broad class that spans different views on how to trade off fairness and accuracy. We prove two types of results about the frontier. First, we identify simple properties of the algorithmic inputs that determine the shape of this frontier. Second, we take an information design perspective on understanding how constraints on information can induce certain desired outcomes. Specifically, we consider an interaction between a policymaker flexibly constraining the inputs and an agent setting the algorithm. We characterize what part of the fairness-accuracy frontier the designer can achieve through appropriate design of the inputs and examine whether it might be optimal for the designer to exclude an input altogether (e.g., excluding group identity in the context of medical predictions).

In our model, a designer chooses an algorithm that takes observed covariates as inputs (e.g., medical scans, lab tests, records of prior hospital visits) and outputs a decision (e.g., whether to recommend a medical procedure). The algorithm's consequences for any given individual are evaluated using a loss function, which can be interpreted as the inaccuracy or harm of the decision. We aggregate losses within two groups, group $r$ (Red) and group $b$ (Blue). Each group's \emph{error} is the expected loss for individuals of that group. An algorithm is understood to be more accurate if it implies lower errors for both groups, and more fair if it implies a smaller absolute difference between the two groups' errors.

To understand the tradeoff between fairness and accuracy, we define the class of \emph{fairness-accuracy (FA) preferences} to be all preferences over group errors that are consistent with the following order: one pair of group errors \emph{FA-dominates} another if the former involves smaller errors for both groups (greater accuracy) and also a smaller difference between group errors (greater fairness).\footnote{ We do not take a stance on the normative desirability of these preferences, instead interpreting our class as encompassing the broad range of  designer preferences that could be relevant in practice.} This strict order is consistent with a broad range of designer preferences, including Utilitarian designers (who minimize the aggregate error in the population), Rawlsian designers (who minimize the greater of the two group errors), and Egalitarian designers (who minimize the absolute difference between group errors) among others. Some of these preferences also correspond directly to optimization problems that have been proposed for use in practice.\footnote{For example, optimizing a Rawlsian preference is equivalent to implementing group distributionally robust optimization \citep{GroupDRO}, and optimizing an Egalitarian preference is equivalent (on a restricted domain) to maximizing accuracy subject to equality of error rates (as considered in \citet{HardtPriceSrebro} among others).}  We define the \emph{fairness-accuracy   frontier} to be the set of all group error pairs that are feasible (i.e., can be implemented by some algorithm given the observed covariates), and moreover FA-undominated within the feasible set. That is, these error pairs  cannot be  improved upon simultaneously in  accuracy and fairness.

A simple property of the algorithm's inputs turns out to be critical for determining the shape of the fairness-accuracy frontier. Say that a covariate vector is \emph{group-balanced} if a group's optimal algorithm (i.e., the one that gives that group the smallest error over all feasible algorithms) yields a lower error for that group than for the other group. Otherwise,  say that the covariate vector is \emph{group-skewed}. While it may be difficult to anticipate in advance of a comprehensive empirical analysis whether group-balance or group-skew is more typical in practice, one scenario in which the latter may arise is if covariates have the same implications for both groups but are measured more accurately for one group than the other (e.g., medical data is recorded more accurately for high-income patients than low-income patients).

\begin{figure}[h]
\begin{center}
\includegraphics[scale=0.6]{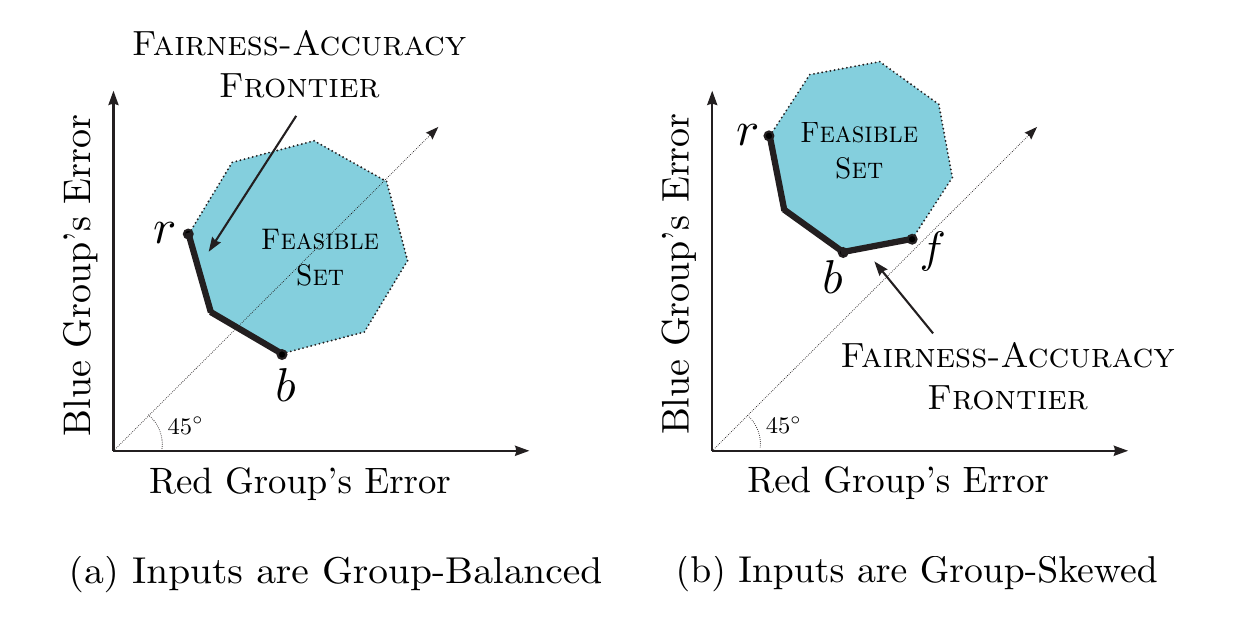}
\end{center}
\caption{The Fairness-Accuracy   Frontier.}
\label{fig:general_intro}
\end{figure}

We provide three complementary characterizations of the fairness-accuracy frontier. Our first and main characterization says that the fairness-accuracy frontier takes either of two possible forms depending on whether the covariate vector is group-balanced or group-skewed (see Figure  \ref{fig:general_intro} above). When inputs are group-balanced, the fairness-accuracy frontier is exactly the standard Pareto frontier, i.e. the set of all feasible error pairs that cannot be simultaneously reduced in both coordinates. This is the part of the lower boundary that begins at the point that is best for group Red (labeled $r$) and ends at the point that is best for group Blue (labeled $b$). Here, the tradeoff is between the accuracy for one group and the accuracy for the other group. On the other hand, when inputs are group-skewed, the fairness-accuracy frontier includes not only the standard Pareto frontier but additionally a positively-sloped part  (the segment from $b$ to the fairness-maximizing point $f$ in Figure \ref{fig:general_intro}) along which both groups' errors increase but the gap between them decreases. Here, the tradeoff is between the accuracy of \emph{both} groups and fairness, signifying a strong conflict between the two objectives. We can conclude from this characterization that a policy proposal that increases errors for both groups, but decreases their gap, can never be justified by fairness considerations if the covariate vector is group-balanced; in contrast, such a policy could potentially be justified were the covariate vector group-skewed.

Our second characterization shows that the fairness-accuracy frontier can be completely recovered as the optimal points for a class of ``simple'' preferences that linearly trade off fairness with accuracy for each group. By varying the weights on fairness versus accuracy, we can trace out the entire fairness-accuracy frontier, and we exploit this parametrization in subsequent comparative static exercises.

Our third and final characterization describes the algorithms that implement the error pairs along the fairness-accuracy frontier. These algorithms turn out to possess a simple threshold structure. Loosely speaking, we can order covariate vectors by how much the decision at each covariate vector impacts group $r$'s error relative to that of group $b$. The algorithms that implement the fairness-accuracy frontier have the property of assigning all early covariate vectors in this sequence to the worst decision for group $g$ (for some group $g$ that we identify), and all the later covariate vectors to the best decision for group $g$. Moving along the fairness-accuracy frontier corresponds to successively swapping the next covariate vector in this sequence from the best decision to the worst decision for group $g$. Our characterization shows that these algorithms (and mixtures between any neighboring algorithms) exactly describe the fairness-accuracy frontier.

In the second part of the paper, we investigate what happens if the designer does not choose the algorithm, but instead regulates the inputs of the algorithm. This question is motivated by settings where a designer has fairness concerns, but another agent setting the algorithm does not. For example, a healthcare provider (agent) determining treatment may seek to maximize the number of correct diagnoses, while a policymaker (designer) may additionally prefer that the accuracy of the provider's treatments be equitable across certain social groups.  In these cases, the policymaker can impose regulation that restricts the inputs available to the algorithm, for example by banning the use of a specific input.

We model this as an information design problem \citep{KamenicaGentzkow} where the designer chooses a garbling of the available inputs, and  an agent chooses an algorithm (based on the garbling) to maximize accuracy. Under weak conditions, it turns out to be without loss for the designer to only control the algorithm's inputs. That is,  any error pair that a designer would choose to implement given full control of the algorithm can also be achieved by appropriately garbling the inputs.

Might the optimal garbling involve completely excluding a covariate from use in the algorithm? We demonstrate two results: First, excluding group identity as an algorithmic input is strictly welfare-reducing for all designers (with FA preferences) if and only if the permitted covariates are group-balanced. Second, when group identity is permitted as an input, then completely excluding any other covariate  makes every designer strictly worse off, so long as that covariate satisfies a mild condition that we call \emph{decision-relevance}. This result can be applied, for instance, to thinking through the relationship between banning group identity and excluding standardized test scores in university admissions decisions. It suggests that when group identity is a permissible input in admission decisions, then excluding test scores is welfare-reducing for \emph{all} designers with the power to flexibly garble covariates. On the other hand, when group identity is not  permitted as an input in college admissions decisions (as is now the case in the United States following the Supreme Court decision in \emph{Students for Fair Admissions, Inc. v. President and Fellows of Harvard College}), then the optimal garbling of covariates for some designer preference may indeed involve completely excluding test scores, and we provide a simple example to illustrate this effect (Section \ref{sec:SimpleExample}). 

Finally, we conclude by illustrating our key definitions and results on two popular healthcare datasets from the algorithmic fairness literature. We conduct two analyses: First, we evaluate whether the covariates in these datasets are group-balanced or group-skewed (finding that one is group-skewed and the other is an interesting case of group-balance in which the group-optimal points lie on the 45 degree line). Second, we depict the fairness-accuracy frontier. This allows us to study the limits of input design for each of the datasets (it turns out to be without loss for the designer in one dataset and with loss in the other) and also to study how the frontier changes when group-specific algorithms are permitted. These empirical illustrations of our framework speak to its potential relevance for assessing fairness-accuracy tradeoffs in practical algorithm design problems.

\subsection{Related Literature} \label{sec:RelatedLit}

Our paper is motivated by recent problems regarding algorithmic bias (Section \ref{sec:AlgFairness}), but adopts a novel perspective on these questions based on approaches from two literatures in economic theory: the  literature on information design (Section \ref{sec:InfoDesign}) and the literature on social preferences and inequality (Section \ref{sec:socialpref}). Building on the former, we model the interaction between a designer flexibly regulating inputs and an agent setting the algorithm. Building on the latter, we focus on understanding  equity-efficiency tradeoffs, and  consider a wide class of preferences that reflects heterogeneity in social preferences.

\subsubsection{Information Design} \label{sec:InfoDesign}
One contribution of our paper is the modeling of the design of algorithmic inputs as an information design problem (see \citet{Kamenica} and \citet{BergemannMorris2019} for recent surveys). This approach complements previous frameworks for modeling the regulation of algorithms, in which policymakers communicate information via cheap talk \citep{CowgillStevenson} or impose restrictions directly on the algorithm \citep*{YangDobbie,RambachanEtAl,BlattnerSpiess}. Our focus on the strategic interaction between an information designer and an agent choosing the algorithm also complements \citet{DovalSmolin}'s characterization of the feasible set across different information policies.\footnote{For example, \citet{DovalSmolin} show that excluding inputs is suboptimal in the sense that more information necessarily increases the feasible set of payoffs. In contrast, in our model it may be strictly optimal for the designer to exclude an input, since a different agent chooses which payoff vector is implemented from among our feasible set.} We view the garbling of inputs as a potentially effective policy tool, which can be implemented through a variety of technological or legal commitments,\footnote{For example, organizations such as the US Census Bureau, Google, Apple, and Microsoft are committed to differential privacy initiatives \citep{DworkRoth}, which take various forms of adding noise to user inputs. \citet{YangDobbie} summarizes the existing law on algorithmic regulation and proposes new legal policies for mitigating algorithmic bias.} and which deserves further attention within the context of algorithmic fairness.  

Conversely, problems regarding algorithmic fairness motivate analyses that depart from typical information design problems in a few interesting ways. For example, the Sender in our framework cannot choose a completely flexible information structure, but is instead constrained to garblings of a primitive covariate vector. Additionally, motivated by heterogeneous attitudes toward fairness (Section \ref{sec:socialpref}), we focus on a frontier of solutions with respect to a wide class of Sender preferences. Our results in Section \ref{sec:AddCovariate} describe how the frontier of solutions responds to changes in the underlying information.  We focus on special cases of this comparative static  that are of interest given our motivation (e.g., adding or removing a covariate), but a more general solution (analogous to \citet{CurelloSinander}'s work on  comparative statics with respect to the Sender's utility function) would be an interesting avenue for future work.
 
 Finally, at the broader intersection of information design and algorithms, \citet{Ichihashi2023} considers optimal information acquisition for crime deterrence, and \citet*{CaplinMartin} draws a connection between different machine learning objectives and costly information design.

\subsubsection{Social Preferences and Inequality}\label{sec:socialpref}

The  literature on social preferences documents substantial heterogeneity in how individuals assess efficiency-equity tradeoffs \citep*{AndreoniMiller,FehrSchmidt,FismanKarivMarkovits,Sullivan}, which is reflected in our broad class of FA-preferences. In this literature, social preferences are  preferences over individual payoffs rather than preferences over group errors, but most have analogues in our setting. For example, the ``social welfare approach" aggregates individual payoffs using differential weights   \citep*{CharnessRabin02,SaezStantcheva16,DworczakKominersAkbarpour21}, and is nested in our class of FA preferences (if we interpret individual payoffs as group errors). We additionally allow for a direct penalty for unequal outcomes, as in the models of ``difference aversion"  or ``inequity aversion" \citep{Loewensteinetal,FehrSchmidt,BoltonOckenfels}.\footnote{Another part of this literature is concerned with intentions and reciprocity \citep{Rabin93,CharnessRabin02} and is outside of our model.}

There is a separate literature studying the equity-efficiency tradeoffs of affirmative action programs. Specifically, \citet{Lundberg} and \citet{ChanEyster} model affirmative action as a ban on the use of group identity in admissions decisions, and show that this can lead organizations to condition on proxies in a way that reduces both efficiency and equity.\footnote{Another set of papers shows that access to group identity must weakly improve the designer's payoffs when the designer has control of the algorithm (see for example \citet{MenonWilliamson}, \citet{Agarwaletal}, \citet{LiptonChouldechovaMcAuley}, \citet{Manski}, and \citet{RambachanEtAl}), as adding group identity is a Blackwell improvement in information.  This is no longer generally the case when the designer cannot choose the algorithm, as in our model in Section \ref{sec:DesignInputs}.} (A similar point is made in \citet{AganStarr} regarding the use of prior criminal history in hiring decisions in ``ban-the-box" policies.) 
\citet{EllisonPathak} empirically quantify the equity and efficiency losses of  race-neutral affirmative action (based on geographic proxies for race) as compared to plans that explicitly consider race.  These papers are related to our study of the impact of excluding group identity, but focus on how a designer's optimal algorithm given group identity compares to the optimal algorithm without. We instead examine how the frontier of feasible outcomes changes when the designer can design a group-dependent garbling versus when the designer must choose a group-independent garbling. These analyses are not nested; see Section \ref{sec:GNew} for more detail.

\subsubsection{Algorithmic Bias} \label{sec:AlgFairness}

The recent literature on algorithmic bias has emerged around the concern that algorithms have error rates that differ substantially across social and demographic groups  (see \citet{aerp&p} and \citet{CowgillTucker} for overviews). In this literature and in the accompanying policy discussion (e.g, \citet{ProPublica}), algorithms are often considered to be ``less fair" if  the harms of the algorithm are more unequally borne across groups, with this comparison formalized as a larger disparity in  error rates across groups \citep{HardtPriceSrebro,KMR,chouldechova}.\footnote{A notable exception is the concept of  individual fairness proposed in \citet{DworkHardtPitassiReingoldZemel}.} 
A growing body of empirical work documents and quantifies these disparate impacts \citep{ObermeyerMullainathan, ArnoldDobbieHull, Pinkham}. 

The tradeoff between accuracy (error rates of the algorithm for each group) and fairness (discrepancy between error rates across groups) is a special kind of equity-efficiency tradeoff. A common approach for resolving this tradeoff is to posit a particular objective criterion \citep{HardtPriceSrebro,DianaDickElzaynKearnsRothSchutzmanSharifiZiani}.  Other papers identify improvements with respect to both objectives simultaneously \citep{Rose,FeigenbergMiller}. Our paper is closest to a smaller part of this literature, which engages with the tension between fairness and accuracy by quantifying fairness-accuracy tradeoffs for specific loss functions \citep{MenonWilliamson} or for specific empirical applications \citep{WeiNiethammer,Goel,LittleWylandtAllen}. We are interested in how this fairness-accuracy tradeoff is moderated by the inputs to the algorithm in general, and provide simple conditions on the inputs that qualitatively govern this tradeoff independently of other details of the loss function or informational environment.

\section{Framework} \label{sec:Framework}

\subsection{Setup and Notation} \label{sec:model}

There is a population of individuals, where each individual is described by a \emph{covariate vector} $X$ taking values in a finite set $\mathcal{X}$, a  \emph{type} $Y$ taking values in a finite set $\mathcal{Y}$,\footnote{We make the  assumption of finiteness to simplify various notations in the exposition. Most of our results generalize to infinite covariate values and/or infinite types.} and a \emph{group identity} $G $ taking values $r$ or $b$.\footnote{Throughout, we assume the definition of the relevant groups to be a primitive of the setting, determined by sociopolitical precedent and outside the scope of our model.}
Throughout we think of $G, X, Y$ as random variables with joint distribution $\mathbb{P}$, and use $p_g := \mathbb{P}(G=g)>0$ to denote the fraction of the population that belongs to group $g \in \{r,b\}$. We impose no additional assumptions on the joint distribution,\footnote{We view $\mathbb{P}$ as the population distribution on which the algorithm is both trained and tested. An interesting direction for future work would be to consider training data that differs in distribution from the data on which the algorithm's errors are evaluated. For example, one could study the optimal sampling of data on which to train the algorithm, or to study feedback loops when the algorithm is trained on data determined by previous algorithms (as in \citet{VohraLeePaiRoth} and \cite{CheKimZhong}).} permitting for example each of the following:

\begin{example}[$X$ reveals or closely proxies for $G$] \label{ex:XRevealsG} The group identity may be an input in the covariate vector $X$, or predictable from inputs in the covariate vector $X$. For example, \citet{BertrandKamenica} show that  data on consumption patterns permits near perfect classification of gender  and a fairly accurate prediction of  other group identities such as income bracket, race, and political ideology.
\end{example}

\begin{example}[Biased Covariates] \label{ex:Bias} The value of an input in $X$ may be systematically biased depending on group identity. For example, if $G$ is income bracket, $Y$ is ability, and $X$ is a test score that can be improved through better access to test prep, the distribution $\mathbb{P}$ may have the property that at every ability level, the conditional distribution of test scores is shifted higher for students in the high-income bracket (i.e., the distribution of $X \mid Y=y,G=r$ first-order stochastically dominates $X \mid Y=y, G=b$ at every $y \in\mathcal{Y}$).  
\end{example}

\begin{example}[Asymmetrically Informative Covariates] \label{ex:AsymmetricInformative} The inputs in $X$ may be more informative about $Y$ for one group than the other. For example, in \citet{ObermeyerMullainathan}, a patient's health care costs are more predictive of their health care needs for White patients than for Black patients, and \citet{Rothstein} shows that SAT scores are more informative about future college grades for high-income students than low-income students.
\end{example}

An \emph{algorithm} is a mapping $a: \mathcal{X} \rightarrow \mathcal{D}$ from realizations of the covariate vector into decisions in $\mathcal{D} = \{0,1\}$. There is a primitive set of algorithms $\mathcal{A}$ (for example, the set of all linear algorithms) and the designer chooses from $\Delta\mathcal{A}$, the set of all randomizations over $\mathcal{A}$. We refer to these randomizations as ``algorithms" whenever there is no risk of confusion. Most of our subsequent results hold for arbitrary choices of $\mathcal{A}$; some of our results hold specifically for the leading special case where $\mathcal{A}$ is the set of all possible algorithms, denoted by $\overline{\mathcal{A}}$, and we are explicit when this is the case.\footnote{Although we refer to $X$ as a covariate vector and $a$ as an algorithm, one can equivalently interpret $X$ as a standard Blackwell signal and $a$ as a decision rule. In this case a restricted class of algorithms would correspond to restrictions on permissible decision rules.}

Some motivating examples of types, group identities, covariate vectors, and decisions are given below:

\smallskip

\emph{Healthcare.} $Y$ is need of treatment, $G$ is socioeconomic class, and the decision is whether the individual receives treatment. The covariate vector $X$ includes possible attributes such as image scans, number of past hospital visits, family history of illness, and blood tests.

\emph{Credit scoring.} $Y$ is creditworthiness, $G$ is gender, and the decision is whether the borrower's loan request is approved. The covariate vector $X$ includes possible attributes such as purchase histories, social network data, income level, and past defaults.

 \emph{Bail.} $Y$ is whether an individual is high-risk or low-risk of criminal reoffense, $G$ is race, and the decision is whether the individual is released on bail. The covariate vector $X$ includes possible attributes such as the individual's past criminal record, psychological evaluations, family criminal background, frequency of moves, or drug use as a child.

\emph{Job hiring.}  $Y$ is whether a job applicant is high or low quality,  $G$ is citizenship, and the decision is whether the applicant is hired. The covariate vector $X$ includes possible attributes such as past work history, resume, and references.

\medskip

The consequence of choosing decision $d$ for an individual whose true type is $y$ is evaluated using a loss function $\ell: \mathcal{D} \times \mathcal{Y} \rightarrow \mathbb{R}$, which we view as a measure of inaccuracy independent of fairness.\footnote{All the main results of the paper extend if we allow the loss function to also depend on group identity, i.e. $\ell: \mathcal{D}\times \mathcal{Y} \times \mathcal{G} \rightarrow \mathbb{R}$. This generalization accommodates additional fairness metrics from the literature, such as equalized odds---see Appendix \ref{app:EqualizedOdds} for  details. \label{fn:GroupDependence}} We further aggregate these losses across individuals within each group:

\begin{definition} For any algorithm $a \in \Delta\mathcal{A}$ and group $g \in \{r,b\}$, the \emph{group $g$ error} is  
\[e_g(a) := \mathbb{E}_{D \sim a(X)}\left[\ell(D,Y) \mid G=g\right].\]
\end{definition}

\noindent That is, group $g$'s error is the average loss for members of group $g$. For example, if the type $Y$ is binary and we consider the misclassification loss function $\ell(d,y)= \mathbbm{1}(d\neq y)$, then $e_g(a)$ is the total probability of a type I or type II error. Other loss functions may put different weights on different kinds of errors.
We view the choice of the right loss function as application-specific, and demonstrate results that hold for arbitrary $\ell$.

Each algorithm $a$ implies a pair of group errors $(e_r(a),e_b(a))$, and we denote this vector by $e(a)$. Throughout this paper,  we use an \emph{improvement in accuracy} to mean a reduction in both group errors, and an \emph{improvement in fairness} to mean a reduction in the absolute difference between the group errors.\footnote{This formulation is consistent with much of the literature on algorithmic fairness, but does not take into account all important fairness considerations.  For example, perfect prediction of criminal offense ($Y$) by the algorithm for both groups does not address historical inequities that have shaped differential base rates of $Y$ across groups. Moreover, as \citet{KasyAbebe} point out, an algorithm that is fair in the narrow context of one decision may perpetuate or exacerbate inequalities within a larger context. We leave to future work the interesting question of how these algorithmic design decisions might impact decisions in a larger dynamic game.} This approach nests many of the fairness criteria that have been proposed in the literature by adopting a particular choice of a loss function (see \citet{AlgFairSurvey} for a recent survey of these criteria). For example, if the type $Y$ is binary and $\ell(d,y)= \mathbbm{1}(d\neq y)$, then $e_r(g)=e_b(g)$ corresponds to equality of misclassification rates. See Appendix \ref{sec:FairCriteria} for further details and examples.\footnote{For instance, we can accommodate equality of false positive/negative rates  \citep{KMR,chouldechova} by letting the loss function depend on group identity.}

Section \ref{sec:Extensions} discusses several extensions of this framework, including generalizing to fairness criteria of the form $\vert \phi(e_r) - \phi(e_b) \vert$ (which includes the ratio of error rates as a special case), and allowing for fairness and accuracy to be defined using different loss functions. 

\subsection{Fairness-Accuracy Preferences}  
The designer has a preference ordering over group error pairs $e=(e_r,e_b) \in \mathbb{R}^2$. First, recall the usual Pareto dominance order.

\begin{definition} The \emph{Pareto dominance (PD)} relation $>_{PD}$ is the strict order on $\mathbb{R}^2$ where $e >_{PD} e'$  if  $e_r\leq e_r'$ and $e_b\leq e_b'$ with at least one inequality strict.
\end{definition}

We now modify this order to include fairness considerations.

\begin{definition} \label{def:Pareto} The \emph{fairness-accuracy (FA) dominance} relation $>_{FA}$ is the strict order on $\mathbb{R}^2$ where $e >_{FA} e'$  if $e_r \leq e_r'$, $e_b\leq e_b'$ and $\vert e_r - e_b \vert \leq \vert e_r'-e_b'\vert$ with at least one inequality strict.\footnote{ \citet{KleinbergMullainathan} define an admissions rule to be a strict improvement over another if it improves both efficiency (the average type of an admitted applicant) and equity (the fraction of admitted students who belong to the disadvantaged group), which is  similar to our FA dominance relation but non-nested, as it involves two loss functions. The  FA-dominance relation in Online Appendix \ref{app:DifferentLoss} generalizes both orders.}
\end{definition}

That is, if it is possible to simultaneously increase accuracy (reducing errors for both groups) and also increase fairness (reducing the gap between these errors), then all designers must prefer this. We call any preference consistent with this order a fairness-accuracy preference.

\begin{definition} \label{def:FApref}
A \emph{fairness-accuracy (FA) preference} $\succeq$ is any total order on $\mathbb{R}^2$ such that $e \succ e'$ whenever $e>_{FA}e'$.\footnote{It is straightforward to see that these orders are unchanged if our measure of fairness, $\vert e_r - e_b \vert$, is replaced with $\phi(\vert e_r - e_b \vert)$ where $\phi$ is a strictly increasing function. } 
\end{definition}

The class of FA preferences reflects a broad range of  views on how to trade off fairness and accuracy, including the following special cases that have been proposed in the literature.
 
\begin{example} [Utilitarian] \label{ex:SWA}  The designer evaluates errors $e=\left(e_{r},e_{b}\right)$
according to the weighted sum in the population. That is, let
$
w_u(e)=-p_{r}e_{r}-p_{b}e_{b}
$,
and let $\succeq_u$ be the ordering represented by $w_u$, so that $e \succeq_u e'$ if and only if $w_u(e)\geq w_u(e')$. (Note that the minority population, which has a lower weight by definition, will be naturally discounted as a group in this evaluation.) A designer with preferences $\succeq_u$ is called \emph{Utilitarian} \citep{Harsanyi1953,Harsanyi1955}. There is also a generalization of the Utilitarian rule which evaluates errors $e$ using $w_{\gamma}(e) = -\gamma_r e_r - \gamma_b e_b$ for arbitrary positive constants $\gamma_r,\gamma_b$ \citep*{CharnessRabin02,SaezStantcheva16,DworczakKominersAkbarpour21,RambachanEtAl}.

 \end{example}

\begin{example} [Rawlsian] The designer evaluates errors $e=\left(e_{r},e_{b}\right)$
according to the greater error. That is, let 
$
w_r(e) = -\max\left\{ e_{r},e_{b}\right\} 
$,
and consider a preference order represented by $w_r$. A designer with such preferences is called \emph{Rawlsian}  \citep{Rawls}.\footnote{This approach is also known as \emph{group distributionally robust optimization} \citep{GroupDRO,Hansen2022TheIO}.}
\end{example}

\begin{example}
[Egalitarian] The designer evaluates errors $e=\left(e_{r},e_{b}\right)$
according to their difference. That is, let
$
w_e(e) = -\left|e_{r}-e_{b}\right|
$,
and consider the lexicographic preference order that first evaluates errors according to $w_e$ and then breaks ties using the Utilitarian utility $w_u$. A designer with such preferences is called \emph{Egalitarian}.
\end{example}

\begin{example}[Constrained Optimization] 

The designer evaluates errors $e=\left(e_{r},e_{b}\right)$ by first determining if they satisfy the fairness constraint $\left|e_{r}-e_{b}\right| \leq c$. If so, then the designer uses the Utilitarian criterion. If instead $\left|e_{r}-e_{b}\right| > c$, then the designer uses the Egalitarian criterion.\footnote{This is a common approach in the algorithmic fairness literature \citep{Ferryetal,MenonWilliamson,CorbettDaviesetal,Agarwaletal,HardtPriceSrebro}.}

\end{example}

\begin{example} [Accuracy then Fairness]  The designer evaluates errors $e=\left(e_{r},e_{b}\right)$ by first evaluating accuracy and then fairness. That is, $e\succ e'$ if  $e>_{PD}e'$, and if not, they are then compared using $w_e$. This is the approach recently proposed by \citet{VivianoBradic}.
\end{example}

\noindent   Our consideration of this wide class of preferences is motivated in part by the experimental literature on social preferences, which documents substantial heterogeneity across individuals' equity-efficiency preferences. For example, when given the choice between different allocations of payoffs across individuals, some experimental subjects choose Pareto-dominated allocations that are more equal (corresponding in our setting to choice of $e=(e_r,e_b)$ over $e'=(e_r',e_b')$ where $e'>_{PD} e$ but $\vert e_r - e_b \vert < \vert e_r'-e_b'\vert$). These are minority preferences in the population \citep{AndreoniMiller,CharnessRabin02}, but constitute  31\% of subjects in an experiment in \citet{FismanKarivMarkovits}. We do not take a normative stance on which FA preferences are more appropriate, instead viewing the class of FA preferences as encompassing a broad range of designer preferences that may be relevant in practice.

\subsection{The Fairness-Accuracy Frontier}

Fixing any covariate vector $X$, we define the feasible set of group error pairs to be those pairs that can be implemented by some algorithm that takes $X$ as input.

\begin{definition}
The \emph{feasible set} given covariate vector $X$ is 
\[
\mathcal{E}_X := \{e(a) : a \in \Delta\mathcal{A}\}.
\]
\end{definition}
We first define the standard Pareto frontier with respect to accuracy for the two groups.

\begin{definition} \label{def:Pareto}
The \emph{Pareto frontier} given $X$ is 
\[
\mathcal{P}_X := \{e\in \mathcal{E}_X : \textrm{there is no } e'\in \mathcal{E}_X \textrm{ such that } e' >_{PD} e \}.
\]
\end{definition}

Our  fairness-accuracy frontier is the set of all group error pairs that are FA-undominated in the feasible set.

\begin{definition} \label{def:Frontier}
The \emph{fairness-accuracy (FA) frontier} given $X$ is 
\[
\mathcal{F}_X := \{e\in \mathcal{E}_X : \textrm{there is no } e'\in \mathcal{E}_X \textrm{ such that } e' >_{FA} e \}.
\]

\end{definition}

The FA frontier can also be interpreted as the set of group error pairs that are optimal for some FA preference.  Moreover, every point in the FA frontier is uniquely optimal for some FA preference, so we cannot exclude any points from the set without hurting some designer (see Appendix \ref{app:Welfare} for details).

\section{The Fairness-Accuracy Frontier} \label{sec:FAFrontier}

This section presents three complementary characterizations of the fairness-accuracy frontier. Section \ref{sec:DefineGroupBalance} defines the properties of \emph{group-balance} and \emph{group-skew} that will play a key role in many of our results. Section \ref{sec:FullDesign} presents our main characterization of the FA frontier, which says that whether the covariate vector is group-balanced or group-skewed determines what kinds of fairness-accuracy tradeoffs are relevant. Section \ref{sec:altchar} presents two additional characterizations of the FA frontier: The first shows that the fairness-accuracy frontier can be completely characterized by considering a class of ``simple'' utility functions, and the second describes a class of algorithms that implements the fairness-accuracy frontier.

\subsection{Key Property: Group-Balance vs Group-Skew} \label{sec:DefineGroupBalance}
For all covariate vectors $X$, the feasible set $\mathcal{E}_X$ is closed and convex (Lemma \ref{lemm:ConvexPolygon}). Two special feasible points are the following.

\begin{definition}[Group Optimal Points]
For any covariate vector $X$, define 
\[
\displaystyle r_X := \argmin_{e \in \mathcal{E}_X} e_r \textrm{\quad\quad\quad\quad}
\displaystyle b_X := \argmin_{e \in \mathcal{E}_X} e_b
\]
to be the feasible points that minimize group $r$'s error and  group $b$'s error respectively, where ties are broken by minimizing the other group's error. 
\end{definition}

Group optimal points characterize the Pareto frontier (Definition \ref{def:Pareto}).

\begin{observation} The Pareto frontier $\mathcal{P}_X$  is the lower boundary of $\mathcal{E}_X$ between $r_{X}$ and $b_{X}$.\footnote{We use \emph{lower boundary between two points} to mean the part of the boundary of the set that lies between the two points and below the line segment connecting the two.} 
\end{observation}

We similarly define the fairness optimal point.

\begin{definition}[Fairness Optimal Point] For any covariate vector $X$, define
\[
\displaystyle f_{X} := \argmin_{e \in \mathcal{E}_X} \vert e_r - e_b \vert
\]
to be the point that minimizes the absolute difference between group errors, where ties are broken by minimizing either group's error.\footnote{This point is the same regardless of which group is used to break ties.} 
\end{definition}

While $r_X$ and $b_X$ respectively denote the points that minimize group $r$ and $b$'s errors, the group whose error is minimized need not have the lower error. 
Whether this is the case turns out to be the key property determining the nature of the fairness-accuracy tradeoff given $X$.

\begin{definition} Covariate vector $X$ is:
\begin{itemize}
    \item \emph{$r$-skewed} if $e_r< e_b$ at $r_X$ and $e_r\leq e_b$ at $b_X$
    \item \emph{$b$-skewed} if $e_b< e_r$ at $b_X$ and $e_b\leq e_r$ at $r_X$
    \item \emph{group-balanced} otherwise
\end{itemize}
If $X$ is $g$-skewed for either group $g$, then we say it is \emph{group-skewed}.\footnote{Group-balance/group-skew additionally depends on the joint distribution $\left(X,Y,G\right)$ and the loss function. Since we are primarily interested in studying how tradeoffs are affected by changes in the covariate vector $X$ (e.g. adding/banning a covariate), we treat the other random variables (and loss function) as fixed primitives of our model.}

\end{definition}

In words, $X$ is $r$-skewed if group $r$'s error is smaller than group $b$'s error not only at the $r$-optimal point $r_X$, but also at the $b$-optimal point $b_X$. Geometrically, this means that $r_X$ and $b_X$ fall to the same side of the 45 degree line. (The feasible set may still intersect the 45 degree line, as in Figure \ref{fig:ParetoSetGX}.)   In contrast, the covariate vector $X$ is group-balanced if at each group's optimal point, its error is lower than that of the other group, implying that $r_X$ and $b_X$ fall to opposite sides of the 45 degree line. 

Loosely speaking, a covariate vector is group-balanced if it is possible to separate accurate predictions for one group from accurate predictions for another, and it is group-skewed if this is not feasible. We provide several examples below in which the covariate vector is either group-balanced or group-skewed (all supporting arguments are contained in Appendix \ref{app:SupportExamples}).

\begin{example}[Strong Independence] \label{ex:SI} Suppose $G \indep (X,Y)$. Then $X$ is group-balanced.
\end{example}

\begin{example}[Unequal Means]
\label{ex:UnequalMeans}
Suppose there are two functions $a_r, a_b \in \mathcal{A}$ such that  $Y=a_g(X) + \varepsilon$ for members of group $g$ where $\varepsilon$ is a mean-zero random variable independent of $X$. If the loss function is $l(d,y)=(d-y)^2$, then $X$ is group-balanced. 
\end{example}

\begin{example}[Asymmetric Predictive Power]
\label{ex:Asymmetric}
Suppose there is a function $a_0 \in \mathcal{A}$ such that  $Y=a_0(X) + \varepsilon_g$ for members of group $g$, where $\varepsilon_r$ and $\varepsilon_b$ are mean-zero random variables independent of $X$. If the loss function is $l(d,y)=(d-y)^2$ and $\varepsilon_b$ has larger variance than $\varepsilon_r$, then $X$ is $r$-skewed. 
\end{example}

\begin{example}[Conditional Independence] \label{ex:CI} Suppose $\mathcal{A} = \overline{\mathcal{A}}$ (all algorithms are feasible) and  $G\indep Y\mid X$,\footnote{Equivalently, $Y=a_0(X) + \varepsilon_X$ for both groups, where the noise term is possibly dependent on $X$.} i.e., once $X$ is observed there is no additional predictive value to knowing a subject's group identity.\footnote{This kind of conditional independence appears for example when the coefficient on group identity is zero in a regression of $Y$ on observables, e.g. \citet{LudwigMullainathan} find that race ($G$) is not predictive of a criminal's risk ($Y$) conditional on arrest $(X)$ in their data.} Then either $r_X=b_X=f_X$ or $X$ is group-skewed.
 
\end{example}

\subsection{Geometric Characterization of the FA Frontier} \label{sec:FullDesign}

We now provide our main characterization of the FA frontier. Theorem \ref{thm:FullDesignPareto} shows that 
whether the covariate vector $X$ is group-balanced or group-skewed determines the shape of the frontier.  

\smallskip

\begin{theorem} \label{thm:FullDesignPareto} \phantom{m}
\begin{itemize}
\item[(a)] If $X$ is group-balanced,  then $\mathcal{F}_X=\mathcal{P}_X$.
\item[(b)] If $X$ is $g$-skewed, then $\mathcal{F}_X$ is the boundary of $\mathcal{E}_X$ between $g_{X}$ and $f_{X}$ containing $\mathcal{P}_X$.
\end{itemize}
\end{theorem}

\begin{figure}[h]
\begin{center}
\includegraphics[scale=0.6]{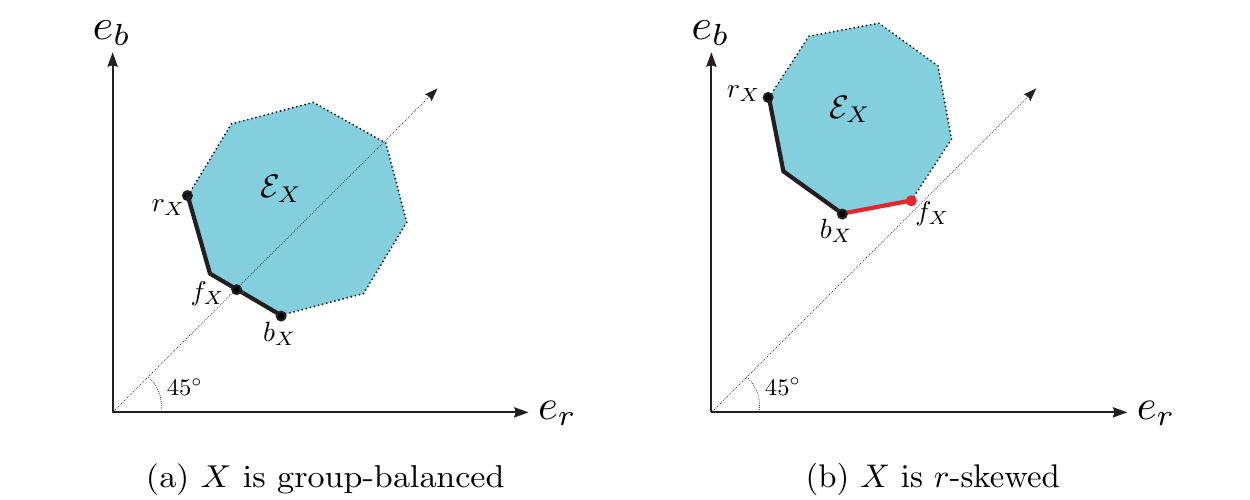}
\end{center}
\caption{\footnotesize{Example feasible set and fairness-accuracy frontier for (a) a group-balanced covariate vector $X$ and (b) an $r$-skewed covariate vector $X$.}}
\label{fig:general}
\end{figure}

These two cases are depicted in Figure \ref{fig:general}. When $X$ is group-balanced and $r_X$ and $b_X$ are distinct, the two points fall on opposite sides of the 45-degree line (Panel (a)), and the fairness-accuracy frontier is the standard Pareto frontier, i.e. the lower boundary of the feasible set connecting these two points. Here, the only tradeoffs are between the accuracy of one group and the accuracy of the other group. When $X$ is $r$-skewed (Panel (b)), then both $r_X$ and $b_X$ fall on the same side of the 45-degree line, and the fairness-accuracy frontier consists not only of the usual Pareto frontier connecting $r_X$ to $b_X$, but additionally a positively sloped line segment connecting the Pareto frontier to $f_{X}$ (depicted in Figure \ref{fig:general} in red). This positively sloped line segment is novel; here, tradeoffs are between fairness and the accuracy of \emph{both} groups. This is a case where there is a strong conflict between fairness and accuracy. We summarize these observations  in the following corollary.

\begin{corollary} \label{corr:Tradeoff} Suppose $f_{X}$ is distinct from $r_X$ and $b_X$. Then $X$ is group-skewed if and only if there are points  $e,e'\in\mathcal{F}_X$ such that $e>_{PD}  e'$ but $\vert e_r - e_b \vert > \vert e_r'-e_b'\vert$.
\end{corollary}

The corollary implies that if the covariate vector is group-skewed (and the fairness optimal point is distinct from the group optimal points), then the FA frontier must consist of a positively-sloped segment along which every pair of points can be Pareto-ranked. In practice, one way in which designers may end up at these Pareto-dominated outcomes is by choosing to ignore certain available covariates. For example, some medical practitioners recommend removing race-based covariates from healthcare prediction algorithms, even as they acknowledge the predictive value of those covariates for both groups \citep{Vyasetal,Cerdena2020,Delgado2021}. Similarly, some university admissions committees have elected to exclude consideration of test scores from admissions decisions.\footnote{Test scores are predictive of college grades for all of the relevant demographic groups (see Section A.5 of \citet{BerkeleySAT}), but are more predictive for applicants in some groups than others \citep{Rothstein}. In Section \ref{sec:GKnown} we return to this application, interpreting the exclusion of test scores slightly differently---not as a choice made by the agent setting the algorithm, but as an informational regulation imposed by a designer whose preferences are different from those of the agent.} Corollary \ref{corr:Tradeoff} says that if the covariate vectors in these settings are group-skewed, then disagreements over whether to implement Pareto-dominated errors can be explained by different fairness preferences. On the other hand, if inputs are group-balanced, then a policy proposal that implements Pareto-dominated errors cannot be rationalized by any fairness-accuracy preference, no matter how strong the designer's weight on fairness. 

\subsection{Alternate Characterizations of the FA Frontier} \label{sec:altchar}

We now provide two additional characterizations of the FA frontier. Proposition \ref{prop:MinimalClass} shows that the FA frontier can be fully recovered from the smaller class of parameterized utility functions that linearly trade off accuracy and fairness. Specifically consider the class of functions
\[
w\left(e\right)=-\gamma_{r}e_{r}-\gamma_{b}e_{b}-\gamma_{f}\left|e_{r}-e_{b}\right|
\]
where $\gamma_{r},\gamma_{b}>0$ and $\gamma_{f}\geq 0$ are (respectively) the designer's weights on accuracy for each group and on fairness.
A preference $\succeq$ is \emph{simple} if it can be represented by a utility of this form, i.e. $e\succeq e^{\prime}$ if and only if $w(e)\geq w(e^{\prime})$. All simple preferences are FA preferences, but not all FA preferences are simple. For example, both Utilitarian and Rawlsian preferences are simple but Egalitarian preferences are not.\footnote{To see this for the Utilitarian designer, set $\gamma_r=p_r$, $\gamma_b=p_b$ and $\gamma_f=0$. To see this for the Rawlsian designer, set $\gamma_r=\gamma_b=\gamma_f=1$. Egalitarian preferences are not simple as they are not continuous.} Given any preference $\succeq$ on error pairs, let
\[
\mathcal{C}_{X}(\succeq):=\left\{ e\in\mathcal{E}_X\text{ : }e\succeq e^{\prime}\text{ for all }e^{\prime}\in\mathcal{E}_X\right\} 
\]
denote the set of $\succeq$-optimal points.
 
\begin{proposition} \label{prop:MinimalClass}
$\mathcal{F}_X$ is the union of $\mathcal{C}_{X}(\succeq)$ over all simple preferences
$\succeq$.
\end{proposition}

This result is in the spirit of \citet{WaldWolfowitz1951}'s complete class theorem, and provides a class of utility functions which is sufficient for generating the FA frontier. Recalling that the usual Pareto frontier can be characterized as the set of optimal points for designers with utility functions of the form $\gamma_r e_r + \gamma_b e_b$, our characterization differs in additionally including the term $\gamma_f \vert e_r-e_b \vert$, which captures a concern for inequity across groups. As it turns out, the addition of this term is sufficient to characterize the entire FA frontier, so the frontier can be traced out simply by shifting the parameter weights $\gamma_r$, $\gamma_b$ and $\gamma_f$.

Our third and final characterization tells us how the design of the optimal algorithm varies along the frontier when the set of algorithms is unconstrained. For each realization $x$ of the covariate vector, let
\begin{align*}
\Delta_g^x & := P(X=x \mid G=g) \cdot  \mathbb{E}[\ell(1,Y) - \ell(0,Y)\mid X=x,G=g]
\end{align*}
denote the change in group $g$'s expected error when the decision at $x$ is moved from $a(x)=0$ to $a(x)=1$. The \emph{$g$-optimal decision} at $x$ is $d=1$ if $\Delta_g^x \geq 0$, and otherwise it is $d=0$.\footnote{It is immaterial how we break the tie here.} Define
\[h(x) = \frac{\Delta^x_b}{\Delta^x_r}\]
to be the ratio of the relative impact of   the decision at $x$ on group $r$'s error to group $b$'s error. Then order the realizations $x_1, \dots, x_n$ of $X$ so that $h(x_1)\leq h(x_2) \leq \dots \leq h(x_n)$, with ties broken arbitrarily. Let $a^*_r$ denote the $r$-optimal algorithm, which assigns the optimal decision for group $r$ to each $x$, and consider the following class of algorithms
\[a_i(x_j):=\left\{\begin{array}{cc}
a^*_r(x_j) & \textrm{if } j \geq i \\
1-a^*_{r}(x_j) & \textrm{if } j < i
\end{array}\right.\]
In words, $a_0$ is the algorithm that (deterministically) assigns to every covariate vector the decision that is best for group $r$; that is $a_0=a_r^*$, the $r$-optimal algorithm. Moving along the sequence $a_1,a_2, \dots$ entails successively switching the assignment at the next realization in the order $x_1,x_2,\dots$.
The following proposition says that the FA frontier is exactly the set of outcomes generated by algorithms of this form.

\begin{proposition} \label{prop:AlgorithmsGB} Suppose $\mathcal{A}=\overline{\mathcal{A}}$, and let $n_c:=\max_i \{h(x_i) < c\}$ for $c\in \{0,1\}$. Then $\mathcal{F}_X$ is the set of errors generated by the algorithms
$\beta a_{i-1} + (1-\beta) a_{i}$ 
for all $(\beta,i)\in \Psi$ where
\begin{itemize}
\item[(a)] if $X$ is group-balanced, then $\Psi = [0,1] \times \{1, \dots, n_0\}$, and 
\item[(b)] if $X$ is $r$-skewed, then  $\Psi = \left([0,1]\times \{1, \dots, n^*-1\}\right) \cup \left([0,\overline{\beta}]\times\{n^*\}\right)$ for some $\overline{\beta} \in [0,1]$ and $n^* \in \{n_0,\dots  n_1\}$. Moreover, if $e_r < e_b$ at $f_X$, then $(\overline{\beta},n^*)=(1,n_1)$. 
\end{itemize}

\end{proposition}

\begin{figure}
    \centering
      \includegraphics{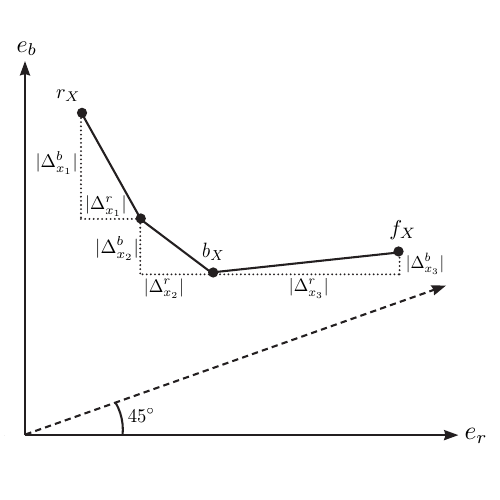}
    \caption{\footnotesize{In this example, the point $r_X$ is implemented by the algorithm $a_0=a_r^*$. The next extreme point is implemented by $a_1$, which satisfies $a_1(x_1)=1-a_r^*(x_1)$, and $a_1(x_i)=a_r^*(x_i)$ for all $i>1$. The point $b_X$ is implemented by the algorithm $a_2$, which satisfies $a_2(x_i)=1-a_r^*(x_i)$ for $i=1,2$, and $a_2(x_i)=a_r^*(x_i)$ for all $i>2$. Finally, $f_X$ is implemented by the algorithm $a_3$, which satisfies $a_3(x_i)=1-a_r^*(x_i)$ for $i=1,2,3$, and $a_3(x_i)=a_r^*(x_i)$ for all $i>3$. In this example, $n_0 = 2$ (i.e., groups $r$ and $b$ have different optimal decisions for covariate vectors $x_1,x_2$), while $n^*=n_1 = 3$.}}
    \label{fig:Prop2}
\end{figure}

The proposition is illustrated in Figure \ref{fig:Prop2}. The two groups disagree over the optimal decision at the realizations $x_1, \dots, x_{n_0}$, since one group prefers $d=1$ while the other group prefers $d=0$. (In the figure, $n_0=2$.) Here, the tradeoff is between the accuracy of one group and that of the other, and these realizations are ordered so that the relative error impact on group $b$ compared to group $r$ is maximized at $x_1$, and decreases as we progress along the sequence. When the assignments for all of these disagreement covariate vectors have been swapped from the $r$-optimal assignment to the $b$-optimal assignment, we obtain an algorithm that implements the $b$-optimal point $b_X$. Part (a) of Proposition \ref{prop:AlgorithmsGB} says that when $X$ is group-balanced, the algorithms $(a_i)_{0 \leq i\leq n_0}$ implement the extreme points of the FA frontier (which is also the Pareto frontier by Theorem \ref{thm:FullDesignPareto}), and mixing between neighboring algorithms of this form yields the entire FA frontier.

When $X$ is group-skewed, the FA frontier continues further by swapping the assignment at covariate vectors $x_{n_0+1}, x_{n_0+2}, \dots$. The two groups agree on the preferred decision at these covariate vectors, so swapping the assignment hurts both groups. Here, the tradeoff is between the accuracy of both groups and fairness; the realizations are ordered so that the relative payoff impact on group $r$  to group $b$ is largest at $x_{n_0+1}$, and decreases as we progress along the sequence. These swaps continue until we reach the fairness-optimal point $f_X$, which is either the extreme point implemented by the algorithm $a_{n^*}$ as in Figure \ref{fig:Prop2} (in which case $\overline{\beta}=1$), or an interior point on the line segment ending on this extreme point (in which case $\overline{\beta}<1$). 

\subsection{Special Cases} \label{sec:SpecialCases}
In some applications, the designer may be able to explicitly  condition the algorithm on group identity.\footnote{Use of such algorithms is illegal in certain contexts. For example, the Equal Opportunity Act prohibits (among other things) lending decisions that condition explicitly on race or gender.} That is, the designer's choice set may be all pairs of algorithms $(a_r,a_b)\in \mathcal{A} = \mathcal{A}_r \times \mathcal{A}_b$, where the algorithm $a_g \in \mathcal{A}_g$ is applied for members of group $g$. (As before, we permit randomizations over pairs of this form.) When there are no constraints on the algorithm, i.e. $\mathcal{A}=\overline{\mathcal{A}}$, then it is equivalent to assume that group identity is included as a covariate in $X$. In this case, the feasible set and fairness-accuracy frontier simplify as follows.

\begin{figure}[h]
        \centering
\includegraphics[scale=0.75]{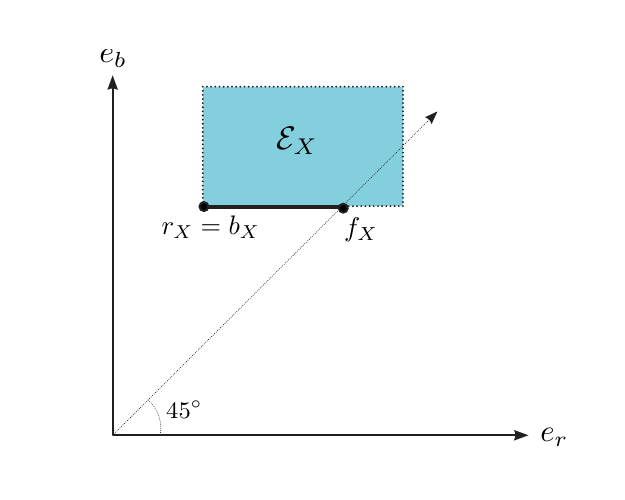}
        \caption{Group-Specific Algorithms}
        \label{fig:ParetoSetGX}
\end{figure}

\begin{proposition}\label{prop:XRevealsG}
Suppose the algorithm can be conditioned on group identity. Then $\mathcal{E}_X$ is a rectangle whose sides are parallel to the axes, and $\mathcal{F}_X$ is the line from $r_X=b_X$ to $f_{X}$.
\end{proposition}

Figure \ref{fig:ParetoSetGX} depicts such a feasible set and fairness-accuracy frontier. 
One endpoint, the Utilitarian-optimal point labeled $r_X = b_X$, gives both groups their minimal feasible error. The other endpoint, the Egalitarian-optimal point $f_{X}$, maximizes fairness. Everywhere along the fairness-accuracy frontier $\mathcal{F}_X$, the worse-off (higher error) group receives its minimal feasible error. Thus when group-specific algorithms are permitted, the fairness-accuracy tradeoff simply becomes a tradeoff between fairness and the welfare of the advantaged group (i.e., the group that the covariate vector is skewed towards). It is also straightforward to see from this result that the error for the disadvantaged group under $X$ must decrease regardless of which FA preference the designer holds.

When the joint distribution $\mathbb{P}$ relating $(X,Y,G)$ satisfies certain independence properties, the FA frontier again simplifies. Specifically, under the assumptions of our previous Examples \ref{ex:SI} and \ref{ex:CI} we obtain the following result.

\begin{proposition}\label{prop:SpecialCase}
\begin{itemize}
    \item[(a)] Suppose $G\indep(X,Y)$. Then $\mathcal{E}_X$ is a line segment on the 45-degree line, and $\mathcal{F}_X$ is a single point.
    \item[(b)] Suppose $G\indep Y \mid X$ and $\mathcal{A}=\overline{\mathcal{A}}$. Then
$\mathcal{F}_X$ is that part of the lower boundary of the feasible set $\mathcal{E}_X$ from the point $b_{X}=r_{X}$ to the point $f_X$.
\end{itemize}
\end{proposition}

These cases are depicted in Figure \ref{fig:SpecialCases}. In Panel (a), the FA frontier is a singleton, and fairness-accuracy preferences
are irrelevant: All designers who agree on the basic
FA-dominance principle outlined in Definition \ref{def:Pareto}
prefer the same algorithm. 
In Panel (b), the left point is the
(shared) group optimal point $r_{X}=b_{X}$, and the right endpoint
is the fairness optimal point $f_{X}$. From $r_{X}=b_{X}$ to $f_{X}$,
the fairness-accuracy frontier consists entirely of positively sloped line segments.
Thus, everywhere along the frontier, the two groups' errors move in
the same direction, implying that the only way to improve fairness
is to decrease accuracy uniformly across groups. The only relevant
difference across designers, then, is how they choose to resolve
strong fairness-accuracy conflicts of this form. As we will see in Section \ref{sec:Empirical}, although these conditions are strong and unlikely to be exactly satisfied, the fairness-accuracy tradeoffs they describe are relevant to real data.

\begin{figure}[h]
    \centering
    \begin{subfigure}[t]{0.45\textwidth}
        \centering
        \includegraphics[width=\textwidth]{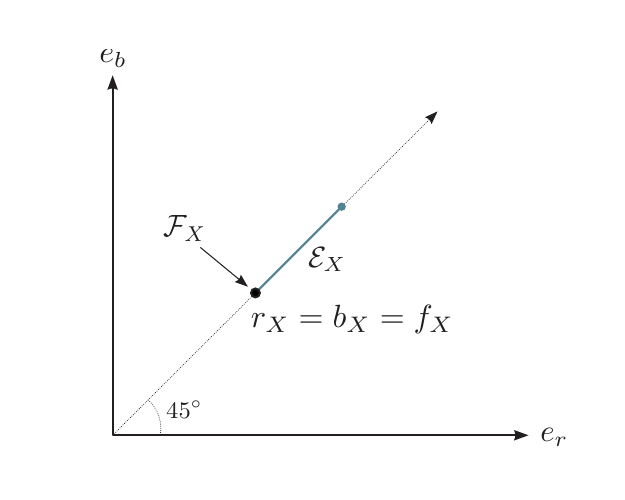}
        \caption{Strong Independence: $G \indep (X,Y)$}
        \label{fig:1a}
    \end{subfigure}
    \hfill
    \begin{subfigure}[t]{0.45\textwidth}
        \centering
        \includegraphics[width=\textwidth]{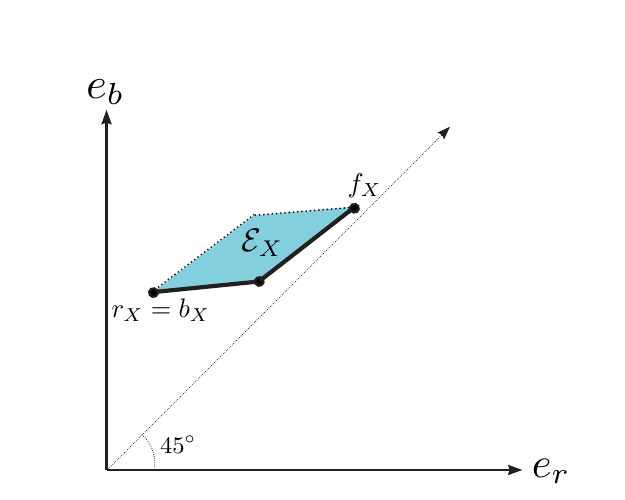}
        \caption{Conditional Independence: $G \indep Y\mid X$}
        \label{fig:1b}
    \end{subfigure}

    \caption{Special Cases}
    \label{fig:SpecialCases}
\end{figure}

\section{Input Design} \label{sec:DesignInputs}

We have so far assumed that the designer directly chooses the best algorithm to maximize a preference that (weakly) responds to both fairness and accuracy. This is a good description of some settings; for example, a company may internalize both fairness and accuracy concerns in its hiring algorithm. But often the algorithm is set by an agent who does not care about fairness across groups, while the inputs used by the algorithm are constrained by a designer who does. For example, a healthcare provider (agent) determining treatment may  seek to maximize the number of correct diagnoses, while a policymaker (designer) may additionally prefer that the accuracy of the provider's treatments be similar across  social groups.  Or, a bank (agent) may seek to maximize profit from loan issuance, while a regulator (designer) may require that the rate at which individuals are incorrectly denied loans does not differ too sharply across groups. In these settings, the designer can often influence the algorithm indirectly by implementing policies that constrain the algorithm's inputs. For example, \citet{ChanEyster} report that as part of an effort to influence Berkeley law school's admissions policy in 1997, UC Berkeley administrators coarsened candidates' LSAT scores into intervals and reported this coarsened variable to the law school admissions committee.

In Section \ref{sec:BayesDesign}, we model this interaction as an information design problem in which the designer constrains the inputs of the algorithm, while the algorithm is chosen by an accuracy-minded agent. In Section \ref{sec:CompareInputDesign}, we provide conditions under which the designer is nevertheless able to implement his favorite point on the fairness-accuracy frontier. In Section \ref{sec:AddCovariate}, we study   whether the designer's optimal regulation of inputs may involve completely banning a covariate such as group identity or a test score.

\subsection{Input Design Model} \label{sec:BayesDesign}

A designer chooses a \emph{garbling} of the covariate vector $X$, which is represented as a mapping $T: \mathcal{X}\rightarrow \Delta\mathcal{T}$ taking realizations of $X$ into distributions over the possible realizations of $T$ (assumed without loss to be finite).\footnote{An interesting direction for future research would be to impose ``procedural fairness'' requirements on the garbling (such as requiring the garbling to be deterministic, or to be monotone in certain covariates), and ask what is achievable under those constraints. In the present paper we maintain throughout full flexibility to garble covariates.} Examples of garblings include the following.

\begin{example}[Banning an Input] \label{ex:DropInput} Suppose $X=(X_1,X_2,X_3)$ and the designer wants to ban the last coordinate $X_3$. The corresponding garbling is $T(x_1,x_2,x_3)=(x_1,x_2)$ with probability 1.
\end{example}

\begin{example}[Coarsening the Input] \label{ex:Coarsen} The set of realizations $\mathcal{X}=\{1,2,3,4\}$ is partitioned into $\{\{1,2\},\{3,4\}\}$, and $T(x)$ reports (with probability 1) the partition element to which $x$ belongs. \end{example}

\begin{example}[Adding Noise] \label{ex:Noise} $T(x)=x+\varepsilon$ where the noise term $\varepsilon$ takes value $+1$ or $-1$ with equal probability.
\end{example}

We view these garblings as information policies that the designer can possibly commit to by law. Real examples of garblings are abundant: The ``ban-the-box" campaign \citep{AganStarr} restricted employers from using  criminal history as an input into hiring decisions (similar to Example \ref{ex:DropInput}); the College Board coarsens a test-taker's answers into an integer-valued score between 400 and 1600 (similar to Example \ref{ex:Coarsen}); and organizations such as the US Census Bureau, Apple, and Google add noise to users' inputs under differential privacy initiatives (similar to Example \ref{ex:Noise}).\footnote{See \citet{GarfinkelAbowdPowazek} for an example reference.}

The agent  chooses
an algorithm $a:\mathcal{T}\rightarrow\mathcal{D}$ that takes as input the garbled variable chosen by the designer. For simplicity in this section, we assume that the full set of algorithms $\overline{\mathcal{A}}$ is available; as before, we allow for randomizations over algorithms. The agent evaluates errors according to
\begin{equation} \label{eq:Utility}
w(e)=-\alpha_{r} e_{r}\left(a\right)-\alpha_{b} e_{b}\left(a\right)
\end{equation}
for some constants $\alpha_{r},\alpha_{b} > 0$ that are known to the designer.\footnote{In Appendix \ref{app:ExtendBD}, we prove additional results for the case when a coefficient $\alpha_g$
is negative so the agent is adversarial or biased against group $g$ and prefers to \emph{increase} error for that group. This could reflect taste-based discrimination by the agent. Note that this falls outside of our class of FA preferences.}$^,$\footnote{We view the typical setting as one in which the policymaker has fairness concerns that the agent does not share, but the reverse case (in which the agent has fairness concerns that the policymaker does not share) is also interesting. See Section \ref{sec:Extensions} for a brief discussion of some technical complications that arise in this case.} When $\alpha_{g}=p_{g}$, the agent is Utilitarian and exclusively cares about aggregate accuracy; otherwise, the agent's preference falls in the broader class of generalized Utilitarian preferences mentioned in Example \ref{ex:SWA}.\footnote{The agent's utility may involve weights different from Utilitarian weights if errors for the two groups are differentially costly for the agent. For example, suppose the agent is a bank manager and group $b$ is wealthier than group $r$. In this case, loans for group $b$ may be of higher value, so that incorrectly classifying creditworthy individuals in group $b$ is more costly. This corresponds to scaling the loss $\ell$ for group $b$ by $\alpha_{b}/p_b>1$.} We can rewrite (\ref{eq:Utility}) as
\begin{alignat*}{1}
w(e)= & -\sum_{g}\alpha_{g}\mathbb{E}\left[\ell\left(a\left(T\right),Y\right)\mid G=g\right]\\
= & -\sum_{t\in\mathcal{T}}p_{t}\sum_{y,g}\frac{\alpha_{g}}{p_{g}} \cdot \mathbb{P}\left(Y=y, G=g \mid T=t\right) \cdot \ell\left(a(t), y \right),
\end{alignat*}
where $p_{t}$ is the probability of $T=t$. Thus the agent's problem
of minimizing ex-ante error is equivalent to  the following
ex-post problem\footnote{When the agent's utility is non-linear in group errors, the ex-ante and ex-post problems are not equivalent in general.}
\begin{equation}
a\left(t\right)\in\argmin_{d\in\mathcal{D}}\sum_{y,g}\frac{\alpha_{g}}{p_{g}} \cdot \mathbb{P}\left(Y=y, G=g\mid T=t\right) \cdot \ell\left(d, y \right).\label{eq:choiceDM}
\end{equation}

\begin{definition} An error pair $e=(e_r,e_b)$ is \emph{implemented by $T$} if  there exists an algorithm $a_{T}$ satisfying (\ref{eq:choiceDM}) such that $e = e(a_T)$.
\end{definition}

Fixing any covariate vector $X$, we define the input-design feasible set to be all error pairs that can be implemented by some garbling $T$, and the input-design fairness-accuracy frontier to be the set of  group error pairs that are FA-undominated in the input-design feasible set.

\begin{definition}
The \emph{input-design feasible set} given covariate vector $X$ is 
\[
\mathcal{E}^{*}_X := \{e(a_T) : T \textrm{ is a garbling of } X\}.
\]
\end{definition}

\begin{definition} 
The \emph{input-design FA frontier} given $X$ is 
\[
\mathcal{F}^{*}_X := \{e\in \mathcal{E}^{*}_{X} : \textrm{there exists no } e'\in \mathcal{E}^{*}_{X} \textrm{ such that } e' >_{FA} e \}.
\]

\end{definition}

\subsection{Comparing Input Design and Algorithm Design} \label{sec:CompareInputDesign}

The following proposition says that under relatively weak conditions, it is without loss to have control only of the algorithm's inputs: Any error pair that a designer would choose to implement in the unconstrained problem  (i.e., given control of the algorithm) can also be achieved under input design. To state the result, we define
\begin{align}
    \label{eq:no_info_utilitarian}
    e_0 := \min_{d \in \mathcal{D}} \left(\alpha_r \cdot \mathbb{E}[\ell(d,Y) \mid G=r ] + \alpha_b \cdot \mathbb{E}[\ell(d,Y) \mid G=b ] \right)
\end{align}
to be the best payoff that the agent can achieve given no information, and 
\begin{equation} \label{eq:H}
H := \{(e_r,e_b) \, : \, \alpha_r e_{r}+ \alpha_b e_{b} \leq e_0 \}
\end{equation}
to be the halfspace including all error pairs that improve the agent's payoff relative to no information. 

\begin{proposition}[When Input Design is Without Loss] \label{thm:FullVersusBayes} The following hold:
\begin{itemize}
\item[(a)] Suppose $X$ is group-balanced. Then, $\mathcal{F}^{*}_{X} = \mathcal{F}_X$  if and only if  $r_X,b_X \in H$.

\item[(b)] Suppose $X$ is $g$-skewed. Then, $\mathcal{F}^{*}_{X} = \mathcal{F}_X$  if and only if $g_{X},f_{X} \in H$.
\end{itemize}
\end{proposition}

This result follows from the subsequent lemma, which says that the input-design feasible set is equal to the intersection of the unconstrained feasible set and $H$, with an analogous statement relating the fairness-accuracy frontiers. Related results appear in \cite{AlonsoCamara} and \citet{Ichihashi}, although we provide an independent argument in Appendix \ref{lemm:BayesDesign} for completeness.

\begin{lemma} \label{lemm:BayesDesign} For every covariate vector $X$, the input-design feasible set is $\mathcal{E}^{*}_{X} = \mathcal{E}_X \cap H$ and the input-design FA frontier is $\mathcal{F}^{*}_{X} = \mathcal{F}_X \cap H$.
\end{lemma}

Clearly the designer cannot hold the agent to a payoff lower than what the agent can guarantee with no information, so $\mathcal{E}^{*}_{X} \subseteq  \mathcal{E}_X \cap H$. In the other direction, we need to show that every point in $\mathcal{E}_X \cap H$ can be implemented by a garbling of $X$. The proof is by construction: If the designer garbles $X$ into recommendations of the decision, then the obedience constraints reduce precisely to the condition that the agent's payoff is improved relative to no information, i.e., the error pair belongs to $H$.  This yields the lemma, and Figure \ref{fig:feasibleSetBayes} illustrates how Proposition \ref{thm:FullVersusBayes} is implied by Lemma \ref{lemm:BayesDesign}.

These results tell us that input design is always sufficient to recover part of the original fairness-accuracy frontier. Moreover, so long as certain points ($r_X$ and $b_X$ in the case of a group-balanced $X$, $r_X$ and $f_{X}$ in the case of an $r$-skewed $X$, or $b_X$ and $f_{X}$ in the case of a $b$-skewed $X$) improve the agent's payoffs relative to no information, then the designer can induce the agent to choose the designer's most preferred outcome even without explicit control of the algorithm. Conversely, when these conditions do not hold, then input design is limiting for some designers.

\begin{figure}[h]
\begin{center}
\includegraphics[scale=0.6]{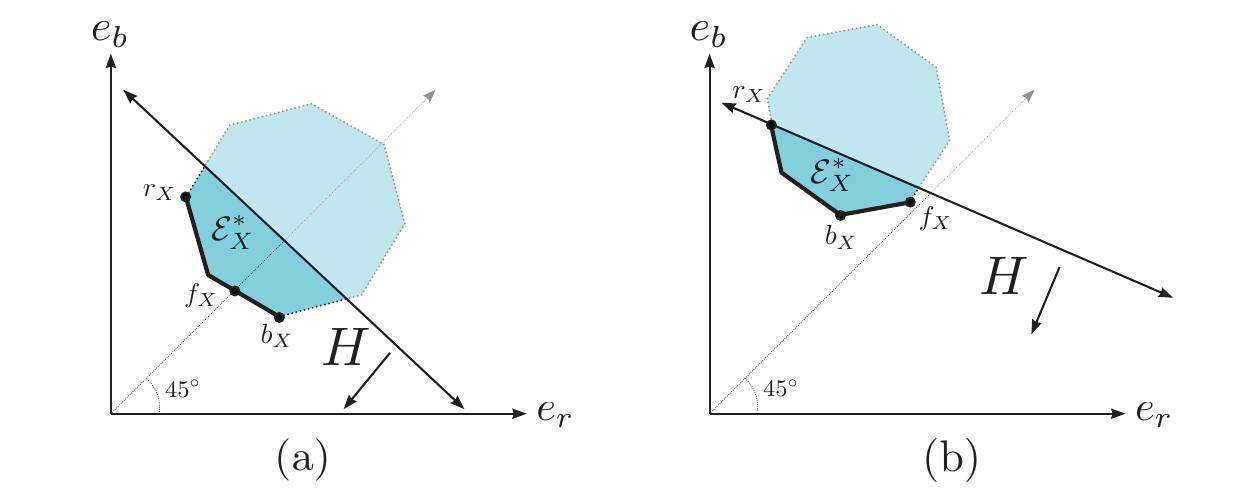}
\end{center}
\caption{\footnotesize{Depiction of an example input-design fairness-accuracy frontier for (a) a group-balanced covariate vector $X$ and (b) an $r$-skewed covariate vector $X$. In Panel (a), it is sufficient to check $r_X,b_X \in H$ to determine whether the entire unconstrained fairness-accuracy frontier belongs to $H$. This condition is satisfied in the figure, so every designer can implement his favorite unconstrained outcome using input design. In Panel (b), it is sufficient to check whether $r_X,f_{X} \in H$. This condition is failed in the figure, so some designer cannot implement his favorite unconstrained outcome using input design.}}
\label{fig:feasibleSetBayes}
\end{figure}

\subsection{Excluding a Covariate}\label{sec:AddCovariate}

We next turn to the question of whether the optimal garbling may involve a complete ban on the use of a specific covariate. For example, healthcare providers disagree over whether race should be a permitted input into clinical prediction algorithms \citep{Vyasetal,Manski2022,Manski}, and universities differ in whether they choose to  exclude consideration of standardized test scores.\footnote{See \url{https://www.nytimes.com/2021/05/15/us/SAT-scores-uc-university-of-california.html}.}

Since the designer and agent have (potentially) misaligned preferences, banning an input can be optimal, and we demonstrate this in a simple example in Section \ref{sec:SimpleExample}. But as we show in Sections \ref{sec:GNew} and \ref{sec:GKnown}, there are two important classes of inputs for which bans are strictly worse for all designers. We formalize ``strictly worse for all designers'' by comparing the fairness-accuracy frontier with and without the input under consideration.

 \begin{definition} For any pair of non-empty sets $S, S' \subset \mathbb{R}^2$,  write $S >_{FA} S'$ if every $e' \in S'$ is FA-dominated by some $e \in S$.
 \end{definition}
 
 \noindent When $\mathcal{F}^*_{X,X'} >_{FA} \mathcal{F}^*_X$, then a complete ban on the covariates in $X'$ is not optimal for any designer with a FA-preference. That is, every designer can achieve a strictly higher payoff by garbling $(X,X')$ rather than by garbling $X$ alone. This comparison is distinct from earlier work that considers fully disclosing $(X,X')$ versus fully disclosing $X$.\footnote{Our property of FA dominance does not in general rank the information policy of fully revealing $X$ versus fully revealing $(X,X')$. That is, it may be that $\mathcal{F}^*_{X,X'}$ FA-dominates $\mathcal{F}^*_X$, but the designer's payoff is higher from  revealing $X$ than from  revealing $(X,X')$. 
} However, in some cases it will be optimal to disclose $(X,X')$ fully; Propositions \ref{prop:XGSimple} and \ref{prop:SimpleXXG} characterize when this is the case for designers with simple preferences.

\subsubsection{Example} \label{sec:SimpleExample}

We start with an example to demonstrate that banning an input can be optimal for the designer. Suppose $\mathcal{Y} = \{0,1\}$ and  $Y$ and $G$  are independently and uniformly distributed, i.e., $\mathbb{P}(Y = y, G = g) = 1/4$ for any $y \in \{0,1\}$ and $g \in \{r, b\}$. Let $X$ be a null signal; that is, $X = x_0$ with probability one. Further let $X'$ be a binary signal with the following conditional probabilities $\mathbb{P}(X' \mid Y, G)$:\footnote{In this example, neither $X$ nor $X'$ reveals group identity. Thus, this example falls outside of the settings considered in the previous two subsections.} 
\[\begin{array}{ccc}
 & X'=1 & X'=0 \\
 Y=1 & 1 & 0 \\
 Y=0 & 0 & 1 
\end{array} \quad \quad \quad \begin{array}{ccc}
    & X'=1 & X'=0 \\
 Y=1 & 0.6 & 0.4 \\
 Y=0 & 0.4 & 0.6 
\end{array}\]
\[G=r \hspace{50mm} G=b\]
That is, $X'$ is perfectly informative about the individuals in group $r$, and imperfectly informative about those in group $b$. Suppose  the loss function is $\ell(d,y) = \mathbbm{1}(d\neq y)$, the agent is Utilitarian ($\alpha_r=p_r=1/2$ and $\alpha_b=p_b=1/2$), and the designer is Egalitarian.

\begin{figure}[h]
    \centering
    \includegraphics[scale=0.6]{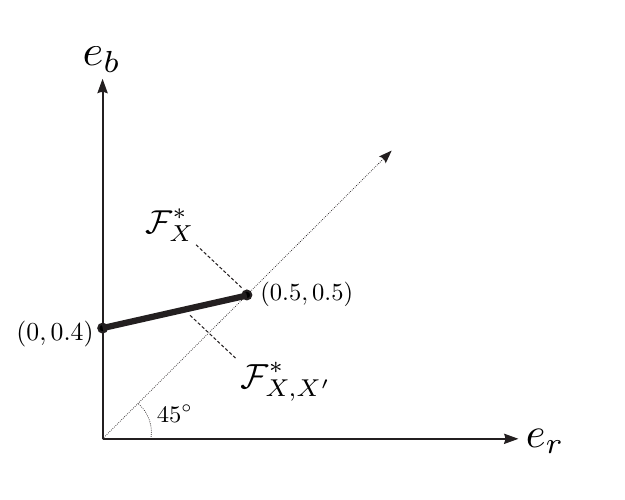}
    \caption{\footnotesize{The fairness-accuracy frontier given $(X,X')$ is the line segment connecting $(0, 0.4)$ with $(0.5, 0.5)$. Every nontrivial garbling of $(X,X')$ leads to a point on this frontier that is different from $(0.5,0.5)$, and hence yields a strictly negative payoff for the designer.}}\label{fig:example}
\end{figure}

The input-design feasible set given $X$ only is the singleton $\{(0.5,0.5)\}$, and the Egalitarian designer's payoff at this point is zero. But if the designer chooses any garbling of $(X,X')$ that provides information about $X'$, his payoff will be strictly negative. Intuitively, the Utilitarian agent uses what he learns about $X'$ to maximize aggregate accuracy. Since this information is necessarily more informative about group $r$ than about group $b$, decisions based on this information increase the gap between the two group errors, reducing the designer's payoff.\footnote{While we assume an Egalitarian designer here for simplicity, a similar construction is possible for any designer who places sufficient weight on fairness considerations.} Thus, it is strictly optimal for the designer to exclude all information about $X'$. 
    
\subsubsection{Excluding Group Identity} \label{sec:GNew}

We next consider the special case in which $X'=G$. The property of group balance (suitably strengthened) turns out to imply that banning group identity is never optimal. 

\begin{definition} Say that $X$ is \emph{strictly group-balanced} if $e_r<e_b$ at $r_X$ and $e_b<e_r$ at $b_X$. \label{def:StrictGroupBalance}
\end{definition}

\noindent Relative to group-balance, strict group-balance rules out covariate vectors $X$ for which $r_X = b_X = f_{X}$. 

\begin{proposition} \label{prop:ExcludeG}
Suppose $r_X,b_X \in H$. Then $\mathcal{F}^*_{X,G} >_{FA} \mathcal{F}^*_X$ if and only if $X$ is strictly group-balanced.\footnote{
The assumption $r_X, b_X \in H$ makes the above result easier to state as an if-and-only-if condition. But it follows from our proof of Proposition \ref{prop:ExcludeG} that even when this assumption fails, strict group-balance is a sufficient condition for the frontier to uniformly worsen when excluding $G$.}
\end{proposition}
 
That is, if (and only if) $X$ is strictly group-balanced, every error pair on the fairness-accuracy frontier given $X$  is FA-dominated by an error pair on the fairness-accuracy frontier given $(X,G)$. This result builds on previous findings that \emph{disparate treatment} (using different rules for individuals in different groups) may be necessary to preclude \emph{disparate impact} (effecting disparate harms across groups).\footnote{This tension between disparate treatment and disparate impact is noted in explicitly in works such as \citet{chouldechova} and \citet{RambachanEtAl}, and is implied by results in \citet{ChanEyster}.} Specifically, Proposition \ref{prop:ExcludeG}  implies that to reduce disparate impact, it may be necessary to impose information policies that are asymmetric across groups. Interestingly, this may not involve fully revealing $G$, so the algorithm may be formally group-blind (thus not exhibiting disparate treatment).\footnote{The algorithm exhibits disparate treatment if, holding all other covariates equal, it yields different outputs depending on the individual's group identity. See  \url{https://www.justice.gov/crt/book/file/1364106/download} for definitions of disparate treatment and impact.}  Nevertheless, if we consider the total procedure---taking into account both information design and algorithm design---then two individuals who are otherwise identical but belong to different groups may receive different distributions of outcomes. This distinction brings up an interesting question regarding how disparate treatment should be conceptualized in settings where both information design and algorithm design are present.

For restricted classes of preferences, we can further characterize how the variable $G$ is used in the designer's optimal garbling. For example, take the class of simple utility functions introduced in Section \ref{sec:altchar}, which are of the form
\[-\gamma_r e_r - \gamma_b e_b - \gamma_f \vert e_r - e_b \vert\]
where $\gamma_r,\gamma_b>0$ and $\gamma_f \geq 0$. We will show that every designer with a utility function of this form can achieve their optimal payoff using either of two garblings of $(X,G)$ that we  now define.

First, for any $(X,G)$, let \emph{the fully revealing garbling} $T_{X,G}$ be the garbling that directly reveals  $(X,G)$, i.e., every $(x,g)$ is mapped to itself with probability 1. 

Second, recall that $a_g^*$ denotes the $g$-optimal algorithm (which assigns to each covariate vector $x$ the decision that minimizes group $g$'s error). Let $\overline{a}_r$ denote the algorithm that instead assigns to each covariate vector $x$ the action $1-a_r^*(x)$ (i.e., the action which is not optimal for group $r$), and let $\overline{e}_r$ denote the group $r$ error under $\overline{a}_r$. Further let  $(e^R_r,e^R_b)$ denote the pair of group errors at the $r$-optimal point $r_X$. The \emph{$r$-shaded garbling}, denoted $T^r_{X,G}$, maps each $(x,b)$ to the message $a_b^*(x)$ with probability 1, and maps each $(x,r)$ to the message $a^*_r(x)$ with probability
\[
\beta = \max\left\{\frac{\overline{e}_r - e_b^R}{\overline{e}_r - e_r^R}, ~~0 \right\} 
\]
and to the message $1-a^*_r(x)$ otherwise. Intuitively, the $r$-shaded garbling preserves all the information in $X$ for members of group $b$, but adds noise to this information for members of group $r$. Fixing any value of $\overline{e}_r$, the amount of weight on the message $a_r^*(x)$ (i.e., the ``right'' action) is increasing in $e_r^R$ and decreasing in $e_b^R$. That is, the better off group $r$ is (and the worse off group $b$ is) at group $r$'s optimal point, the more noise the $r$-shaded garbling adds to group $r$'s covariates. 

\begin{proposition} \label{prop:XGSimple} Suppose the designer has a simple preference with parameters $(\gamma_r,\gamma_b,\gamma_f)$. Fix any $X$ for which $f_{X,G} \in H$. Suppose  (without loss) that $(X,G)$ is either group-balanced or $r$-skewed. Then:
\begin{itemize}
    \item[(a)] The fully revealing garbling $T_{X,G}$ is optimal if $\gamma_r \geq \gamma_f$.
    \item[(b)] The $r$-shaded garbling $T^r_{X,G}$ is optimal if $\gamma_r \leq \gamma_f$.
\end{itemize}
\end{proposition}

That is, if the designer's weight on the inequity term $\vert e_r - e_b \vert$ is sufficiently strong, he will prefer the $r$-shaded garbling. Otherwise, he achieves his optimal point by revealing all of the information in $(X,G)$ un-garbled. In the knife-edge case $\gamma_r=\gamma_f$, the designer is indifferent between these garblings (both are optimal).

This result complements papers that compare choice between decision rules based on $(X,G)$ to choice between decision rules based on $X$ alone (e.g., \citet{ChanEyster} and \citet{RambachanEtAl}). We find that under certain conditions, disclosing $(X,G)$ is not only superior to disclosing $X$, but is in fact optimal over the class of all possible garblings of $(X,G)$.

\subsubsection{Excluding a Covariate When Group Identity is Known} \label{sec:GKnown}

Next compare the frontier implemented by garblings of $(X,G)$ with the frontier implemented by garblings of $(X,G,X')$, where $X$ and $X'$ are arbitrary covariate vectors.

\begin{definition} Say that $X'$ is \emph{decision-relevant over $X$} if for each group $g$,   
there are realizations $(x,x')$ and $(x,\tilde{x}')$ of $(X,X')$ such that
\[\{1\} = \argmin_{d \in \mathcal{D}} \mathbb{E}[\ell(d,Y) \mid (X,X',G) = (x,x',g)]\]
while
\[\{0\} = \argmin_{d \in \mathcal{D}} \mathbb{E}[\ell(d,Y) \mid (X,X',G)=(x,\tilde{x}',g)]\]
where each of $(x,x',g)$ and $(x,\tilde{x}',g)$ has strictly positive probability.
\end{definition}

This weak condition requires only that there is some individual in each group $g$ for whom the decision that maximizes (expected) accuracy is different given $X$ and given $(X,X')$. For example, if $X'$ is a test score and $X$ is high school GPA, then $X'$ is decision-relevant for group $g$ if taking the test score into consideration reverses the admission decision for at least one individual in group $g$ relative to the decision based on GPA alone. 

\begin{proposition} \label{prop:ExcludeX'overG}
Suppose $r_X,b_X \in H$.  For any covariate vector $X$ and any covariate vector $X'$ that is decision-relevant over $X$, $\mathcal{F}^*_{X,X',G} >_{FA} \mathcal{F}^*_{X,G}$. 
\end{proposition}

This result says that so long as the designer has access to group identity, then bans cannot be justified for any minimally informative covariate. If we restrict to the class of simple preferences, we can again characterize how that covariate should be used in the designer's optimal garbling. 

\begin{proposition} \label{prop:SimpleXXG} Suppose the designer has a simple preference with parameters $(\gamma_r,\gamma_b,\gamma_f)$. Fix any $(X,X')$, and suppose without loss that $(X,X',G)$ is either group-balanced or $r$-skewed. Then:
\begin{itemize}
    \item[(a)] The fully revealing garbling $T_{X,X',G}$ is optimal when $\gamma_r \geq \gamma_f$.
    \item[(b)] The $r$-shaded garbling $T^r_{X,X',G}$ is optimal when $\gamma_r \leq \gamma_f$.
\end{itemize}
\end{proposition}

We can apply these results to comment on the question of whether to ban test scores in college admissions decisions. College entrance exams are decision-relevant for admissions, even given the rest of the application \citep{BerkeleySAT}.\footnote{Specifically, Section A of \citet{BerkeleySAT} finds that test scores are predictive of college success, predictive above other covariates (such as GPA), and and predictive for all demographic groups that they consider (with individuals disaggregated by factors such as parental education, family income, and racial/ethnic identity).} Thus Proposition \ref{prop:ExcludeX'overG} implies that so long as group identities are permissible inputs for college admission decisions, then excluding test scores is welfare-reducing for all designers with the ability to garble available covariates. On the other hand, if group identity is not  permitted as an input into college admissions decisions, then   a sufficiently fairness-minded designer may find it optimal to  completely exclude test scores. With regards to the recent Supreme Court case \emph{Students for Fair Admissions, Inc. v. President and Fellows of Harvard College}, our result suggests that banning affirmative action  nationwide may give universities with certain FA preferences reason to  ban the use of test scores in admissions decisions.\footnote{\citet{DesseinFrankelKartik} demonstrate a similar finding in a model in which universities experience costs when making decisions that differ from the preferences of a broader society.}  Proposition \ref{prop:SimpleXXG} further says that among designers with simple preferences, those who sufficiently value accuracy will prefer to  reveal test scores for all students, while those who sufficiently value fairness will prefer to  use the full informational content of test scores for students in the disadvantaged group, but to add noise to this information for students in the advantaged group. This latter garbling is consistent with \citet{BerkeleySAT}'s report that one use of test scores at UC Berkeley (prior to the university's move to test-blind admissions in 2021) was to identify otherwise ineligible applicants from relatively disadvantaged backgrounds.

While our framework abstracts away from many important features of the college admissions process---including capacity constraints (see Subsection \ref{sec:cc}), access to testing \citep{GargLiMonachou} and test-optional admissions \citep{DesseinFrankelKartik})---the link between the availability of group identity and the value of additional information, such as test scores, is one that we believe holds more generally. The crucial point is that when group identity is available, the designer can tailor how the additional information is used for each group separately.  In this sense, access to group identity has a positive spillover effect for the value of other covariates, guaranteeing that there is some (possibly group-dependent) garbling of the other information that aligns the agent's and designer's incentives.

We conclude by returning to our example in Section \ref{sec:SimpleExample}, and illustrating how our previous observations change when group identity is available as a covariate.

\begin{figure}[h]
    \centering
    \includegraphics[scale=0.6]{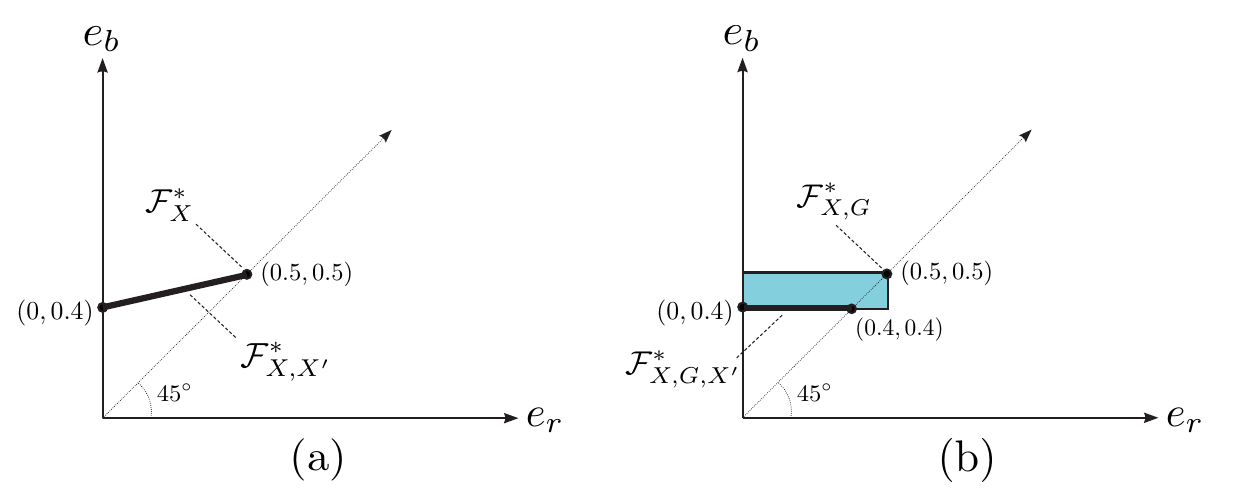}
    \caption{(a) A comparison of the input-design fairness-accuracy frontiers given $X$ versus given $(X,X')$; (b) A comparison of the input-design fairness-accuracy frontiers given $(X,G)$ versus given $(X,G,X')$}.\label{fig:example2}
\end{figure}

Panel (a) reproduces the previous Figure \ref{fig:example} (which compares the frontiers $\mathcal{F}^*_{X}$ and $\mathcal{F}^*_{X,X'}$), while  Panel (b) compares the fairness-accuracy frontiers $\mathcal{F}^{*}_{X,G}$ and $\mathcal{F}^{*}_{X,G,X'}$. Without access to group identity, we argued that the Egalitarian designer preferred to ban the covariate $X'$. With access to group identity, we see that the Egalitarian designer is able to achieve the superior outcome $(0.4,0.4)$ by garbling $(X,G,X')$. Thus while making information about $X'$ available to the agent is strictly harmful for the designer when group identity is not available, this ceases to be true once the designer can garble $X'$ in different ways depending on $G$.

\section{Empirical Application} \label{sec:Empirical}

We have so far focused on general conceptual findings that hold across settings, but our framework can also be used to better understand the fairness-accuracy tradeoffs in specific datasets. In this final section, we empirically illustrate some of our key definitions on two popular healthcare datasets from the algorithmic fairness literature. This analysis demonstrates that our framework is amenable to computation, and speaks to the potential relevance of our framework for empirical evaluations of algorithms. 

Section \ref{sec:Data} describes the two datasets. Section \ref{sec:BalanceSkew} evaluates group-balance and group-skew, finding that the first dataset is group-skewed, while the second is an interesting case of group-balance in which the group-optimal points approximately lie on the 45 degree line. Section \ref{sec:EstimateFrontier} depicts the fairness-accuracy frontier for both datasets, setting $\mathcal{A}$ to be the set of linear algorithms.  We find that in one dataset, the fairness-accuracy frontier consists primarily of strong fairness-accuracy conflicts (where the designer can only increase fairness by decreasing accuracy for both groups), while in the other there is nearly no conflict between fairness and accuracy.

Having estimated these frontiers, we apply Proposition \ref{thm:FullVersusBayes} to study the effectiveness of input design for each of the datasets. We find that for one of the datasets, input design is sufficient to recover the entire frontier, while in the other there are designers who cannot implement their favorite point using input design. Finally, we apply Proposition \ref{prop:XRevealsG} to study how the fairness-accuracy frontier changes when the designer is permitted group-specific algorithms. We find that for both datasets, neither Utilitarian nor Egalitarian designers have substantially improved payoffs, but designers with moderate fairness and accuracy preferences do benefit.

\subsection{Data} \label{sec:Data}  

Our first healthcare dataset consists of the 48,784 patient observations reported in \citet{ObermeyerMullainathan}. The covariate vector $X$ includes 8 demographic variables, 34 indicators of specific chronic illnesses, 13 healthcare cost variables, and 94 biomarker and medication variables.  We take as the two group identities whether the patient self-reported as Black or White, denoted $g \in \{b, w\}$. 

As described in \citet{ObermeyerMullainathan}, a large academic hospital used the covariates in $X$ to identify high-risk patients to enroll in an intensive health care program,  automatically enrolling those top 3\% ``highest risk'' patients into this program. In the data, ``true health needs'' are reported as each patient's total number of active chronic illnesses in the subsequent year, where the 97\% percentile is $6$ health conditions.  We thus define the patient's type $Y$ to be an indicator for whether  their true health needs are strictly larger than the 97\% percentile, and  consider algorithms $a: \mathcal{X} \rightarrow \{0,1\}$ that predict $Y$.
We use $\ell(d,y) = \mathbbm{1}(d\neq y)$ as our loss function, implying that algorithms are more accurate if they have a lower misclassification rate for each group, and more fair if they have a smaller disparity between the misclassification rates for the two groups.

Our second dataset is from \citet{strack2014impact} and contains 101,766 clinical care observations for patients with diabetes diagnoses. The covariate vector $X$ includes 25 variables, including demographic data (e.g., age) and medical information (e.g., diabetic medications and number of inpatient days).  We take gender (woman or man) to be the group identities, denoted $g \in \{w, m\}$. The patient's type $Y$ is whether the patient was readmitted to the hospital after release. This variable is reported in the data.  We consider algorithms $a: \mathcal{X} \rightarrow \{0,1\}$ that predict whether patients will be readmitted upon release, and again use $\ell(d,y) = \mathbbm{1}(d\neq y)$ as our loss function.

For both datasets, we suppose that the designer chooses a single algorithm for both groups and does not have access to group identity. (This is consistent with the analysis of \citet{ObermeyerMullainathan}.) We subsequently consider how the fairness-accuracy frontier changes when group-specific algorithms are permitted.

\subsection{Estimating Group-Balance versus Group-Skew} \label{sec:BalanceSkew}
For each group $g$, let $(e_r^g,e_b^g)$ denote the group-$g$ optimal point. We will first provide point estimates of these group-optimal points, given which we will form conjectures for whether the data are group-balanced or group-skewed. Subsequently we will formulate suitable statistical tests of those conjectures.

Figure \ref{fig:GBGS} reports five-fold cross-validated estimates of the group-optimal points in each dataset.\footnote{We randomly split the data into five equally sized subsets. In each step of the procedure, we designate one of these subsets to be the test set. We then split the remaining data (the training set) into group-$r$ and group-$b$ observations, and train a random forest algorithm on each part separately, thus obtaining a classifier for each group. We apply group $g$'s classifier to the test set and evaluate each group's error under this classifier,  obtaining an estimate of the group-$g$ optimal point $g_X=(e_r^G,e_b^G)$. Finally, we average these estimates across the five choices of which subset to use as the training data.} In the \citet{ObermeyerMullainathan} data, the classification error for Black patients is higher than for White patients at both groups' optimal points, suggesting that the covariate vector $X$ is $w$-skewed. In contrast, for the \citet{strack2014impact} data, the two group-optimal points are nearly on the 45 degree line---that is, the consequences for the two groups are nearly the same regardless of whether the algorithm designer minimizes error for female patients or male patients. This suggests that that the covariate vector satisfies a special case of group-balance.\footnote{The substantially larger errors for the latter dataset are due to a more difficult prediction problem: In the \citet{strack2014impact} data, the outcome variable is roughly equally likely to be 1 or 0 (respectively, 0.46 versus 0.54), whereas in \citet{ObermeyerMullainathan} data, the outcome variable takes the value 0 for almost all patients. Thus, for example, the naive rule that classifies all patients as 0 in the first data achieves a misclassification rate of 0.06 for Black patients and 0.02 for White patients, but these misclassification rates are not feasible for the second dataset.}

\begin{figure}
\begin{center}
\includegraphics[scale=0.4]{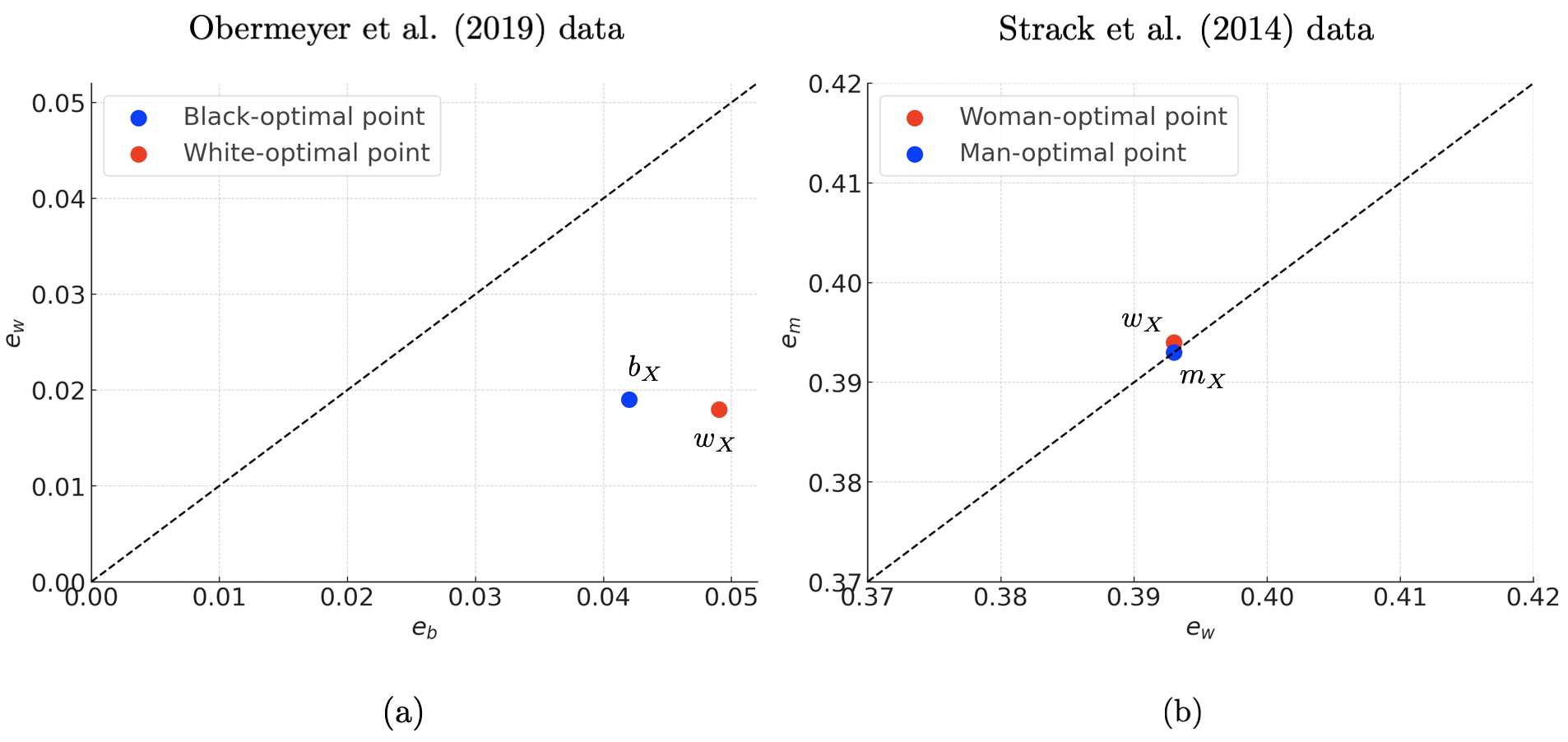}
\end{center}
\caption{\footnotesize{The group-optimal points for the \citet{ObermeyerMullainathan} data are estimated to be
$b_X=(\hat e_b^b, \hat e_w^b) = (0.042, 0.019)$ and $w_X=(\hat e_b^w, \hat e_w^w) = (0.049, 0.018)$, suggesting group-skew. The group-optimal points for the \citet{strack2014impact} data are estimated to be
$
w_X=(\hat e_w^w, \hat e_m^w) = (0.393, 0.394)$ and $m_X = (\hat e_w^m, \hat e_m^m) = (0.393, 0.393)$, suggesting group-balance.}} \label{fig:GBGS}
\end{figure}

We now evaluate these conjectures with suitable statistical tests. To assess whether the \citet{ObermeyerMullainathan} dataset is $w$-skewed, we test the null hypothesis
\begin{equation} \label{eq:null1}
H_0:
e_w^b \geq e_b^b
\end{equation}
against the alternative $H_1:
e_w^b < e_b^b$. (Since $e_w^w \leq e_w^b$ and $e_b^b \leq e_b^w$, the inequality $e_w^b < e_b^b$ implies $e_w^w \leq e_b^w$; thus, we only need to verify this part of the definition of group-skew.) We will conclude that the data is $w$-skewed if we can reject this null hypothesis at the desired significance level.

For the \citet{strack2014impact} data, because Figure \ref{fig:GBGS} suggests that the knife-edge condition $e_w^w = e_m^w = e_w^m = e_m^m$ is approximately met, we need to be more careful with our formulation of the null and alternative.\footnote{For example, a natural test for strict group-balance (Definition \ref{def:StrictGroupBalance}) would be to formulate the null $H_0: e_w^b \leq e_b^b \mbox{ OR } e_w^w \geq e_b^w$ against the alternative $H_1: e_w^b > e_b^b \mbox{ AND } e_w^w < e_b^b$, but Figure \ref{fig:GBGS} does not suggest the data is strictly group-balanced. Thus we do not expect to reject this null.} We test a relaxed notion of group-balance corresponding to whether $e_w^w$,  $e_m^w$, $e_w^m$, and $e_m^m$ are located within a small  neighborhood of one another. Specifically, we test
\begin{equation} \label{eq:null2}
H_0 \colon |e_w^w - e_m^w| \geq \delta \text{ OR } |e_w^m - e_m^m| \geq \delta
\end{equation}
against the alternative $H_1 \colon |e_f^w - e_m^w| < \delta \text{ AND } |e_f^m - e_m^m| < \delta$ for $\delta = 0.01$. Under the alternative, both group-optimal points have the property that each group's misclassification rate is within 1 percentage point of the other. (For a sense of whether 0.01 is ``small,'' recall that our estimate for the misclassification rates at the group-optimal points are about 0.40.) When the two group-optimal points are nearly on the 45-degree line, they must nearly coincide. We will thus conclude (approximate) group-balance if (\ref{eq:null2}) is rejected at the desired significance level.

Following \cite{ALTO}, we implement the following procedure to test these null hypotheses: In each of $K=5$ iterations, we randomly split the data into a training set (consisting of two-thirds of the observations) and a test set (consisting of the remaining observations). On the training subset, we  search for the algorithm that minimizes group $g$'s error on the group $g$ observations. Then we evaluate each group's error under this algorithm on the test data, and compute a $p$-value for the suitable null hypothesis via bootstrap (see Appendix~\ref{app:bootstrap} for details). A valid test is obtained by rejecting the null hypothesis whenever the median $p$-value across these iterations falls below half of the desired significance level \citep{ALTO}. For a 5\% significance level, this corresponds to rejecting whenever the median $p$-value falls below 0.025.

We report the outcome of this test for each dataset under two possible specifications of the set of algorithms $\mathcal{A}$: the set of unconstrained algorithms $\overline{\mathcal{A}}$ and the set of all linear algorithms. For the former specification, we employ a random forest algorithm to search for each group's optimal algorithm among the class of unconstrained algorithms. For the latter, we use
logistic regression to look for the linear classifier that minimizes each group's error on the training data. Table \ref{tab:p_val} reports the output of this analysis. We reject the null in (\ref{eq:null1}) at a $5\%$ significance level for both specifications of $\mathcal{A}$, suggesting that the covariates in the \citet{ObermeyerMullainathan} data are group-skewed. (Recall that we reject at a $5\%$ significance level if the median $p$-value is less than $0.025$.) We reject the null in (\ref{eq:null2}) at a $10\%$ significance level for the unconstrained set of algorithms and at a $5\%$ significance level for the set of linear algorithms, suggesting that the covariates in the \citet{strack2014impact} dataset are (approximately) group-balanced with group-optimal points on the 45 degree line.

\begin{table}[hbt]
  
    \footnotesize % Consider setting the font size as desired
    \begin{tabular}{
        @{} % Removes leading space
        l % Left aligned for text in the first column
        *{2}{c} % Three columns with numeric content, aligned at the decimal point
        | % Vertical line to separate groups
        *{2}{c} % Repeat for each group of three columns
        % Centered column for p^M
        @{} % Removes trailing space
    }
    \toprule
    & \multicolumn{2}{c}{\bfseries Dataset 1} & \multicolumn{2}{c}{\bfseries Dataset 2} \\
    \cmidrule(lr){2-3} \cmidrule(lr){4-5} 
         & {Group-Balance/Skew} & {median $p$-value} & {Group-Balance/Skew} & {median $p$-value}  \\
    \midrule
        Unconstrained & $w$-skewed & $<0.0001$
        & group-balanced & $0.0437$
        \\
        Linear& $w$-skewed & $<0.0001$
        & group-balanced & $<0.0001$
        \\
    \bottomrule
    \end{tabular}
    \vspace{8pt}
      \caption{ \footnotesize{We report whether each dataset is group-balanced or group-skewed for the two specifications of $\mathcal{A}$. Reported $p$-values are computed via bootstrap with size 10,000.}}\label{tab:p_val}

\end{table}

Thus, by Corollary \ref{corr:Tradeoff}, the first dataset involves a potentially strong conflict between fairness and accuracy, where designers who put sufficient weight on reducing disparities across Black and White patients may prefer to increase errors for both groups. In contrast, fairness considerations across male and female patients do not rationalize implementing Pareto-dominated error rates for the second dataset, regardless of how much weight the designer puts on fairness. 

\subsection{The Fairness-Accuracy Frontier} \label{sec:EstimateFrontier}

We next depict the feasible set $\mathcal{E}_X$ and FA frontier $\mathcal{F}_X$ for each of these datasets, setting $\mathcal{A}$ to be the set of linear algorithms, which we will subsequently denote by $\mathcal{A}_\ell$.
First observe that the extreme points of the feasible set  can be found by solving the optimization problem
$
    \min \left\{\alpha_r e_r(a) + \alpha_b e_b(a) \colon a \in \mathcal{A}_\ell \right\}$ 
for different choices of $(\alpha_r, \alpha_b) \in \mathbb{R}^2$. We consider the empirical analogue of this optimization problem, in which $e_g(a)$ is replaced by its sample analogue:
\begin{equation} \label{eq:Opt}
\min \left\{\alpha_r \hat e_r(a) + \alpha_b \hat e_b(a) \colon
\forall g, \ \hat e_g(a) = \frac{1}{n_g} \sum_{i \colon G_i = g} \ell\left(a(X_i), Y_i\right), \ 
a \in \mathcal{A}_\ell \right\},
\end{equation}
with $n_g$ denoting the number of group-$g$ observations in the dataset.

Any linear classifier $a \in \mathcal{A}_\ell$ can be represented as $a(x) = \mathbbm{1} \left\{x^\top \beta \geq 0 \right\}$ for some $\beta \in \mathbb{R}^{\dim(\mathcal{X})}$.\footnote{We assume that $X$ includes a constant term,  so it is without loss to set the threshold to zero.}
Thus,  we can recast the optimization problem in (\ref{eq:Opt}) as the following mixed-integer linear program (MILP):
\[
\mathrm{(MILP)} \left[
    \begin{matrix}
        \displaystyle{\min_{\beta, \hat e_r \geq 0 , \hat e_b \geq 0, \hat Y}} & \alpha_r \hat e_r + \alpha_b \hat e_b \\
        \mathrm{s.t.}
        & \displaystyle{\hat e_r = \frac{1}{n_r}} \sum_{i \colon G_i = r} \mathbbm{1}  \left\{ \hat Y_i \neq Y_i \right\} \\[8pt]
        & \displaystyle{\hat e_b = \frac{1}{n_b}} \sum_{i \colon G_i = b} \mathbbm{1} \left\{ \hat Y_i \neq Y_i \right\} \\[8pt]
        
        & \displaystyle{\frac{X_i^\top \beta}{C_i} < \hat Y_i \leq 1 + \frac{X_i^\top \beta}{C_i}} \mbox{ and } \hat{Y}_i \in \{0,1\}  & \text{for all $i$} \\
    \end{matrix}
    \right.
\]
where $C_i$ is chosen to satisfy $C_i > \sup_{\beta \in \mathcal B} |X_i^\top \beta|$, with $\mathcal B$ some compact set such that we restrict $\beta \in \mathcal B$.\footnote{The third constraint is equivalent to $\hat Y_i = \mathbbm{1} \left\{X_i^\top \beta \geq 0\right\}$ for all $i$ since if $X_i^\top \beta$ is weakly positive then the constraint $\hat Y_i \in \{0,1\}$ implies $\hat Y_i =1$, while if $X_i^\top \beta$ is strictly negative then the constraint $\hat Y_i \in \{0,1\}$ implies $\hat Y_i=0$.} In principle, we could use this MILP to estimate all the extreme points of the feasible set by varying the normal vector $(\alpha_r, \alpha_b)$ within $[0, 2\pi) \times [0, 2\pi)$. Since this program is computationally burdensome, we instead employ a standard relaxation of the MILP to a linear program (LP) by eliminating the integer constraints using a hinge surrogate loss function (see Appendix~\ref{app:LP} for  details).\footnote{The same reduction underlies the construction of support vector machines for linear classification.
See \cite{hastie2009elements} for a textbook reference.
}
We solve this LP for $(\alpha_r, \alpha_b) \in [0, \pi/72, 2\pi/72, \dots, 2 \pi)^2$, estimate $(\hat e_r, \hat e_b)$ via five-fold cross-validation for each $(\alpha_r, \alpha_b)$, and finally take the convex hull of the estimated error pairs.\footnote{
For each pair $(\alpha_r, \alpha_b)$, execute the following steps: Partition the dataset into five subsets of equal size. In each of the five iterations, use one subset for testing and the remaining four for training. Determine $\hat \beta^{(k)}$ by solving the LP with the training data, which establishes a linear classifier $a^{(k)}$. Compute the estimates of group errors $(\hat e_r^{(k)}, \hat e_b^{(k)})$ for classifier $a^{(k)}$ by sample analogues using test data. Finally, calculate the average group error $\hat e_g$ for each group $g$ across all five folds: $\hat e_g = \frac{1}{5} \sum_{k=1}^{5} \hat e_g^{(k)}$.
}

\begin{figure}[htbp]
Obermeyer et al.'s dataset \\
\begin{tabular}{cc}
  \begin{minipage}[t]{0.45\linewidth}
    \centering
    \includegraphics[keepaspectratio, scale=0.65]{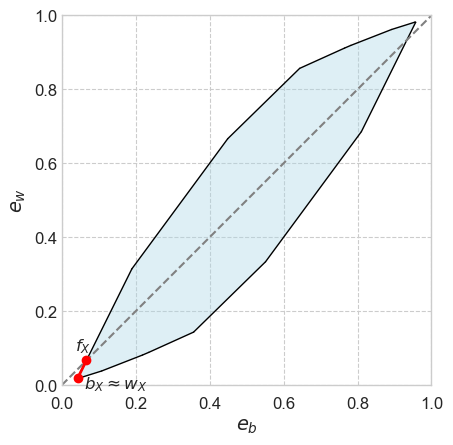}
    \label{fig:science_whole}
    \subcaption{}
  \end{minipage} &
  \begin{minipage}[t]{0.45\linewidth}
    \centering
    \includegraphics[keepaspectratio, scale=0.65]{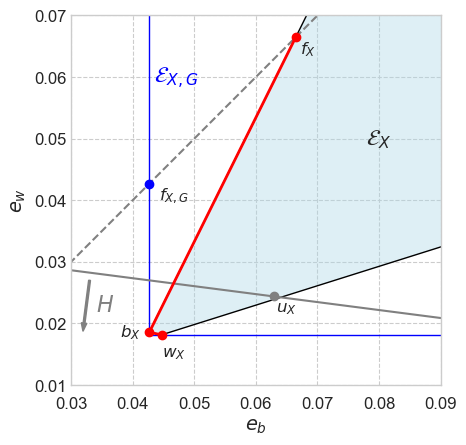}
    \label{fig:diabetes_gender_whole}
    \subcaption{}
  \end{minipage}
\end{tabular}

Strack et al.'s dataset

\begin{tabular}{cc}
  \begin{minipage}[t]{0.45\linewidth}
    \centering
    \includegraphics[keepaspectratio, scale=0.65]{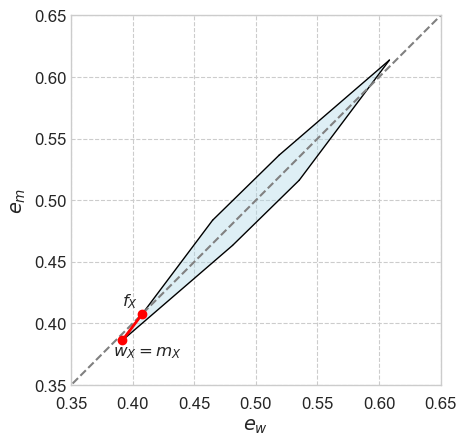}
    \label{fig:science_withG}
    \subcaption{}
  \end{minipage}
   &
  \begin{minipage}[t]{0.45\linewidth}
    \centering
    \includegraphics[keepaspectratio, scale=0.65]{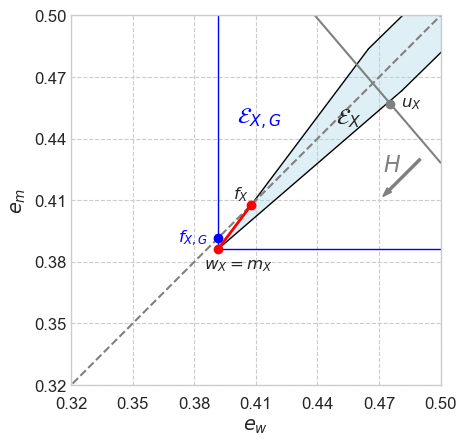}
    \label{fig:diabetes_gender_withG}
    \subcaption{}
  \end{minipage}
\end{tabular}

  \caption{Estimated feasible set for the class of linear algorithms}
  \label{fig:estimated_feasible_set}
  \vspace{12pt}
  
  \footnotesize
  Panels (A) and (C) display the estimated feasible sets for the datasets from Obermeyer et al. and Strack et al., respectively. In these panels, the black solid lines define the boundaries of the feasible sets $\mathcal{E}_X$, and the gray dotted lines serve as 45-degree reference lines. Panel (A) identifies $b_X$ and $w_X$ as the optimal points for Black and White patients, while Panel (C) marks $w_X$ and $m_X$ as the optimal points for women and men. Both panels highlight $f_X$ as the fairness-optimal point.
  Panels (B) and (D) examine the augmented feasible sets $\mathcal{E}_{X, G}$ for the same datasets with covariate $(X, G)$, represented by blue solid lines. In these panels, $f_{X, G}$ indicates the fairness-optimal point for each dataset with $(X, G)$. The point $u_X$ indicates the error pair chosen by the Utilitarian agent given no information, and the gray solid line in Panel (B) and (D) represents the Utilitarian's indifference curve through $u_X$.
  \end{figure}

Panels (A) and (B) of Figure~\ref{fig:estimated_feasible_set} depict the estimated feasible set and FA frontier for the \citet{ObermeyerMullainathan} data. Panel (A) shows the entire feasible set, while Panel (B) zooms in on the FA frontier. Consistent with our previous hypothesis test, the group-optimal points $b_X$ and $w_X$ lie on the same side of the 45 degree line. Additionally, we find that $b_X$ and $w_X$ are very similar---indeed, we cannot reject that $b_X$ and $w_X$ are different points at any reasonable significance level.\footnote{Specifically, we test the null hypothesis \[H_0: |e_b^b - e_b^w| \geq \delta \quad \mbox{OR} \quad |e_w^b - e_w^w| \geq \delta\]
against the alternative $H_1: |e_b^b - e_b^w| < \delta$ and $|e_r^b - e_r^w| < \delta$, with $\delta=0.01$. The resulting median $p$-value is less than $0.0001$, indicating insufficient evidence to assert a significant difference in the group-optimal points.
}
Thus the main tradeoff between fairness and accuracy is whether the designer is willing  to increase errors for both groups in return for a decrease in the disparity between group errors. The depicted FA frontier qualitatively resembles Panel (B) of Figure \ref{fig:SpecialCases} (Conditional Independence), and is consistent with a setting in which the optimal algorithm is the same for both groups, but the measured covariates are more predictive of the outcome for one group (White patients) than the other (Black patients).\footnote{We cannot directly test the assumptions of Example \ref{ex:CI} in our data, due to an insufficient number of observations per covariate vector.}

Panels (C) and (D) depict the analogous figures for the \citet{strack2014impact} data. Consistent with the result of our earlier hypothesis test, the estimated group-optimal points are very close to the 45-degree line (and consequently are also close to the estimated fairness-optimal point $f_X$). For this dataset, there is almost no tradeoff between fairness and accuracy. In fact, the depicted feasible set and FA frontier qualitatively resemble Panel (A) of Figure \ref{fig:SpecialCases} (Strong Independence), and are consistent with a setting in which the joint distributions of covariates and outcomes are nearly identical for the two groups. 

The two datasets also differ in their implications for the limits of input design. In Panels (B) and (D), we plot the error pair corresponding to the best payoff that a utilitarian agent can achieve with no information (labeled $u_X$), as well as the agent's indifference curve at this point.\footnote{
The payoff of the Utilitarian on the indifference curve is $e_0$ defined in equation~\eqref{eq:no_info_utilitarian}, setting $(\alpha_r, \alpha_b) \coloneqq (p_r, p_b)$. Point $u_X = (e_r^U, e_b^U)$ satisfies $p_r e_r^U + p_b e_b^U = e_0$.
} In Panel (B), the fairness-optimal point $f_X$ falls outside the halfspace $H$ of error pairs that the utilitarian agent prefers over $u_X$. Thus, by Proposition \ref{thm:FullVersusBayes}, input design is with loss for some designers. Specifically, sufficiently fairness-motivated designers cannot induce a utilitarian agent to implement their favorite error outcome, regardless of which garbling of the available covariates they choose. In contrast, in Panel (D) the points $w_X$, $m_X$, and $f_X$ all belong to the corresponding halfspace $H$, implying that every designer with a FA preference can achieve his most preferred outcome.

Finally, we can apply Proposition \ref{prop:XRevealsG} to quantify the possible fairness-accuracy improvements when different algorithms are used to make predictions for each group. Panels (B) and (D) of Figure~\ref{fig:estimated_feasible_set}  illustrate the change in the fairness-accuracy frontier when group-specific linear algorithms are permitted. In both datasets, neither the Utilitarian designer, nor the designer whose payoff is $-\vert e_r - e_b \vert$, benefit from using group-specific algorithms: The Utilitarian designer's payoff remains approximately the same, since the group-optimal points change very little when group-specific algorithms are permitted. The designer whose payoff is $-\vert e_r-e_b\vert$ does not benefit either, as both $f_X$ and $f_{X,G}$ lie on the 45 degree line (and thus yield an identical payoff of zero for this designer). 

This finding is especially interesting when juxtaposed with a recent debate over whether to use race-blind healthcare algorithms. Advocates for race-aware algorithms often adopt a Utilitarian perspective (e.g., \citet{Manski}), while proponents of race-blind healthcare algorithms typically argue from the perspective of minimizing healthcare inequalities across groups (e.g., \citet{Vyasetal}). But for these two datasets, it is intermediate designers who value both fairness and accuracy, rather than these two extremes, that benefit the most from use of group-specific algorithms.\footnote{If we consider the class of simple preferences for these two datasets, then for any fixed values of $\alpha_r$ and $\alpha_b$, the increase in the designer's payoff from use of group-specific algorithms is concave in $\alpha_f$ with an interior maximum.}  Whether this particular finding extends to  other datasets will depend on specific details of their FA frontiers, but our framework provides a general methodology that can be applied case-by-case to study the particular fairness-accuracy implications of different datasets and algorithmic constraints.

\section{Extensions} \label{sec:Extensions}

\subsection{Different loss functions for evaluating fairness and accuracy.} When defining the strict order $>_{FA}$, we used the same loss function to evaluate both accuracy and fairness. This is suitable, for example, for healthcare decisions where both the healthcare provider (designer) and patients agree that more accurate decisions are better, and so fairness can be reasonably evaluated as the disparity in accuracy across groups. In other cases where the subjects' utility function is different from the designer's, policymakers sometimes evaluate accuracy using one loss function and fairness using another. For example, if the algorithm guides hiring decisions, then fairness may be evaluated as the difference in hiring rates across groups, even while accuracy is evaluated based on whether suitable candidates are hired.  In Appendix \ref{app:DifferentLoss} we develop a more general version of our framework that allows for different loss functions, and extends Theorem \ref{thm:FullDesignPareto} and Corollary \ref{corr:Tradeoff} under a  generalization of group-balance. 

\subsection{Beyond absolute difference for evaluating fairness.} Our main analysis assumes that (un)fairness is evaluated according to the absolute difference of errors between the two groups, i.e. $\vert e_r - e_b \vert$. A natural extension is to consider $\vert \phi(e_r) - \phi(e_b) \vert$ where $\phi$ is some continuous  strictly increasing function. For instance, if $\phi$ is $\log$, then this corresponds to evaluating fairness using the ratio of errors rather than their difference. Theorem \ref{thm:FullDesignPareto} holds for any such $\phi$ with the fairness optimal point $f_{X}$ suitably defined.\footnote{To see why, first note that no interior point can be on the frontier. Otherwise, we can always find some $\epsilon_{1},\epsilon_{2}>0$ such
that $\left|\phi\left(e_{r}-\epsilon_{1}\right)-\phi\left(e_{b}-\epsilon_{2}\right)\right|\leq\left|\phi\left(e_{r}\right)-\phi\left(e_{b}\right)\right|$
so $\left(e_{r}-\epsilon_{1},e_{b}-\epsilon_{1}\right)>_{FA}\left(e_{r},e_{b}\right)$ yielding a contradiction. The rest of the proof follows as in Theorem \ref{thm:FullDesignPareto}.} We further demonstrate that the frontier becomes larger (smaller) whenever $\phi$ is concave (convex). Thus, for example, evaluating fairness using ratios instead of absolute difference results in a larger frontier, although the qualitative properties of this frontier are unchanged.

\subsection{Other agent preferences in the input design problem.} Section \ref{sec:DesignInputs} considers misaligned incentives between a designer controlling inputs and an agent setting the algorithm. There, we assume that the agent cares about accuracy and prefers for both group errors to be lower. In Appendix \ref{app:ExtendBD}, we consider what happens when this misalignment is more extreme and the agent is adversarial (i.e. negatively biased) towards one of the two groups, preferring that group's error to be higher. We generalize several results from Section \ref{sec:DesignInputs} and show  that even if the agent is negatively biased, it can still be optimal for the designer to provide information about group identity (so long as the bias is not too extreme). 

Two other potential generalizations would permit the agent and designer to have different loss functions, or permit the agent to  care about fairness. In both cases, the set of points that the agent prefers over the prior (what we defined to be $H$) is no longer a halfspace from the designer's perspective. Moreover,  non-linearities in the agent's objective function imply that the agent's ex-ante and ex-post problems may be different, and so it is relevant whether the agent commits to the algorithm or chooses the decision after the realization of the garbling. We consider these problems beyond the scope of the present paper, and leave them as open questions for future work.

\subsection{Capacity constraints.}\label{sec:cc} In our main model, we allow the designer unconstrained choice of any algorithm. In a few of the applications of interest, there may be an additional capacity constraint on the algorithm, e.g., if only a fixed number of students can be admitted in admissions decisions. One way to formulate a capacity constraint is a restriction on the ex-ante probability of assignment of decision $d=1$ (e.g., admit). In this case, the set of error pairs satisfying the constraint can be shown to be a convex set, so the feasible set is simply the intersection between the feasible set (as we have defined) and the convex set of error pairs that satisfy this capacity constraint. Our Theorem \ref{thm:FullDesignPareto} then applies for this new feasible set, although the fairness-accuracy frontier as characterized in Proposition \ref{prop:XRevealsG} may no longer be a horizontal line. 

\subsection{More than two groups or two decisions.}  We have assumed that there are two groups $\mathcal{G}=\{r,b\}$. Some of our results, such as Proposition \ref{thm:FullVersusBayes}, can be shown to directly extend for any finite $\mathcal{G}$. However, in order to extend our other results, we would first have to specify a definition of fairness for multiple groups. One possible generalization of the FA-dominance relationship is to say that a vector of group errors $(e_g)_{g \in \mathcal{G}}$ FA-dominates another vector $(e_g')_{g \in \mathcal{G}}$ if $e_g \leq e_g'$ for every group $g$, and also $\vert e_g - \frac{1}{\vert \mathcal{G} \vert} \sum_{g \in \mathcal{G}} e_g  \vert \leq \vert e_g' - \frac{1}{\vert \mathcal{G} \vert} \sum_{g \in \mathcal{G}} e_g'  \vert$ for every $g\in \mathcal{G}$, with at least one inequality holding strictly. That is, fairness is improved if each group's error is closer to the average group error. We expect our characterization in Theorem \ref{thm:FullDesignPareto} to extend qualitatively in this case. 

 We have also assumed that there are two decisions $\mathcal{D}=\{0,1\}$. All of our  results in Section \ref{sec:FAFrontier} about the unconstrained problem directly extend for any finite $\mathcal{D}$. However, Lemma \ref{lemm:BayesDesign} (the relationship between the input-design fairness-accuracy frontier and the unconstrained fairness-accuracy frontier) relies on the assumption of a binary decision. We leave a characterization of the input design frontier  for this more general case to future work. 

\section{Conclusion}

We conclude with possible directions for future work, and some mention of recent work in these directions.

First, our proposed fairness-accuracy frontier is defined with respect to the true underlying population distribution $\mathbb{P}$. In practice, analysts may instead have access to a sample of observations from this distribution. A natural question is whether, and how, the analyst can estimate the FA frontier (or deduce its properties) from such a sample.  \citet{liu2024inference} provides an approach for nonparametrically estimating our fairness-accuracy frontier. They characterize the asymptotic distribution of their estimator (showing that it converges to a tight Gaussian process as the sample size grows large). Thus, under certain conditions, it is possible to nonparametrically estimate the FA frontier (i.e., for the unconstrained set of algorithms $\overline{\mathcal{A}}$). \citet{ALTO} focus on the related question of how we can tell whether a given algorithm produces group errors that are on the FA frontier. This question is especially important in light of disparate impact claims, since algorithms that are shown to have disparate impact can sometimes be justified if that disparate impact is shown to be necessary to achieve other business objectives, such as accuracy. \citet{ALTO} propose a sample-splitting approach, which among other things can be used to discern from a sample of observations whether a given algorithm is simultaneously fairness- and accuracy-improvable.

Second, our paper takes the population distribution $\mathbb{P}$ as exogenously given. An interesting direction for follow-up work is the question of how to optimally acquire information with fairness-accuracy objectives in mind. In such a framework, $\mathbb{P}$ would be endogenous to the information acquisition choices of the designer. Our results shed some light on certain aspects of this design. For example, suppose it were possible to acquire new covariates that turned a group-skewed covariate vector into a group-balanced covariate vector. Corollary \ref{corr:Tradeoff} implies that such a change would change the nature of the fairness-accuracy conflict, eliminating the need to consider Pareto-dominated outcomes as a means to improve fairness. We leave to future work a more detailed exploration of  endogenously chosen covariates and their fairness-accuracy consequences.

Finally, in our input design problem we have assumed that the designer has knowledge of the agent's preferences. In practice, the designer may not know the agent's preferences precisely, or may face a problem of designing regulation for many agents simultaneously, where these agents hold different preferences. An interesting direction would thus be to consider optimal garblings given uncertainty over the agent's preferences, or garblings that optimize a worst-case criterion over a set of agent preferences. 

\bigskip

\appendix

\section{Proofs and Supporting Materials for Section \ref{sec:FullDesign}} 

\subsection{Supplementary Material to Section \ref{sec:DefineGroupBalance}} \label{app:SupportExamples}

Examples \ref{ex:SI} and \ref{ex:CI} 
respectively follow from Parts (a) and (b) of Proposition \ref{prop:SpecialCase}. Suppose the assumptions in Example \ref{ex:UnequalMeans} are satisfied. Given this loss function, the group-optimal point $g_X$ is implemented by the algorithm $a_g$. Moreover, we have
\begin{align*}
    \mathbb{E}[l(a_b(X),Y) \mid G=b] = Var(\varepsilon)  = \mathbb{E}[l(a_r(X),Y) \mid G=r] \leq \mathbb{E}[l(a_b(X),Y) \mid G=r].
\end{align*}
Symmetrically $\mathbb{E}[l(a_r(X),Y) \mid G=r] \leq \mathbb{E}[l(a_r(X),Y) \mid G=b]$. Note that if $a_r \not = a_g$ then both inequalities are strict and $X$ is strictly group-balanced. Otherwise, $r_X$ and $b_X$ are on the 45 degree line and $X$ is group-balanced.

Finally, suppose the assumptions in Example \ref{ex:Asymmetric} are satisfied. Given this loss function, the group-optimal points $g_X$ coincide, and they are implemented by the algorithm $a_0$. We have
\begin{align*}
    \mathbb{E}[l(a_0(X),Y) \mid G=b] & = \sum Var(Y\mid X=x,G=b) dP(x\mid G=b) \\
    & = Var(\varepsilon_b) > Var(\varepsilon_r) \\
    & = \sum Var(Y\mid X=x,G=r) dP(x\mid G=r) = \mathbb{E}[l(a_0(X),Y) \mid G=r]
\end{align*}
so the covariate vector is $r$-skewed.

\subsection{Characterization of the Feasible Set}

\begin{lemma} \label{lemm:ConvexPolygon} $\mathcal{E}_X$ is a closed and convex polygon. 
\end{lemma}

\begin{proof}Given a (randomized) algorithm $a$, we slightly abuse notation and let $a(x)$ denote the probability of choosing decision $d=1$ at covariate vector $x$. We further let $x_{y,g}$ denote the conditional probability that $Y = y$ and $G = g$ given $X = x$. Finally, let $p_x$ denote the  probability of $X = x$. Then the group errors can be written as follows: 
\begin{alignat*}{1}
e_g(a) & =\mathbb{E}\left[a\left(X\right)\ell\left(1,Y\right)+\left(1-a\left(X\right)\right)\ell\left(0,Y\right) :  G=g\right]\\
 & =\sum_{x} p_x \left(a\left(x\right)\sum_{y}\frac{x_{y,g}}{p_{g}}\ell\left(1,y\right)+\left(1-a\left(x\right)\right)\sum_{y}\frac{x_{y,g}}{p_{g}}\ell\left(0,y\right)\right),
\end{alignat*}
where $p_g$ is the prior probability that $G = g$. Note that this is linear in the algorithm $a\in [0,1]^\mathcal{X}$. Since $\Delta\mathcal{A}$ is closed and convex with a finite number of extreme points,
\[
\mathcal{E}_X = \{e(a) : a \in \Delta\mathcal{A}\}.
\]
is also closed and convex. Moreover, it is a polygon.
\end{proof}

\subsection{Proof of Theorem \ref{thm:FullDesignPareto}}

First, note that since $e>_{FA}e^{\prime}$
implies $e>_{PD}e^{\prime}$, $\mathcal{P}_{X}\subset\mathcal{F}_{X}$. Define the points 
\begin{alignat*}{1}
\bar{j} & :=\arg\max_{e\in\mathcal{E}_{X}}\left(e_{b}-e_{r}\right)\\
\underline{j} & :=\arg\min_{e\in\mathcal{E}_{X}}\left(e_{b}-e_{r}\right)
\end{alignat*}
where ties are broken in favor accuracy (either lower $e_{r}$ or
$e_{b}$). Note that $\bar{j}$ and $\underline{j}$ are well-defined
points in $\mathcal{E}_{X}$. Let $\mathcal{J}_{X}\subset\mathcal{E}_{X}$
denote the boundary of $\mathcal{E}_{X}$ between $\bar{j}$ and $\underline{j}$
(if $\mathcal{E}_{X}$ is a line segment, then $\mathcal{J}_{X}=\mathcal{E}_{X}$
unless $\mathcal{E}_{X}$ has slope $1$ in which case $\mathcal{J}_{X}=\left\{ \underline{j}\right\} =\left\{ \bar{j}\right\} $
is the left-most point of the line). Note that $\mathcal{P}_{X}\subset\mathcal{J}_{X}$
(see Figure \ref{fig:charproof}). 

\begin{figure}[h]
    \begin{center}
        \includegraphics[scale=0.65]{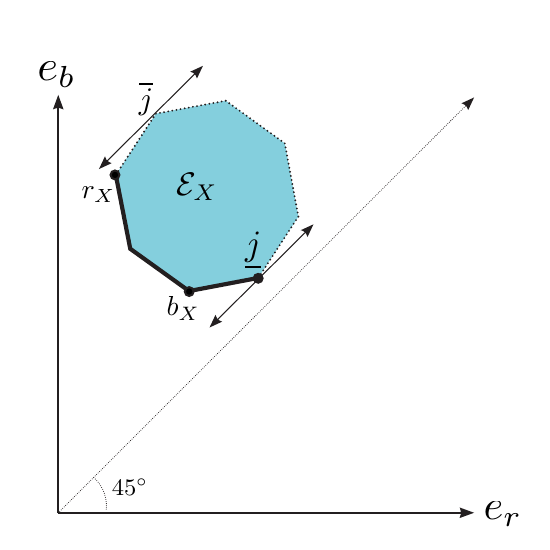}
        \caption{The points $\bar{j}$ and $\underline{j}$ in the feasible set $\mathcal{E}_{X}$ } \label{fig:charproof}
    \end{center} 
\end{figure}

We will first show that $\mathcal{F}_{X}\subset\mathcal{J}_{X}$.
Note that this is trivial if $\mathcal{E}_{X}$ is line segment so
suppose otherwise, and let $e\in\mathcal{F}_{X}$ but $e\not\in\mathcal{J}_{X}$.
This means that we can find $e^{\prime}=\left(e_{r}-\delta,e_{b}-\delta\right)\in\mathcal{E}_{X}$
for small $\delta>0$ and note that $e^{\prime}<e$ and $\left|e_{r}-e_{b}\right|=\left|e_{r}^{\prime}-e_{b}^{\prime}\right|$.
Thus, $e^{\prime}>_{FA}e$ so $e\not\in\mathcal{F}_{X}$ yielding
a contradiction.

Now, consider the group-balanced case where $e_{b}<e_{r}$ at $b_{X}$
and $e_{r}<e_{b}$ at $r_{X}$. We will show that $\mathcal{F}_{X}\subset\mathcal{P}_{X}$.
Let $e\in\mathcal{F}_{X}$ and since $\mathcal{F}_{X}\subset\mathcal{J}_{X}$,
$e\in\mathcal{J}_{X}$. First, suppose $e$ is on the boundary between
$\bar{j}$ and $r_{X}$ (not including $r_{X}$). In this case, we
can find small $\delta>0$ such that $\left(e_{r},e_{b}-\delta\right)\in\mathcal{E}_{X}$
but $\left(e_{r},e_{b}-\delta\right)>_{FA}e$ yielding a contradiction.
The case for $e$ is on the boundary between $b_{X}$ and $\underline{j}$
is symmetric. This implies $e\in\mathcal{P}_{X}$ as desired so $\mathcal{F}_{X}=\mathcal{P}_{X}$.
Note that this same argument applies to the group-balanced case where
$e_{r}=e_{b}$ at $r_{X}$ and $b_{X}$ so $r_{X}=b_{X}$. 

Next, consider the $r$-skewed case where $e_{r}<e_{b}$ at $r_{X}$
and $e_{r}\leq e_{b}$ at $b_{X}$. Let $e\in\mathcal{F}_{X}\subset\mathcal{J}_{X}$
and by the same argument above, we can rule out $e$ being on the
boundary between $\bar{j}$ and $r_{X}$ (not including $r_{X}$).
Moreover, we can also rule out $e$ being on the boundary between
$f_{X}$ and $\underline{j}$ (not including $f_{X}$) since $\left(e_{r}-\delta,e_{b}\right)\in\mathcal{E}_{X}$
but $\left(e_{r}-\delta,e_{b}\right)>_{FA}e$ for small $\delta$
(if $\mathcal{E}_{X}$ is a line segment, then just note that $f_{X}>_{FA}e$).
Thus, what remains to be shown is that all the boundary points between
$b_{X}$ and $f_{X}$ must be in $\mathcal{F}_{X}$. Let $e$ be any
such point. Since $e$ is on the boundary of $\mathcal{E}_{X}$ between
$b_{X}$ and $f_{X}$, we can find some line with slope between $0$
and $1$ such that all of $\mathcal{E}_{X}$ lies above this line.
If $e\not\in\mathcal{F}_{X}$, then there exists some $e^{\prime}\in\mathcal{E}_{X}$
such that $e^{\prime}>_{FA}e$. But this implies that $e^{\prime}$
must lie strictly beneath this line, yielding a contradiction. This
establishes that $\mathcal{F}_{X}$ is exactly the boundary between
$r_{X}$ and $f_{X}$ containing $\mathcal{P}_{X}$. The proof for
the $b$-skewed case is symmetric.

\subsection{Proof of Corollary \ref{corr:Tradeoff}}

First, suppose $X$ is group-balanced, so by Theorem \ref{thm:FullDesignPareto} $\mathcal{F}_X=\mathcal{P}_X$, so there cannot be two points $e,e^\prime \in \mathcal{F}_X$ such that $e>_{PD}e^\prime$. Now, suppose $X$ is $r$-skewed. Note that in this case, $b_{X}$
and $f_{X}$ must both be weakly above the 45-degree line. Let $e=b_{X}$
and $e^{\prime}=f_{X}$. By definition, we have $e_{b}\leq e_{b}^{\prime}$
and $e_{b}-e_{r}>e_{b}^{\prime}-e_{r}^{\prime}$ (where the strict
inequality follows from the fact that $b_{X}\not=f_{X}$). This implies
that
\[
e_{r}-e_{r}^{\prime}<e_{b}-e_{b}^{\prime}\leq0
\]
so $e_{r}<e_{r}^{\prime}$. Thus, $e>_{PD}e^{\prime}$ but $\left|e_{b}-e_{r}\right|>\left|e_{b}^{\prime}-e_{r}^{\prime}\right|$
as desired.

\subsection{Proof of Proposition \ref{prop:MinimalClass}} \label{app:Welfare}

We prove the following stronger result where the equivalence between (1) and (2) is the content of Proposition \ref{prop:MinimalClass}.
 
\begin{proposition} \label{prop:pareto_smallest}
The following are equivalent:
\begin{enumerate}
\item $e\in\mathcal{F}_X$
\item $e\in\mathcal{C}_{X}(\succeq)$ for some FA preference
$\succeq$
\item $\left\{ e\right\} =\mathcal{C}_{X}(\succeq)$ for some
FA preference $\succeq$ 
\item $e\in\mathcal{C}_{X}(\succeq)$ for some simple FA preference
$\succeq$
\end{enumerate}
\end{proposition}

This is stronger than Proposition \ref{prop:MinimalClass} in that it additionally shows that $\mathcal{F}_X$ is minimal in the sense that we cannot exclude any points from $\mathcal{F}_X$ without hurting some designer. This is because for every point $e\in\mathcal{F}_X$, there exists some FA preference $\succeq$ such that $e$ is the \emph{unique} optimal error pair given $\succeq$ within the feasible set $\mathcal{E}_X$. 

\begin{proof}
We will first show that (3) implies (2) implies (1) implies (3). Note
that (3) implies (2) is trivial. To see why (2) implies (1), suppose
$e\in\mathcal{C}_{X}(\succeq)$ for some FA preference
$\succeq$ but $e\not\in\mathcal{F}_X$. Thus, there exists
some $e^{\prime}>_{FA}e$ so $e^{\prime}\succ e$ yielding a contradiction.

We now prove that (1) implies (3). Fix some $e^{*}\in\mathcal{F}_X$
and let $\psi:\mathbb{R}\rightarrow\left(0,1\right)$ be a strictly decreasing
function. Define
\[
w\left(e\right)=\begin{cases}
1+\psi\left(e_{r}+e_{b}\right) & \text{if }e=e^{\ast}\text{ or }e>_{FA}e^{\ast}\\
\psi\left(e_{r}+e_{b}\right) & \text{otherwise}
\end{cases}
\]
and let $\succeq$ be the corresponding preference. We will show that
$\succeq$ is an FA preference. Suppose $e>_{FA}e^{\prime}$ so $e>_{PD}e'$ which implies $\psi\left(e_{r}+e_{b}\right)>\psi\left(e_{r}^{\prime}+e_{b}^{\prime}\right)$.
If both points FA-dominate $e^{\ast}$ or neither do, then $w\left(e\right)>w\left(e^{\prime}\right)$.
The only remaining case is when $e>_{FA}e^{\ast}$ but $e^{\prime}$
does not FA-dominate $e^{\ast}$, in which case 
\[
w\left(e\right)=1+\psi\left(e_{r}+e_{b}\right)>1>\psi\left(e_{r}^{\prime}+e_{b}^{\prime}\right)=w\left(e^{\prime}\right)
\]
Thus, $\succeq$ is an FA preference. Now, since $e^{*}\in\mathcal{F}_X$,
there exists no other $e\in\mathcal{E}_X$ such that $e>_{FA}e^{\ast}$.
That means that for all $e\in\mathcal{E}_X\backslash\left\{ e^{\ast}\right\} $,
$w\left(e^{\ast}\right)>w\left(e\right)$ so $\left\{ e^{\ast}\right\} =\mathcal{C}_{X}(\succeq)$.
This proves (3).

Finally, we show the equivalence of (1) and (4). Note that (4) implies
(2) which implies (1) from above. We now show that (1) implies (4).
Fix some $e^{*}\in\mathcal{F}_X$, so by Theorem \ref{thm:FullDesignPareto},
$e^{*}$ must either belong to the lower boundary from $r_{X}$ to
$b_{X}$ or the lower boundary from $b_{X}$ to $f_{X}$, where the
latter case only happens when $X$ is $r$-skewed (we omit the symmetric
situation when $X$ is $b$-skewed). If $e^{*}$ belongs to the boundary
from $r_{X}$ to $b_{X}$, then from the proof of Theorem \ref{thm:FullDesignPareto}
we know that $e^{*}$ belongs to an edge of this boundary that has
negative slope. Thus there exists a vector $(\gamma_{r},\gamma_{b})$
that is normal to this edge, such that $e^{*}$ maximizes $\gamma_{r}e_{r}+\gamma_{b}e_{b}$
among all feasible points. Since this edge has negative slope, it
is straightforward to see that $\gamma_{r},\gamma_{b}<0$. So $e$
maximizes the simple utility $\gamma_{r}e_{r}+\gamma_{b}e_{b}$ as
desired.

If instead $X$ is $r$-skewed and $e^{*}$ belongs to the boundary
from $b_{X}$ to $f_{X}$, then again $e^{*}$ belongs to an edge
of this boundary. But now this edge must have weakly positive slope
(since the edge starting from $b_{X}$ has weakly positive slope by
the definition of $b_{X}$, and since the boundary is convex). In
addition, this slope must be strictly smaller than $1$ because otherwise
$f_{X}$ would be farther away from the 45-degree line compared to
its adjacent vertex on this boundary. It follows that the outward
normal vector $(\beta_{r},\beta_{b})$ to the edge that $e^{*}$ belongs
to satisfies $\beta_{r}\geq0\geq-\beta_{r}>\beta_{b}$. The point
$e^{*}$ of interest maximizes $\beta_{r}e_{r}+\beta_{b}e_{b}$ among
all feasible points. Now let us choose any $\gamma_{f}$ to belong
to the interval $(\beta_{b},-\beta_{r})$, which is in particular
negative. Further define $\gamma_{r}=\beta_{r}+\gamma_{f}<0$ and
$\gamma_{b}=\beta_{b}-\gamma_{f}<0$. Then $\beta_{r}e_{r}+\beta_{b}e_{b}$
can be rewritten as $\gamma_{r}e_{r}+\gamma_{b}e_{b}+\gamma_{f}(e_{b}-e_{r})$.
If we consider the simple utility $\gamma_{r}e_{r}+\gamma_{b}e_{b}+\gamma_{f}\vert e_{b}-e_{r}\vert$,
then for any other feasible point $e^{**}$ it holds that 
\begin{align*}
\gamma_{r}e_{r}^{**}+\gamma_{b}e_{b}^{**}+\gamma_{f}\vert e_{b}^{**}-e_{r}^{**}\vert & \leq\gamma_{r}e_{r}^{**}+\gamma_{b}e_{b}^{**}+\gamma_{f}(e_{b}^{**}-e_{r}^{**})\\
 & =\beta_{r}e_{r}^{**}+\beta_{b}e_{b}^{**}\\
 & \leq\beta_{r}e_{r}^{*}+\beta_{b}e_{b}^{*}\\
 & =\gamma_{r}e_{r}^{*}+\gamma_{b}e_{b}^{*}+\gamma_{f}(e_{b}^{*}-e_{r}^{*})\\
 & =\gamma_{r}e_{r}^{*}+\gamma_{b}e_{b}^{*}+\gamma_{f}\vert e_{b}^{*}-e_{r}^{*}\vert,
\end{align*}
where the first inequality holds since $\gamma_{f}\leq0$ and the
last equality holds because $e^{*}\in\mathcal{F}_X$ must be weakly
above the 45-degree line. Hence the above inequality shows that $e^{*}$
maximizes the simple utility we have constructed, completing the proof.
\end{proof}

\subsection{Proof of Proposition \ref{prop:AlgorithmsGB}} 

Recall the proof of Lemma \ref{lemm:ConvexPolygon} where we showed
that the group errors can be written as
\begin{alignat*}{1}
e_{g}\left(a\right) & =\sum_{x}p_{x}\left(a\left(x\right)\sum_{y}\frac{x_{y,g}}{p_{g}}\ell\left(1,y\right)+\left(1-a\left(x\right)\right)\sum_{y}\frac{x_{y,g}}{p_{g}}\ell\left(0,y\right)\right)
\end{alignat*}
Letting
\[
c_{i}\left(x\right):=\left(\sum_{y}\frac{x_{y,r}}{p_{r}}\ell\left(i,y\right),\sum_{y}\frac{x_{y,b}}{p_{b}}\ell\left(i,y\right)\right)\in\mathbb{R}^{2}
\]
for $i\in\left\{ 0,1\right\} $, we can write the group error pair
as
\[
e\left(a\right)=\sum_{x}p_{x}\left(a\left(x\right)c_{1}\left(x\right)+\left(1-a\left(x\right)\right)c_{0}\left(x\right)\right)
\]
If we let $E_{x}\subset\mathbb{R}^{2}$
denote the line segment from $c_{0}\left(x\right)$ to $c_{1}\left(x\right)$,
then since $\mathcal{A}=\bar{\mathcal{A}}$, we have
\[
\mathcal{E}_{X}=\left\{ e\left(a\right)\text{ : }a\in\Delta\bar{\mathcal{A}}\right\} =\sum_{x}p_{x}E_{x}
\]
In other words, the feasible set is the Minkowski mixture of the line
segments $E_{x}$. Moreover, if we let $S\left(\mathcal{E},z\right)$
denote the support functional for any closed convex set $\mathcal{E}\subset\mathbb{R}^{2}$
and unit vector from the origin $z\in\mathbb{R}^{2}$, then by the
linearity of the support functional (see Theorem 1.7.5 of \cite{Schneider93}),
\begin{alignat}{1}
S\left(\mathcal{E}_{X},z\right) & =S\left(\sum_{x}p_{x}E_{x},z\right)=\sum_{x}p_{x}S\left(E_{x},z\right).\label{eq:support}
\end{alignat}

Consider the case where $X$ is group-balanced so $\mathcal{F}_{X}=\mathcal{P}_{X}$.
Order the (non-singleton) line segments $L_{1},L_{2},\dots$ in $\mathcal{F}_{X}$
as we trace out $\mathcal{F}_{X}=L_{1}\cup L_{2}\cup\cdots$ from
$r_{X}$ to $b_{X}$. Note that the slopes of $L_{i}$ are increasing
in $i$. 

First let $z_{0}$ denote the vector from
the origin to $\left(-1,0\right)$. Then by equation \eqref{eq:support},
$e(a_{r}^{\ast})\in S\left(\mathcal{E}_{X},z_0\right)=\sum_{x}p_{x}S\left(E_{x},z_0\right)$ so $e(a_{r}^{\ast})$ is a mixture of the left-most point of $E_{x}$
for all $x\in\mathcal{X}$. In other words, the $r$-optimal algorithm corresponds to minimizing group $r$'s error for each realization  $x \in \mathcal{X}$.

Next, consider the first segment $L_{1}$ and without loss,
assume it is non-singleton. Let $z_{1}\in\mathbb{R}^{2}$ be the unit
vector from the origin such that $L_{1}=S\left(\mathcal{E}_{X},z_{1}\right)$.
Then again by equation \eqref{eq:support}, there must exist some $x\in\mathcal{X}$ such that $S\left(E_{x},z_{1}\right)$
is not a singleton (as otherwise $L_1$ would be a singleton). Moreover, the slope of any such $E_{x}$
must be the same as the slope of $L_{1}$. We now argue that  $E_{x}$
must have the most negative slope among all $E_{x'}$, $x' \in \mathcal{X}$. To see why, suppose towards contradiction that some other $E_{x^{\prime}}$ has a strictly more negative slope. This implies that $S\left(E_{x'},z_{1}\right)$ is the right-most point of the segment $E_{x'}$ (recall that the vector $z_1$ is orthogonal to the segment $E_x$). Now
$e\left(a_{r}^{\ast}\right)\in L_{1}=S\left(\mathcal{E}_{X},z_{1}\right)$
so by equation \eqref{eq:support}, $a_{r}^{\ast}\left(x^{\prime}\right)$
must correspond to the right-most point of $E_{x^{\prime}}$. However, this contradicts with our observation above that $a_{r}^*$ has support only on the left-most points of $(E_x)_{x \in \mathcal{X}}$. This shows that $L_{1}$ corresponds to $E_{x_1}$ (i.e. the most negative slope) and thus is the segment
from $r_{X}=e\left(a_{r}^{\ast}\right)$ to $e\left(a_{1}\right)$.

We now prove the result by induction, paralleling the argument above. Consider two consecutive line segments $L_{i}$ and $L_{i+1}$
and let $e\left(a_{i}\right)$ be the shared point between
the two segments. Let $z_{i}$ and $z_{i+1}$ be the unit vectors from the
origin such that $L_{i}=S\left(\mathcal{E}_{X},z_{i}\right)$ and $L_{i+1}=S\left(\mathcal{E}_{X},z_{i+1}\right)$.
As before, there must exist $x\in\mathcal{X}$ such that the slope
of $E_{x}$ is the same as the slope of $L_{i+1}$. Moreover,
we will now show that the slope of $E_{x}$ must be the most negative slope out
of the remaining $E\left(x_{i+1}\right),E\left(x_{i+2}\right),\dots$.
Suppose otherwise so there exists some $E_{x^{\prime}}$ with an 
even more negative slope. By the induction argument, the slope of $E_{x^{\prime}}$
must be strictly between $E\left(x_{i}\right)$ and $E\left(x_{i+1}\right)$.
Now, $e\left(a_{i}\right)\in L_{i+1}=S\left(\mathcal{E}_{X},z_{i+1}\right)$
so by equation \eqref{eq:support}, $a_{i}\left(x^{\prime}\right)$
corresponds to the right-most point of $E_{x^{\prime}}$.
However, we also have $e\left(a_{i}\right)\in L_{i1}=S\left(\mathcal{E}_{X},z_{i}\right)$
by induction, so $a_{i}\left(x^{\prime}\right)$ corresponds
to the left-most point of $E_{x^{\prime}}$. This yields
a contradiction. By induction, $a_{1},a_{2},\dots$ are exactly
the algorithms that generate $\mathcal{F}_{X}=L_{1}\cup L_{2}\cup\cdots$.

Finally, consider the case where $X$ is $r$-skewed and let $\mathcal{F}_{X}=L_{1}\cup L_{2}\cup\cdots$ which starts at $r_X$ and ends at $f_X$. By the same argument as above,  $a_{1},a_{2},\dots$ are exactly
the algorithms that generate $\mathcal{F}_{X}=L_{1}\cup L_{2}\cup\cdots$ bearing in mind that the point $f_X$ may correspond to an interior point of some $L_i$. This concludes the proof.

\subsection{Proof of Proposition \ref{prop:XRevealsG}}

Recall that we showed in the proof of Lemma \ref{lemm:ConvexPolygon}
that group errors can be written as
\begin{alignat*}{1}
e_{g}\left(a\right) & =\sum_{x}p_{x}\left(a\left(x\right)\sum_{y}\frac{x_{y,g}}{p_{g}}\ell\left(1,y\right)+\left(1-a\left(x\right)\right)\sum_{y}\frac{x_{y,g}}{p_{g}}\ell\left(0,y\right)\right)
\end{alignat*}
Letting
\[
c_{i}\left(x\right):=\left(\sum_{y}\frac{x_{y,r}}{p_{r}}\ell\left(i,y\right),\sum_{y}\frac{x_{y,b}}{p_{b}}\ell\left(i,y\right)\right)\in\mathbb{R}^{2}
\]
for $i\in\left\{ 0,1\right\} $, we can write the group error pair
as
\[
e\left(a\right)=\sum_{x}p_{x}\left(a\left(x\right)c_{1}\left(x\right)+\left(1-a\left(x\right)\right)c_{0}\left(x\right)\right)
\]
Let $\mathcal{X}_{g}\subset\mathcal{X}$ be the realizations of the
covariate that reveal group identity $g$. We thus have
\begin{alignat*}{1}
\mathcal{E}_{X}= & \left\{ e\left(a\right)\text{ : }a\in\Delta\mathcal{A}\right\} \\
= & \left\{ \sum_{x\in\mathcal{X}_{r}}p_{x}\left(a_{r}\left(x\right)c_{1}\left(x\right)+\left(1-a_{r}\left(x\right)\right)c_{0}\left(x\right)\right)\right.\\
 & +\left.\sum_{x\in\mathcal{X}_{b}}p_{x}\left(a_{b}\left(x\right)c_{1}\left(x\right)+\left(1-a_{b}\left(x\right)\right)c_{0}\left(x\right)\right)\text{ : }a_{r}\in\Delta\mathcal{A}_{r},a_{b}\in\Delta\mathcal{A}_{b}\right\} 
\end{alignat*}
Note that this is the Minkowski addition of two sets. Since $x_{y,b}=0$
for all $x\in\mathcal{X}_{r}$, $c_{1}\left(x\right)$ and $c_{0}\left(x\right)$
are all points on the vertical axis for all $x\in\mathcal{X}_{r}$.
By symmetric reasoning, $c_{1}\left(x\right)$ and $c_{0}\left(x\right)$
are all points on the horizonal axis for all $x\in\mathcal{X}_{b}$.
Since $\mathcal{E}_{X}$ is the Minkowski addition of a set on the
vertical axis and a set on the horizonal axis, it must be a rectangle.
Moreover, $r_{X}=b_{X}$ must be its bottom-left vertex.

Finally, suppose without loss of generality that $r_X = b_X$ lies above the 45-degree line. If the rectangle $\mathcal{E}_X$ does not intersect the 45-degree line, then it is easy to see that $f_{X}$ must be the bottom-right vertex of $\mathcal{E}_X$. In this case the fairness-accuracy frontier is the entire bottom edge of the rectangle, which is a horizontal line segment. If instead the rectangle $\mathcal{E}_X$ intersects the 45-degree line, then $f_{X}$ is the intersection between the bottom edge of $\mathcal{E}_X$ and the 45-degree line. Again the fairness-accuracy frontier is the horizontal line segment from $r_X = b_X$ to $f_{X}$. This proves the result.

\subsection{Proof of Proposition \ref{prop:SpecialCase}}

Part (a):  We continue to follow the notation laid out in the
proof of Lemma \ref{lemm:ConvexPolygon}. Note that under strong independence,
\begin{align*}
\frac{x_{y,r}}{x_{y,b}} & =\frac{\mathbb{P}(Y=y,G=r\mid X=x)}{\mathbb{P}(Y=y,G=b\mid X=x)}\\
 & =\frac{\mathbb{P}(G=r\mid Y=y,X=x)}{\mathbb{P}(G=b\mid Y=y,X=x)}=\frac{p_{r}}{p_{b}}.
\end{align*}
Thus $\frac{x_{y,r}}{p_{r}}=\frac{x_{y,b}}{p_{b}}$ for all $x,y$.
It follows that the line segment $E(x)$, which connects the two points
$\left(\sum_{y}\frac{x_{y,r}}{p_{r}}\ell\left(1,y\right),\sum_{y}\frac{x_{y,b}}{p_{b}}\ell\left(1,y\right)\right)$
and $\left(\sum_{y}\frac{x_{y,r}}{p_{r}}\ell\left(0,y\right),\sum_{y}\frac{x_{y,b}}{p_{b}}\ell\left(0,y\right)\right)$,
lies on the 45-degree line. Therefore $\mathcal{E}_X=\sum_{x}E_{x}\cdot p_{x}$
is also on the 45-degree line.

Part (b):  We will show that $b_{X}=r_{X}$ under conditional
independence, and the result then follows from Theorem \ref{thm:FullDesignPareto}. Recall from the proof of Lemma \ref{lemm:ConvexPolygon}
that
\[
\mathcal{E}_X=\sum_{x\in\mathcal{X}}E_{x}p_{x}
\]
where
\begin{align*}
E_{x}= & \left\{ \lambda\left(\sum_{y}\frac{x_{y,r}}{p_{r}}\ell\left(1,y\right),\sum_{y}\frac{x_{y,b}}{p_{b}}\ell\left(1,y\right)\right)\right.\\
 & \left.+\left(1-\lambda\right)\left(\sum_{y}\frac{x_{y,r}}{p_{r}}\ell\left(0,y\right),\sum_{y}\frac{x_{y,b}}{p_{b}}\ell\left(0,y\right)\right)\text{ : }\lambda\in\left[0,1\right]\right\} 
\end{align*}
Under conditional independence, $x_{y,g}=x_{y}x_{g}$ (with $x_y := P(Y=y \mid X=x)$ and $x_g := P(G=g \mid X=x)$) so we have 
\[
E_{x}=\left\{ \left(\lambda\sum_{y}x_{y}\ell\left(1,y\right)+\left(1-\lambda\right)\sum_{y}x_{y}\ensuremath{\ell}\left(0,y\right)\right)\left(\frac{x_{r}}{p_{r}},\frac{x_{b}}{p_{b}}\right)\text{\text{ : }\ensuremath{\lambda\in\left[0,1\right]}}\right\} 
\]
 
This means that for each realization $x\in\mathcal{X}$, the outcome
that gives the lower error for group $r$ also gives the lower error
for group $b$. In other words, when $\sum_{y}x_{y}\ell\left(1,y\right)\leq\sum_{y}x_{y}\ensuremath{\ell}\left(0,y\right)$,
then the decision $d=1$ is optimal for both groups (and vice-versa if the inequality is reversed). Consider the following algorithm:
\[
a\left(x\right)=\begin{cases}
1 & \text{if }\sum_{y}x_{y}\ell\left(1,y\right)\leq\sum_{y}x_{y}\ensuremath{\ell}\left(0,y\right)\\
0 & \text{if }\sum_{y}x_{y}\ell\left(1,y\right)>\sum_{y}x_{y}\ensuremath{\ell}\left(0,y\right)
\end{cases}
\]
This algorithm will deliver the lowest error for both groups and
\[
(e_r\left(a\right),e_b\left(a\right))=r_{X}=b_{X}
\]
as desired.

\section{Proofs and Supporting Materials to Section \ref{sec:DesignInputs}}

\subsection{Proof of Lemma \ref{lemm:BayesDesign}}

We first characterize the input-design feasible set, and later study the input-design fairness-accuracy frontier. It is clear that regardless of what garbling the designer gives the agent, the agent's payoff will be weakly better than what can be achieved under no information. Thus any error pair that is implementable by input-design must belong to the halfspace $H$. Such an error pair must also belong to the feasible set $\mathcal{E}_X$, so we obtain the easy direction $\mathcal{E}^{*}_{X} \subseteq \mathcal{E}_X \cap H$ in the lemma. 

Conversely, we need to show that a feasible error pair $(e_r, e_b) \in \mathcal{E}_X$ that satisfies $\alpha_r e_r + \alpha_b e_b \leq e_0$ can be implemented by some garbling $T$. Consider a garbling $T$ that maps $X$ to $\Delta(\mathcal{D})$, with the interpretation that the realization of $T(x)$ is the recommended decision for the agent. If we abuse notation to let $a(x)$ denote the probability that the recommendation is $d = 1$ at covariate vector $x$, then this algorithm $a$ needs to satisfy the following obedience constraint for $d = 1$:\footnote{By a version of the revelation principle, such garblings together with the following obedience constraints are without loss for studying the feasible decisions, in a general setting.} 
\[
\sum_{y, g} \frac{\alpha_g}{p_g} \sum_{x} p_{x,y,g} \cdot a(x) \cdot \ell(1,y) \leq \sum_{y, g} \frac{\alpha_g}{p_g} \sum_{x} p_{x,y,g} \cdot a(x) \cdot \ell(0,y).
\]
The above is just equation \eqref{eq:choiceDM} adapted to the current setting with the observation that given the recommendation $T = 1$, the conditional probability of $Y = y$ and $G = g$ is proportional to the recommendation probability $\sum_{x} p_{x,y,g} \cdot a(x)$, where we use $p_{x, y, g}$ as a shorthand for $\mathbb{P}(X = x, Y = y, G = g)$. 

Let us rewrite the above displayed inequality as 
\[
\sum_{x,y,g} p_{x,y,g}  \frac{\alpha_g}{p_g} \cdot a(x) \ell(1,y) \leq \sum_{x,y,g} p_{x,y,g}  \frac{\alpha_g}{p_g} \cdot a(x) \ell(0,y).
\]
If we add $p_{x,y,g}  \frac{\alpha_g}{p_g}(1-a(x)) \ell(0,y)$ to each summand above, we obtain
\begin{equation}\label{eq:obed_a=1}
\sum_{x,y,g} p_{x,y,g}  \frac{\alpha_g}{p_g} \cdot \left(a(x)\ell(1,y) + (1-a(x))\ell(0,y)\right) \leq \sum_{x,y,g} p_{x,y,g} \frac{\alpha_g}{p_g} \cdot \ell(0,y).
\end{equation}
Now, the LHS above can be rewritten as $\sum_{x,y,g} p_{x,y,g}\frac{\alpha_g}{p_g} \cdot \mathbb{E}_{D \sim a(x)} [\ell(D,y) \mid X = x, Y = y, G = g]$, which is also equal to $\sum_{g} \alpha_g \cdot \mathbb{E}_{D \sim a(x)} [\ell(D,Y) \mid G = g]$. This is precisely the agent's expected loss when following the designer's recommended decisions. 

On the other hand, the RHS in \eqref{eq:obed_a=1} can be seen to be the agent's expected loss when taking the decision $d = 0$ regardless of the designer's recommendation. Thus, we deduce that the obedience constraint for the recommendation $d=1$ is equivalent to \eqref{eq:obed_a=1}, which simply says that the agent's payoff under the designer's recommendation should be weakly better than the constant decision $d = 0$ ignoring the recommendation. Symmetrically, the other obedience constraint for the recommendation $d = 0$ is equivalent to the agent's payoff being better than the constant decision $d = 1$. Put together, these obedience constraints thus reduce to the requirement that the designer's recommendation gives the agent a payoff that exceeds what can be achieved with no information. 

For any error pair $(e_r, e_b)$ that is feasible under unconstrained design, we can construct a garbling $T$ that implements it by recommending the desired decision. If $(e_r, e_b)$ belongs to the halfspace $H$, then by the previous analysis we know that obedience is satisfied. Thus $(e_r, e_b)$ is implementable under input-design, showing that $\mathcal{E}_X \cap H = \mathcal{E}^{*}_{X}$ as desired.

Finally we turn to the fairness-accuracy frontier and argue that $\mathcal{F}^{*}_{X} = \mathcal{F}_X \cap H$. In one direction, if an error pair is FA-undominated in $\mathcal{E}_X$ and implementable under input design, then it is also FA-undominated in the smaller set $\mathcal{E}^{*}_{X}$. This proves $\mathcal{F}_X \cap H \subseteq \mathcal{F}^{*}_{X}$. In the opposite direction, suppose for contradiction that a certain point $(e_r, e_b) \in \mathcal{F}^{*}_{X}$ does not belong to $\mathcal{F}_X \cap H$. Since $\mathcal{F}^{*}_{X} \subseteq \mathcal{E}^{*}_{X} \subseteq H$, we know that $(e_r, e_b)$ must not belong to $\mathcal{F}_X$. Thus by definition of $\mathcal{F}_X$, $(e_r, e_b)$ is FA-dominated by some other error pair $(\widehat{e_r}, \widehat{e_b}) \in \mathcal{E}_X$. In particular, we must have $\widehat{e_r} \leq e_r$ and $\widehat{e_b} \leq e_b$, which implies $\alpha_r \widehat{e_r} + \alpha_b \widehat{e_b} \leq \alpha_r e_r + \alpha_b e_b \leq e_0$ (the first inequality uses $\alpha_r, \alpha_b > 0$ and the second uses $(e_r, e_b) \in \mathcal{F}^{*}_{X} \subseteq\mathcal{E}^{*}_{X}$). It follows that the FA-dominant point $(\widehat{e_r}, \widehat{e_b})$ also belongs to $H$ and thus $\mathcal{E}^{*}_{X}$. But this contradicts the assumption that $(e_r, e_b)$ is FA-undominated in $\mathcal{E}^{*}_{X}$. Such a contradiction completes the proof.

\subsection{Supplementary Material to Section \ref{sec:SimpleExample}} \label{app:ExclusionOptimal}

In this section, we compute the input-design feasible set and fairness-accuracy frontier for the example in Section \ref{sec:SimpleExample}. Since $X$ is a null signal, garblings of $(X, X')$ are the same as garblings of $X'$. Without loss, we can restrict attention to garblings of $X'$ that take two values, $d=1$ and $d=0$, which correspond to the designer's decisions for the agent. Any such garbling can be identified with a pair $(\alpha,\beta)$, where $\alpha$ is the probability with which $X'=1$ is mapped into $d=1$, and $\beta$ is the probability with which $X'=0$ is mapped into $d=1$. It is easy to check that the agent's obedience constraint reduces to the simple inequality $\alpha \geq \beta$, which intuitively requires the agent to choose $d=1$ more often when $X' = 1$.

For any pair $(\alpha, \beta)$, the two groups' errors can be calculated as 
\[
e_r(\alpha,\beta) = \frac{1}{2}(1-\alpha) + \frac{1}{2}\beta = 0.5 - 0.5(\alpha-\beta),
\] 
\[
e_b(\alpha,\beta) = \frac{1}{2}\cdot 0.6(1-\alpha) + \frac{1}{2} \cdot 0.4(1-\beta) + \frac{1}{2}\cdot 0.4\alpha + \frac{1}{2} \cdot 0.6\beta = 0.5 - 0.1(\alpha-\beta).
\]
So as $\alpha - \beta$ ranges from $0$ to $1$, the implementable group errors constitute the line segment connecting $(0, 0.4)$ with $(0.5, 0.5)$. This entire line segment is also the fairness-accuracy frontier $\mathcal{F}^{*}_{X,X'}$, as illustrated in Figure \ref{fig:example} in the main text.

For an Egalitarian designer, sending the null signal $X$ leads to the point $(0.5, 0.5)$ and yields a payoff of 0. In contrast, we say that the designer ``makes use of $X'$ over $X$" if the garbling $T$ is \emph{not} independent of $X'$ conditional on $X$ (in this example the conditioning is irrelevant since $X$ is null). Whenever $T$ is not independent of $X'$, then for some realizations of $T$ the agent believes $X' = 1$ is more likely, which makes $d=1$ strictly optimal. Thus, whenever the designer makes use of $X'$ in the garbling, the agent is strictly better off compared to the null signal, and the resulting error pair must be distinct from $(0.5, 0.5)$. But given the shape of the implementable set, this means that the designer is strictly worse off when any information about $X'$ is provided to the agent.

\subsection{Proof of Proposition \ref{thm:FullVersusBayes}} 

We now deduce Proposition \ref{thm:FullVersusBayes} from Lemma \ref{lemm:BayesDesign}. If $X$ is group-balanced, then by Theorem \ref{thm:FullDesignPareto} we know that $\mathcal{F}_X$ is the part of the boundary of $\mathcal{E}_X$ that connects $r_X$ to $b_X$ from below. Clearly, $\mathcal{F}^{*}_{X} = \mathcal{F}_X$ can only hold if $r_X, b_X \in \mathcal{F}^{*}_{X} \subseteq H$, so we focus on the ``if'' direction of the result. Suppose $r_X, b_X \in H$, then we claim that the entire lower boundary of $\mathcal{E}_X$ from $r_X$ to $b_X$ belongs to $H$. Indeed, let $m_X$ be the point that has the same $e_r$ as $r_X$ and the same $e_b$ as $b_X$. Geometrically, $m_X$ is such that the line segments $m_Xr_X$ and $m_Xb_X$ are parallel to the axes. Because $r_X, b_X$ have respectively minimal group errors in the feasible set $\mathcal{E}_X$, and because we are considering the lower boundary, any point on this lower boundary $\mathcal{F}_X$ must belong to the triangle with vertices $r_X, b_X$ and $m_X$. Since $r_X, b_X, m_X$ all belong to the halfspace $H$ ($m_X \in H$ because the agent's payoff weights $\alpha_r, \alpha_b$ are non-negative), we deduce that $\mathcal{F}_X \subseteq H$. Hence whenever $r_X, b_X \in H$, we have by Lemma \ref{lemm:BayesDesign} that $\mathcal{F}^{*}_{X} = \mathcal{F}_X \cap H = \mathcal{F}_X$. This argument proves Proposition \ref{thm:FullVersusBayes} in the group-balanced case. 

Suppose instead that $X$ is $r$-skewed (a symmetric argument applies to the $b$-skewed case). To generalize the above argument, we need to show that whenever $r_X, f_{X}$ belong to $H$, then so does the entire lower boundary connecting these points. To see this, note that by the definition of $b_X$ and $f_{X}$, the lower boundary connecting these two points consists of positively sloped edges.\footnote{If we start from $b_X$ and traverse the lower boundary to the right until $f_{X}$, then the first edge of this boundary must be positively sloped because $b_X$ has minimum $e_b$. The final edge of this boundary must also be positively sloped, since otherwise the starting vertex of this edge would be closer to the 45-degree line than $f_{X}$. It follows by convexity that the entire boundary from $b_X$ to $f_{X}$ has positive slopes.} So across all points on this part of the lower boundary, $f_{X}$ maximizes $\alpha_r e_r + \alpha_b e_b$. Thus the assumption $f_{X} \in H$ implies that the lower boundary from $b_X$ to $f_{X}$ belongs to $H$. In particular $b_X \in H$, which together with $r_X \in H$ implies that the lower boundary from $r_X$ to $b_X$ also belongs to $H$ (by the same argument as in the group-balanced case before). Hence the entire lower boundary from $r_X$ to $f_{X}$ belongs to $H$, as we desire to show.

\subsection{Proof of Proposition \ref{prop:ExcludeG}}

We first present a simple lemma which conveniently restates the property of FA-dominance for frontiers:

\begin{lemma}\label{lemm:UniformWorsen}
$\mathcal{F}^*_{X,X'} >_{FA} \mathcal{F}^*_X$ if and only if $\mathcal{F}^{*}_{X}$ does not intersect with $\mathcal{F}^{*}_{X,X'}$.
\end{lemma}

The proof of this lemma is straightforward: If there exists a point in $\mathcal{F}^{*}_{X}$ that also belongs to $\mathcal{F}^{*}_{X,X'}$, then this point is not FA-dominated by any point in $\mathcal{F}^{*}_{X,X'}$. On the other hand, suppose no point in $\mathcal{F}^{*}_{X}$ belongs to $\mathcal{F}^{*}_{X,X'}$. Note that any point in $\mathcal{F}^{*}_{X}$ is implementable via a garbling of $X$ and thus implementable via a garbling of $X, X'$. Thus any such point belongs to $\mathcal{E}^{*}_{X,X'}$, and since it is not FA-optimal in this set, it must be FA-dominated by some FA-optimal point in this (compact) set, as we desire to show.

Below we use Lemma \ref{lemm:UniformWorsen} to deduce Proposition \ref{prop:ExcludeG}. The key observation is that whether or not $G$ is excluded does not affect the minimal (or maximal) feasible error for either group. This is because if we want to minimize the error of a particular group $g$ using an algorithm that depends on $X$, then we essentially condition on $G = g$ anyways. 

With this observation, suppose $X$ is strictly group-balanced. Then $r_X$ lies strictly above the 45-degree line and $b_X$ lies strictly below. Since we assume $r_X, b_X \in H$, Proposition \ref{thm:FullVersusBayes} tells us that the input-design fairness-accuracy frontier $\mathcal{F}^{*}_{X}$ is the same as the unconstrained fairness-accuracy frontier $\mathcal{F}_X$, and by Theorem \ref{thm:FullDesignPareto} this frontier is the lower boundary of the feasible set $\mathcal{E}_X$ connecting $r_X$ to $b_X$. By Lemma \ref{lemm:UniformWorsen}, we just need to show that in this case the lower boundary of $\mathcal{E}_X$ from $r_X$ to $b_X$ does not intersect with the input-design fairness-accuracy frontier $\mathcal{F}^{*}_{X,G}$ given $(X, G)$. To characterize the latter frontier, let $m_X = r_{X, G} = b_{X, G}$ denote the error pair that has the same $e_r$ as $r_X$ and the same $e_b$ as $b_X$. Without loss of generality assume $m_X$ lies weakly above the 45-degree line. Then from Proposition \ref{prop:XRevealsG} we know that the unconstrained fairness-accuracy frontier $\mathcal{F}_{X,G}$ is the horizontal line segment from $m_X$ to $f_{X,G}$. This point $f_{X, G}$ is the intersection between the line segment $m_Xb_X$ and the 45-degree line (here we use the fact that $m_X$ lies above the 45-degree line and $b_X$ lies below). As $b_X \in H$, the points $m_X$ and $f_{X, G}$ also belong to $H$ because they have equal $e_b$ and smaller $e_r$ compared to $b_X$. Hence the input-design fairness-accuracy frontier $\mathcal{F}^{*}_{X,G}$ is also the line segment from $m_X$ to $f_{X, G}$. To see that this horizontal line segment does not intersect the boundary of $\mathcal{E}_X$ from $r_X$ to $b_X$, just note that $b_X$ is the only point on that boundary with the same (minimal) $e_b$ as any point on the horizontal line segment. But $b_X$ does not belong to that line segment because it is strictly below the 45-degree line. This proves the result when $X$ is strictly group-balanced. 

Now suppose $X$ is not strictly group-balanced. Then $r_X$ and $b_X$ lie weakly on the same side of the 45-degree line, and without loss of generality let us assume they lie weakly above. It is still the case that the unconstrained fairness-accuracy frontier $\mathcal{F}_{X,G}$ is the horizontal line segment from $m_X$ to $f_{X, G}$. But in the current setting $f_{X, G}$ must be weakly closer to the 45-degree line than $b_X$, which means that $b_X$ now lies in between $m_X$ and $f_{X, G}$. In other words, $b_X \in \mathcal{F}_X$ and $b_X \in \mathcal{F}_{X,G}$. But by assumption, $b_X$ also belongs to $H$. So Lemma \ref{lemm:BayesDesign} tells us that $b_X$ belongs to the input-design fairness-accuracy frontiers $\mathcal{F}^{*}_{X}$ and  $\mathcal{F}^{*}_{X,G}$. This shows that the two frontiers $\mathcal{F}^{*}_{X}$ and $\mathcal{F}^{*}_{X,G}$ intersect, which completes the proof by Lemma \ref{lemm:UniformWorsen}. 

\subsection{Proof of Proposition \ref{prop:XGSimple}}

By Proposition \ref{prop:XRevealsG}, the fairness-accuracy frontier $\mathcal{F}_{X,G}$ is a line segment whose left endpoint is $r_{X,G}=b_{X,G}$ and right endpoint is $f_{X,G}$, with $e_b$ being constant on this frontier. 
We will first identify the optimal points for the designer in the unconstrained problem (i.e., if the designer were given full control over the algorithm, as in our Section \ref{sec:model}), and then show that the garblings indicated in the result implement these points.

Suppose first that $(X,G)$ is group-balanced. Then $\mathcal{F}_{X,G}$ is the singleton $r_{X,G}=b_{X,G}=f_{X,G}$, which is optimal for all FA preferences. So clearly this point is the designer's preferred point given any simple preference.

Now suppose $(X,G)$ is $r$-skewed. Then the frontier $\mathcal{F}_{X,G}$ consists entirely of points $(e_r,e_b)$ satisfying $e_b \geq e_r$. So maximizing the designer's payoff $-\gamma_re_r -\gamma_b e_b -\gamma_f \vert e_r-e_b\vert$ over error pairs on the FA frontier is equivalent to solving
\[\max_{(e_r,e_b) \in \mathcal{F}_{X,G}} (\gamma_f - \gamma_r) e_r - (\gamma_b + \gamma_f) e_b.\]
Since $e_b$ is constant on $\mathcal{F}_{X,G}$, the designer simply maximizes $(\gamma_f - \gamma_r) e_r$. If $\gamma_f < \gamma_r$, then the designer minimizes $e_r$ and thus $r_{X,G}=b_{X,G}$ is the uniquely optimal point for the designer. If $\gamma_f=\gamma_r$, then all points on $\mathcal{F}_{X,G}$ are optimal. Finally if $\gamma_f > \gamma_r$, then the designer's weight on $e_r$ is strictly positive, so that $f_{X,G}$ is the uniquely optimal point for the designer.

It remains to determine how the designer can implement these points with garblings. Again start with the group-balanced case. Sending the agent the fully revealing garbling means that the agent's feasible set is $\mathcal{E}_{X,G}$. Since by assumption the agent's preferences can be written as $-\alpha_r e_r - \alpha_b e_b$ for $\alpha_r,\alpha_b >0$, the agent optimally chooses $r_{X,G}=b_{X,G}$ as desired by the designer. This yields part (a) of the result. Moreover, when $(X,G)$ is group-balanced,  
\[\beta= \max\left\{\frac{\overline{e}_r - e_b^R}{\overline{e}_r - e_r^R}, ~~0 \right\}  =1\] since $e_r^R = e_b^R$ at $r_{X,G} = b_{X,G} = f_{X,G}$. Thus the $r$-shaded garbling maps each $(x,g)$ to the message $a_g^*(x)$ with probability 1, which the agent optimally obeys. Thus again the agent optimally chooses $r_{X,G}$.

Now consider the $r$-skewed case. By the same logic as for the group-balanced case, sending the agent the fully revealing signal leads the agent to choose $r_{X,G}$, which is optimal for the designer when $\gamma_f \leq \gamma_r$. This yields part (a) of the result. Part (b) of the result follows by noticing that the $r$-shaded garbling is precisely the garbling constructed in the proof of Lemma \ref{lemm:BayesDesign}, which recommends an action to the agent and leads to the error pair $(e_b^R, \min\{e_b^R, \overline{e}_r\})$ assuming that the agent follows the recommendation. By Proposition \ref{prop:XRevealsG} this error pair is the point $f_{X,G}$, which we assume belongs to $H$. So by the proof of Lemma \ref{lemm:BayesDesign}, the agent optimally follows the recommendation and chooses $f_{X,G}$, which is optimal for designers with $\gamma_f \geq \gamma_r$. This completes the proof.

\subsection{Proof of Proposition \ref{prop:ExcludeX'overG}}

Let
$
    \underline{e}_g = \min \{e_g \mid e \in \mathcal{E}_{X,G}\}$ and $
    \overline{e}_g = \max \{e_g \mid e \in \mathcal{E}_{X,G}\}$
be the minimal and maximal feasible errors for group $g$ given covariate vector $(X,G)$, and define
$\underline{e}^*_g  = \min \{e_g \mid e \in \mathcal{E}_{X,G,X'}\}$ and $
    \overline{e}^*_g = \max \{e_g \mid e \in \mathcal{E}_{X,G,X'}\}$
to be the corresponding quantities given $(X,G,X')$. 
The following lemma says that additional access to $X'$ reduces the minimal feasible error for group $g$ relative to $(X,G)$ if and only if $X'$ is decision-relevant over $X$ for group $g$.

\begin{lemma} \label{lemm:DecisionRelevant} $\underline{e}^*_g < \underline{e}_g$ if $X'$ is decision-relevant over $X$ for group $g$, and $\underline{e}^*_g = \underline{e}_g$ if it is not.
\end{lemma}

\begin{proof} Let $a_g: \mathcal{X} \rightarrow \{0,1\}$ be any strategy mapping each realization of $X$ into an  optimal outcome for group $g$, i.e., 
\[
a_g(x) \in \argmin_{d\in \{0,1\}} \mathbb{E}\left[\ell(d,Y) \mid G=g,X=x)\right] \quad \forall x \in \mathcal{X}.
\]
Likewise let $a^*_g: \mathcal{X} \times \mathcal{X}' \rightarrow \{0,1\}$ satisfy 
\[
a^*_g(x,x') \in \argmin_{d\in \{0,1\}} \mathbb{E}\left[\ell(d,Y) \mid G=g, X=x, X'=x')\right] \quad \forall x\in \mathcal{X}, ~\forall x'\in \mathcal{X}'.
\]
By optimality of $a^*_g$, for all  $x\in \mathcal{X}$ and $x'\in \mathcal{X}'$,
\begin{align} \label{eq:compare}
    \mathbb{E}\left[ \ell(a^*_g(x,x'), Y) \right. & \left. \mid G=g, X=x, X'=x' \right]  \leq \mathbb{E}\left[\ell(a_g(x),Y) \mid G=g,X=x,X=x'\right] 
\end{align}
Suppose $X'$ is decision-relevant over $X$ for group $g$. Then there exist $x\in \mathcal{X}$ and $x',\tilde{x}'\in \mathcal{X}'$ such that the optimal assignment for group $g$ is uniquely equal to 1 at $(x,x')$ and 0 at $(x,\tilde{x}')$, where both $(x,x')$ and $(x,\tilde{x}')$ have positive probability conditional on $G=g$. But then (\ref{eq:compare}) must hold strictly at either $(x,x')$ or $(x,\tilde{x}')$. By taking the expectation of \eqref{eq:compare} conditional on $G = g$, we obtain
\begin{align*}
    \underline{e}^*_g & = \mathbb{E}\left[ \ell(a^*_g(X,X'), Y) \mid G=g \right] < \mathbb{E}\left[\ell(a_g(X),Y) \mid G=g\right] = \underline{e}_g.
\end{align*}
If $X'$ is not decision-relevant over $X$ for group $g$, then (\ref{eq:compare}) holds with equality at every $x,x'$, and the equivalence $\underline{e}^*_g=\underline{e}_g$ follows.
\end{proof}

We now use Lemma \ref{lemm:UniformWorsen} and \ref{lemm:DecisionRelevant} to prove Proposition \ref{prop:ExcludeX'overG}. First suppose $(X,G)$ is $r$-skewed, in which case $r_X = b_X$ lies strictly above the 45-degree line. By Proposition \ref{prop:XRevealsG}, the unconstrained fairness-accuracy frontier $\mathcal{F}_{X,G}$ is then the horizontal line segment from $r_{X,G} = b_{X,G}$ to $f_{X,G}$. 

If $X'$ is not decision-relevant over $X$ for group $b$, then from Lemma \ref{lemm:DecisionRelevant} we know that the minimal feasible error for group $b$ is the same given $(X,G, X')$ as given $(X,G)$. By assumption that $(X,G)$ is $r$-skewed,  group $b$'s minimal error given $(X,G)$ exceeds group $r$'s minimal error given $(X,G)$. Since group $b$'s minimal error is the same given $(X,G)$ and $(X,G, X')$, while group $r$'s minimal error is weakly smaller given $(X,G,X')$ compared to $(X,G)$, it must be that group $b$ minimal error given $(X,G, X')$ also exceeds the group $r$ minimal error given $(X,G, X')$. In other words, $r_{X,G, X'} = b_{X, G,X'}$ also lies strictly above the 45-degree line, and the fairness-accuracy frontier $\mathcal{F}_{X,G, X'}$ is the horizontal line segment from $r_{X,G, X'} = b_{X,G, X'}$ to $f_{X, G,X'}$. Crucially, this line segment shares the same $e_b$ as the line segment from $r_{X,G} = b_{X,G}$ to $f_{X,G}$. In addition, as $r_{X,G, X'}$ must have weakly smaller $e_r$ than $r_{X,G}$, and $f_{X,G, X'}$ must be weakly closer to the 45-degree line than $f_{X,G}$, we deduce that the unconstrained fairness-accuracy frontier $\mathcal{F}_{X,G,X'}$ is a horizontal line segment that is a superset of the line segment $\mathcal{F}_{X,G}$. Thus, in particular, $r_{X,G} = b_{X,G}$ belongs to both of these frontiers. Lemma \ref{lemm:BayesDesign} thus imply that $r_{X,G} = b_{X,G}$ also belongs to the input-design fairness-accuracy frontiers $\mathcal{F}^{*}_{X,G}$ and $\mathcal{F}^{*}_{X,G,X'}$ ($r_{X,G} = b_{X,G}$ belongs to $H$ because this point can be implemented by giving $(X,G)$ to the agent, who will then minimize both groups' errors given this information). The result then follows from Lemma \ref{lemm:UniformWorsen}.

If $X'$ is decision-relevant over $X$ for group $b$, then Lemma \ref{lemm:DecisionRelevant} tells us that $\underline{e}^*_b < \underline{e}_b$ with strict inequality. There are two cases to consider here. One case involves $\underline{e}^*_b > \underline{e}^*_r$, so that $(X, G,X')$  is $r$-skewed just as $(X,G)$ is. Then the unconstrained fairness-accuracy frontier $\mathcal{F}_{X,G,X'}$ is again a horizontal line segment, but with $e_b$ equal to $\underline{e}^*_b$. Since $\underline{e}^*_b < \underline{e}_b$, this frontier is parallel but lower than the fairness-accuracy frontier $\mathcal{F}_{X,G}$. Thus $\mathcal{F}_{X,G}$ does not intersect $\mathcal{F}_{X,G,X'}$. As their subsets, the input-design fairness-accuracy frontiers $\mathcal{F}^{*}_{X,G}$ and $\mathcal{F}^{*}_{X,G,X'}$ also do not intersect. Thus the result follows from Lemma \ref{lemm:UniformWorsen}. In the remaining case we have $\underline{e}^*_b \leq \underline{e}^*_r$, so that $(X,G, X')$  is $b$-skewed. Then the unconstrained fairness-accuracy frontier $\mathcal{F}_{X,G, X'}$ is now a \emph{vertical} line segment with $e_r = \underline{e}^*_r$. The points on this frontier have varying $e_b$, but any of the $e_b$ does not exceed $\underline{e}^*_r$ because these points are below the 45-degree line. Because $\underline{e}^*_r \leq \underline{e}_r < \underline{e}_b$, we thus know that any point on the frontier $\mathcal{F}_{X,G, X'}$ has strictly smaller $e_b$ compared to any point on $\mathcal{F}_{X,G}$. Once again these two unconstrained frontiers do not intersect, and nor do the input-design frontiers. This proves Proposition \ref{prop:ExcludeX'overG} when $(X,G)$ is $r$-skewed. 

A symmetric argument applies when $(X,G)$ is $b$-skewed, so below we focus on the case where $(X,G)$ is group-balanced. That is, $r_{X,G} = b_{X,G}$ lies on the 45-degree line. In this case the fairness-accuracy frontiers $\mathcal{F}_{X,G}$ and $\mathcal{F}^{*}_{X,G}$ are both this singleton point. If $X'$ is not decision-relevant over $X$ for group $b$, then Lemma \ref{lemm:DecisionRelevant} tells us that $\underline{e}^*_b = \underline{e}_b = \underline{e}_r \geq \underline{e}^*_r$. When equality holds the fairness-accuracy frontiers $\mathcal{F}_{X,G,X'}$ and $\mathcal{F}^{*}_{X,G,X'}$ are also the singleton point $r_{X,G} = b_{X,G}$, and the result trivially holds. If we instead have strict inequality $\underline{e}^*_b = \underline{e}_b > \underline{e}^*_r$, then $(X,G, X')$ is $r$-skewed and the unconstrained fairness-accuracy frontier $\mathcal{F}_{X,G,X'}$ is a horizontal line segment with one of the endpoints being $f_{X,G, X'} = r_{X,G} = b_{X,G}$. Thus $r_{X,G} = b_{X,G}$ belongs also to the input-design fairness-accuracy frontier $\mathcal{F}^{*}_{X,G,X'}$, showing that $\mathcal{F}^{*}_{X,G}$ and $\mathcal{F}^{*}_{X,G,X'}$ intersect. The result again follows from Lemma \ref{lemm:UniformWorsen}.

Conversely, suppose $X'$ is decision-relevant over $X$ for both groups. Then by Proposition \ref{prop:XRevealsG}, the unconstrained frontier $\mathcal{F}_{X,X'}$ is either a horizontal line segment with $e_b = \underline{e}^*_b < \underline{e}_b = \underline{e}_b$, or a vertical line segment with $e_r = \underline{e}^*_r < \underline{e}_r = \underline{e}_b$. Either way the point $r_X = b_X$ does not belong to this frontier, showing that $\mathcal{F}_X$ does not intersect with $\mathcal{F}_{X,X'}$. Hence $\mathcal{F}^{*}_{X}$ and $\mathcal{F}^{*}_{X,X'}$ also do not intersect, and Lemma \ref{lemm:UniformWorsen} concludes the proof. This completes the proof of Proposition \ref{prop:ExcludeX'overG}. 

\subsection{Proof of Proposition \ref{prop:SimpleXXG}}

Omitted since it is nearly identical to the proof of Proposition \ref{prop:XGSimple}.

\section{Supplementary Material to Section \ref{sec:Empirical}}

\subsection{Bootstrap Algorithm}
\label{app:bootstrap}
Let $(Y_i, X_i, G_i)_{i=1}^{\ell}$ be the test data and $\hat P$ be its empirical distribution.
Fix a sample split and any algorithm obtained from the training data.
To test the null hypothesis \eqref{eq:null1} against its alternative, we first compute the test statistic $\hat T \coloneqq \hat e_b^b - \hat e_w^b$, where $\hat e_g^b \coloneqq \frac{1}{n_g} \sum_{i \colon G_i = g} \ell(a(X_i), Y_i)$ is the sample estimate of $e_g^b$ on the test data.
Let $(Y_i^*, X_i^*, G_i^*)_{i=1}^{\ell}$ be a bootstrap sample drawn i.i.d.~from the test data and denote the bootstrap analogue of $\hat T$ by $\hat T^*$. 
The critical value for the test is based on the quantile of the bootstrap distribution of the test statistics. Specifically, let $\Psi$ be the (conditional) cumulative distribution function of $\hat{T}^{*} - \hat{T}$ given $\hat{P}$. In practice, these are not known exactly, but can be approximated via Monte-Carlo by taking $B$ independent bootstrap samples from $\hat{P}$ and computing the empirical distribution
\[
\hat \Psi(x) \coloneqq 
\frac{1}{B}\sum_{b=1}^B \mathbbm{1}\{\hat{T}^*_{b} - \hat{T} \le x\},\]
where $\hat{T}^*_{b}$ denotes the bootstrap test statistic computed on the $b$th sample. The $p$-value is given by $1 - \hat \Psi(\hat T)$.

To test the null hypothesis \eqref{eq:null2} against its alternative, let
$\hat T^W$ and $\hat T^M$ be the empirical analogues of $|e_w^w - e_m^w| - \delta$ and $|e_w^m - e_m^w| - \delta$, respectively. The bootstrap analogues $\hat T^{W*}$ and $\hat T^{M*}$ and the cdfs $\hat \Psi^W$ and $\hat \Psi^M$ are defined in the same way as before. The $p$-value is then given by
$\max \{\hat \Psi^W(\hat T^W), \hat \Psi^M(\hat T^M)\}$.

\subsection{Estimation of Feasible Sets via Linear Programming}
\label{app:LP}

First, let us examine the case where $\alpha_r$ and $\alpha_b$ are nonnegative.
When we compute the group errors, we substitute the classification loss
$
\mathbbm{1} \left\{ \mathbbm{1}\{\beta^\top x \geq 0\} \neq y \right\}
=
1 - \mathbbm{1} \left\{ (2y - 1) \beta^\top x > 0 \right\}
$
(for $\beta^\top x \neq 0$) with the hinge surrogate loss
$
\max \left\{1 -  (2y-1)\beta^\top x, 0 \right\}
$ (see also Figure~\ref{fig:surrogate_loss}.)
This allows us to formulate the following alternative LP with linear constraints (without integer constraints):
\[
\mathrm{(LP)} \left[
    \begin{matrix}
        \displaystyle{\min_{\beta, \hat e_r, \hat e_b,  w}} & \alpha_r \hat e_r + \alpha_b \hat e_b \\
        \mathrm{s.t.} 
        & \displaystyle{\hat e_r = \frac{1}{n_r}} \sum_{i \colon G_i = r} w_i + \lambda \|\beta\|_{1} \\[8pt]
        & \displaystyle{\hat e_b = \frac{1}{n_b}} \sum_{i \colon G_i = b} w_i + \lambda \|\beta\|_{1} \\[8pt]
        & w_i \geq 1 - (2 Y_i - 1)X_i^\top \beta  & (\text{for all $i$})  \\
        & w_i \geq 0 & (\text{for all $i$})\\
        & \beta \in \mathbb{R}^{\dim{\mathcal{X}}},\ \hat e_r \geq 0,\  \hat e_b \geq 0 
    \end{matrix}
    \right.
\]
Note that we also allow for the option of adding $L_1$ regularization into the definition of the sample error rates $e_b$ and $e_r$ in the definition of LP. Such regularization could be important to avoid over-fitting to the training data when $X$ is high-dimensional. We set $\lambda \coloneqq 0.001$ for the \citet{ObermeyerMullainathan} dataset, and $\lambda \coloneqq 0$ for the \citet{strack2014impact} dataset.
\begin{figure}
    \centering
    \includegraphics[height=5cm]{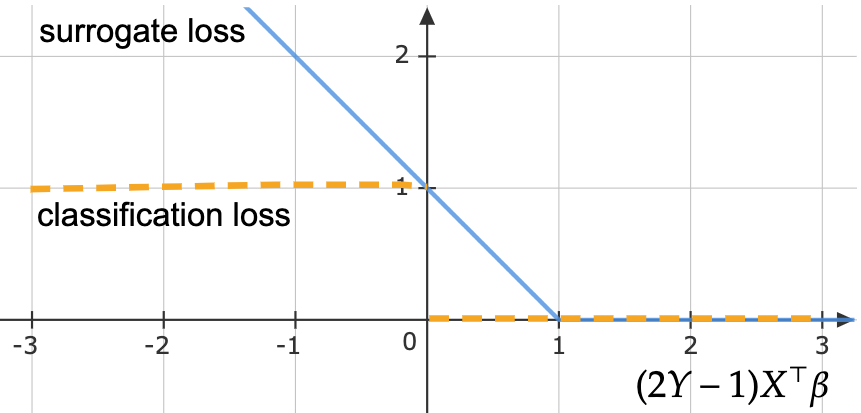}
    \caption{Classification and surrogate loss}
    \label{fig:surrogate_loss}
\end{figure}

When either $\alpha_r$ or $\alpha_b$ is strictly negative, we need a slight adjustment to the above LP. Suppose that $\alpha_r < 0$ and $\alpha_b \geq 0$ (other cases can be considered analogously.) The LP defined above becomes unbounded and has no optimal solution:
for any $M \in \mathbb{R}$, there exists a feasible solution $(\beta, \hat e_r, \hat e_b, w)$ such that $\alpha_r \hat e_r + \alpha_b \hat e_b < M$. To see this is possible, observe that we can make the value of the objective function arbitrarily small by increasing $w_i$ for some $i$ with $G_i=r$, which increases $\hat e_r$ (note that $w_i$ is not bounded from above).

To address the issue, we perform a change of variables before defining LP:
let $\tilde \alpha_r \coloneqq - \alpha_r > 0$ and $\tilde e_r \coloneqq 1 - \hat e_r$. Then, we have
$
\alpha_r \hat e_r + \alpha_b \hat e_b
=
\tilde \alpha_r \tilde e_r + \alpha_b \hat e_b + \tilde \alpha_r.
$
Since $\tilde \alpha_r$ is a constant, the original optimization problem
\[
\min\left\{ \alpha_r \hat e_r(a) + \alpha_b \hat e_b(a) \colon
\forall g, \ \hat e_g(a) = \frac{1}{n_g} \sum_{i \colon G_i=g} \ell(a(X_i), Y_i), \ a \in \mathcal{A}
\right\},
\]
is equivalent to the following:
\footnote{
The two problems are equivalent in the sense that the solutions of the two optimization problem have a clear one-to-one mapping: $1 - \hat e_r = \tilde e_r$.
}
\[
\min\left\{ \tilde \alpha_r \tilde e_r(a) + \alpha_b \hat e_b(a) \colon
\forall g, \ \hat e_g(a) = \frac{1}{n_g} \sum_{i \colon G_i=g} \ell(a(X_i), Y_i),\ \tilde e_r(a) = 1 - \hat e_r(a), \ a \in \mathcal{A}
\right\},
\]
As before, since we consider the class of linear algorithms, the latter optimization problem can be written as an MILP, which is the same as the previous MILP except for the contraint regarding $\hat e_r$: the constraint regarding $\hat e_r$ in the previous MILP is replaced by the following:
\begin{align*}
    \tilde e_r &= 1 - \hat e_r
    =
    \frac{1}{n_r} \sum_{i \colon G_i=r} \left(1 - \mathbbm{1}\left\{ \hat Y_i \neq Y_i \right\} \right)
    =
     \frac{1}{n_r} \sum_{i \colon G_i=r} \mathbbm{1}\left\{ \hat Y_i = Y_i \right\} 
     =
      \frac{1}{n_r} \sum_{i \colon G_i=r} \mathbbm{1}\left\{\hat Y_i \neq \tilde Y_i \right\},
\end{align*}
where $\tilde Y_i \coloneqq \mathbbm{1}\{Y_i = 0\}$ (the label is flipped.) This means that the new MILP has exactly the same form as the previous MILP if we flip the labels of group $r$. Then, the new MILP can be reduced to the LP by the same procedure as before.

In general, we can derive the LP by the following procedure:
\begin{enumerate}
    \item Formulate the LP as in the original case where $\alpha_b>0$ and $\alpha_r>0$;
    \item For group $g$ with $\alpha_g < 0$, flip the labels of group-$g$ observations, i.e., use $\tilde Y_i \coloneqq \mathbbm{1}\{Y_i = 0\}$ for group $g$ as if it were the true outcome variable;
    \item Replace $\alpha_g$ in the objective function with $|\alpha_g|$ for each $g$.
\end{enumerate}

\newpage
\setcounter{page}{1}
\setcounter{section}{15}
\setcounter{subsection}{0}

\clearpage
\noindent \begin{center}
{\large{}Online appendix to the paper}
\par\end{center}{\large \par}

\noindent \begin{center}
{\LARGE{} Algorithmic Design: A Fairness-Accuracy Frontier}
\par\end{center}{\LARGE \par}

\medskip{}

\begin{center}
Annie Liang \quad Jay Lu \quad Xiaosheng Mu \quad Kyohei Okumura
\par\end{center}{\large \par}

\noindent \begin{center}
\today
\par\end{center}

\subsection{Different Loss Functions} \label{app:DifferentLoss}

In this section, we generalize Theorem \ref{thm:FullDesignPareto} to cover cases where fairness and accuracy are evaluated using different loss functions.

Assume the set of covariate vectors $\mathcal{X}$ is finite, and let $a: \mathcal{X} \rightarrow \Delta(D)$ describe a generic algorithm and $\mathcal{A}$ denote the
set of all algorithms. As in the main text, there is a loss function $\ell: \mathcal{Y} \times \mathcal{D} \rightarrow \mathbb{R}$ such that each group $g$'s error rate under algorithm $a$ is $e_g = \mathbb{E}_{D \sim a(X)} \left[\ell(D,Y) \mid G=g\right]$. Different from the main text, the 
 \emph{unfairness} of algorithm $a\in\mathcal{A}$ is measured by
$
\left|h\left(a\right)\right|
$
where $h:\mathcal{A}\rightarrow\mathbb{R}_{+}$ is any linear function. This includes as a special case
\begin{equation} \label{eq:ExampleDifferentLoss}
h\left(a\right)=\mathbb{E}_{D\text{\ensuremath{\sim}}a\left(X\right)}\left[\tilde{\ell}\left(D,Y\right)\left|G=r\right.\right]-\mathbb{E}_{D\text{\ensuremath{\sim}}a\left(X\right)}\left[\tilde{\ell}\left(D,Y\right)\left|G=b\right.\right]
\end{equation}
where $\widetilde{\ell}: \mathcal{Y} \times \mathcal{D} \rightarrow \mathbb{R}$ is a ``fairness" loss function. Our previous approach is returned when $h$ takes the formulation in (\ref{eq:ExampleDifferentLoss}) and $\widetilde{\ell}$ is identical to $\ell$. 

For each pair of error rates $e\in\mathcal{E}_X$, we
define 
\[
d\left(e\right):=\min_{e\left(a\right)=e}\left|h\left(a\right)\right|
\]
to be the minimal unfairness that can be achieved using an algorithm that yields error pair $e$. This is well defined since $\left|h\left(\cdot\right)\right|$ is
continuous and the set of algorithms $\{a: e\left(a\right)=e\}$
is compact. 

We now extend the definitions of FA-dominance and the FA frontier.
\begin{definition} Let $>_{FA}$ be the strict order on $\mathcal{E}_X$ where 
$e>_{FA}e^{\prime}$ if $e<e'$
and $d\left(e\right)\leq d\left(e^{\prime}\right)$.
\end{definition}

\begin{definition}
The FA frontier is $\mathcal{F}_X := \{e\in \mathcal{E}_X : \textrm{no } e'\in \mathcal{E}_X \textrm{ such that } e' >_{FA} e \}.$
\end{definition}

\noindent When $h(a)$ has the formulation (\ref{eq:ExampleDifferentLoss}) and the accuracy and fairness loss functions $\tilde{\ell}=\ell$ coincide, then  $d\left(e\right)=\left|e_r-e_b\right|$,
and so these definitions reduce to Definitions \ref{def:Pareto}, \ref{def:FApref} and \ref{def:Frontier}.

Define 
\[\underline{\delta}=\min_{e\in\mathcal{E}}d\left(e\right)\]
to be the minimal level of unfairness that is achievable by a feasible algorithm. For any $\delta\geq \underline{\delta}$, 
\[
\mathcal{E}_{\delta}:=\left\{ e\in\mathcal{E}_X\text{ : }d\left(e\right)\leq\delta\right\} 
\]
is the set of errors achievable by an algorithm whose unfairness is weakly less than $\delta$.

\begin{definition} For each $\delta \geq \underline{\delta}$, define 
$r_{X}\left(\delta\right) = \argmin_{e \in \mathcal{E}_\delta} e_r$  and $b_{X}\left(\delta\right) = \argmin_{e \in \mathcal{E}_\delta} e_b$ to be the group-optimal points within each set $\mathcal{E}_{\delta}$, where we break ties by choosing the point that minimizes the other group's error. Further define 
$R_{X}  =\left(r_{X}\left(\delta\right)\right)_{\delta \geq \underline{\delta}}$ and
$B_{X} =\left(b_{X}\left(\delta\right)\right)_{\delta \geq \underline{\delta}}$ 
to be the set of group-optimal points as we vary over the level of unfairness. 
\end{definition}

\begin{definition} For any convex set $E \subset \mathbb{R}\times \mathbb{R}$, let $\mathcal{P}\left(E\right)$ denote the usual Pareto frontier of $E$, i.e., all points $e \in E$ where no other $e'\in E$ is weakly smaller in each entry and strictly smaller in at least one. 
\end{definition}

\begin{definition} The fairness-optimal set is $F_{X}:=\mathcal{P}\left(\mathcal{E}_{\underline{\delta}}\right)$.
\end{definition}

We now characterize the FA frontier for this general case.

\begin{theorem} \label{thm:gen_frontier}
$\mathcal{F}_X$ is the closed set bounded by
$R_{X}$, $B_{X}$, $\mathcal{P}_{X}$
and $F_{X}$. 
\end{theorem}

The property of group-balance generalizes as follows.

\begin{definition}[Generalized Group-Balance] Say that $X$ is \emph{generalized group-balanced} if $F_{X} \subseteq \mathcal{P}_X$.
\end{definition}

That is, $X$ is generalized group-balanced if the fairness-optimal set belongs to the usual Pareto frontier. This reduces to the condition  in the main text when $h$ takes the form given in (\ref{eq:ExampleDifferentLoss}) and $\widetilde{\ell} = \ell$. Several of our previous results extend under this generalization of group-balance. For example, group-balance again identifies when the fairness-accuracy frontier is equivalent to the usual Pareto frontier.

\begin{proposition} \label{prop:gen_gbalance}
If $X$ is generalized group-balanced, then $\mathcal{F}_{X}=\mathcal{P}_X$.
\end{proposition}

The related Corollary \ref{corr:Tradeoff} also extends.

\begin{corollary}  $X$ fails generalized group-balance if and only if there are points  $e,e'\in \mathcal{F}_X$ such that $e<e'$ but $d(e)>d(e')$.
\end{corollary}

In some cases, generalized group-balance reduces further. One such case is when $X \in \{0,1\}$ is binary, and $h$ follows the formulation in (\ref{eq:ExampleDifferentLoss}) where the fairness loss function $\ell(d,y) = \mathbbm{1}(d=0)$ is an indicator for whether the decision is equal to 0 (e.g., not getting hired); in this case, fairness is measured as the absolute difference in the conditional probability of being assigned $d=0$ given membership in either group. Let the accuracy loss function $\ell(d,y) = \mathbbm{1}(d\neq y)$ be the standard misclassification loss. 
For each $g \in \{r,b\}$ define
\[a_{g}^x := \mathbbm{1}\left[\mathbb{P}(Y=1 \mid X=x,G=g)\geq 1/2\right]\]
to be the optimal action for group $g$ given signal realization $x$, breaking ties in favor of $d=1$. Then generalized group-balance reduces to the following easily checkable condition.

\begin{claim} \label{claim:BinaryX} 
$X$ fails generalized group-balance if and only if $a_{r0}=a_{b0}$ and $a_{r1} = a_{b1}$ and these values are distinct---that is, the optimal action is the same for both groups given either covariate realization, and this common optimal action differs across covariate realizations. 
\end{claim}

The proof of Claim \ref{claim:BinaryX} and all other results mentioned in this section are contained below.

\subsubsection{Proofs of Theorem \ref{thm:gen_frontier} and Proposition \ref{prop:gen_gbalance}.} To save on notation we suppress dependence on $X$ in what follows, using $\mathcal{F}$ for the fairness-accuracy frontier, $\mathcal{E}$ for the feasible set and $\mathcal{A}$ for the set of algorithms. We first show that the fairness-accuracy  frontier is the union of the Pareto frontiers of the unfairness sublevel sets.

\begin{lemma} \label{lem:dcont}
$d\left(\cdot\right)$ is continuous and convex.
\end{lemma}
\begin{proof}
We first show convexity. Consider $e_{1},e_{2}\in\mathcal{E}_X$
and let $a_{i}$ be the algorithm that minimizes unfairness among all algorithms yielding error pair $e_{i}$; that
is, $d\left(e_{i}\right)=\left|h\left(a_{i}\right)\right|$ and $e\left(a_{i}\right)=e_{i}$
for $i\in\left\{ 1,2\right\} $. Since $e\left(\cdot\right)$ and
$h\left(\cdot\right)$ are linear,
\begin{alignat*}{1}
d\left(\lambda e_{1}+\left(1-\lambda\right)e_{2}\right) & =d\left(\lambda e\left(a_{1}\right)+\left(1-\lambda\right)e\left(a_{2}\right)\right)=d\left(e\left(\lambda a_{1}+\left(1-\lambda\right)a_{2}\right)\right)\\
 & \leq\left|h\left(\lambda a_{1}+\left(1-\lambda\right)a_{2}\right)\right|=\left|\lambda h\left(a_{1}\right)+\left(1-\lambda\right)h\left(a_{2}\right)\right|\\
 & \leq\lambda\left|h\left(a_{1}\right)\right|+\left(1-\lambda\right)\left|h\left(a_{2}\right)\right|\\
 & =\lambda d\left(e_{1}\right)+\left(1-\lambda\right)d\left(e_{2}\right)
\end{alignat*}
as desired. 

We now show continuity. Consider the correspondence $\varphi:\mathcal{E}\rightrightarrows\mathcal{A}$
where
\[
\varphi\left(e\right):=\left\{ a\in\mathcal{A}\text{ : }e\left(a\right)=e\right\} 
\]
Note this is compact-valued. We will show $\varphi$ is continuous.
To show upper hemicontinuity, consider sequences $e^{k}\rightarrow e$ and $a^{k}\rightarrow a$
where each $a^{k}\in\varphi\left(e^{k}\right)$. Then by definition of $\varphi$, each  $e\left(a^{k}\right)=e^{k}$, 
which further implies $e\left(a\right)=e$ as $e\left(\cdot\right)$ is continuous.
Thus, $a\in\varphi\left(e\right)$ proving upper hemicontinuity. 

We now show lower hemicontinuity. Consider $e^{k}\rightarrow e$ and
some $a\in\varphi\left(e\right)$. For each $e^{k}$, let $a^{k}\in\varphi\left(e^{k}\right)$
be the closest point in $\varphi\left(e^{k}\right)$ to $a$. Since $\varphi\left(e^{k}\right)$ is a linear subspace, $a^{k}$
is unique and well-defined. We will show that $\left|a^{k}-a\right|\rightarrow0$.
Suppose otherwise, in which case we can find some $n$ and $\varepsilon>0$ such
that for all $k>n$, $\left|a^{\prime}-a\right|\geq\varepsilon$ for
all $a^{\prime}\in\varphi\left(e^{k}\right)$. But that means $\varphi\left(e\right)$
is also strictly separated from $a$ yielding a contradiction. This
proves $\varphi$ is also lower hemicontinuous and thus continuous.
Since $h\left(\cdot\right)$ is continuous, by the maximum theorem,
$d\left(\cdot\right)$ is continuous. 
\end{proof}

\begin{lemma}
$\mathcal{E}_{\delta}$ is closed and convex.
\end{lemma}
\begin{proof}
Immediate from the fact that $d\left(\cdot\right)$ is convex and continuous (Lemma \ref{lem:dcont}).
\end{proof}
\begin{lemma}
$\mathcal{F}=\bigcup_{\delta\geq0}\mathcal{P}\left(\mathcal{E}_{\delta}\right)$
\end{lemma}
\begin{proof}
First, suppose $e\in\mathcal{P}\left(\mathcal{E}_{\delta}\right)$
for some $\delta\geq0$. Suppose $e\not\in\mathcal{F}$ so there exists
some $e^{\prime}\in\mathcal{E}$ that FA-dominates $e$. Thus, $d\left(e^{\prime}\right)\leq d\left(e\right)$
so $e^{\prime}\in\mathcal{E}_{\delta}$. Note that if $e_{g}^{\prime}<e_{g}$
for some group $g$, then this contradicts $e\in\mathcal{P}\left(\mathcal{E}_{\delta}\right)$.
Thus, it must be that $e_{g}^{\prime}=e_{g}$ for both groups $g$
so $e^{\prime}=e$ yielding a contradiction.

Now, let $e\in\mathcal{F}$ and consider $\delta=d\left(e\right)$.
Clearly, $e\in\mathcal{E}_{\delta}$. Note that if $e\not\in\mathcal{P}\left(\mathcal{E}_{\delta}\right)$,
then there exists another $e^{\prime}\in\mathcal{E}_{\delta}$ that
Pareto dominates $e$. But since $d\left(e^{\prime}\right)=d\left(e\right)$,
$e^{\prime}$ also FA-dominates $e$ yielding a contradiction.
\end{proof}

\noindent \textbf{Completion of the proof of Theorem \ref{thm:gen_frontier}}. We will first prove $r_{X}\left(\cdot\right)$ is continuous. First,
let
\begin{alignat*}{1}
\mathcal{A}_{\delta} & :=\left\{ a\in\mathcal{A}\text{ : }\left|h\left(a\right)\right|\leq\delta\right\} 
\end{alignat*}
and note that
\[
\mathcal{A}_{\delta}=\left\{ a\in\mathcal{A}\text{ : }-\delta\leq h\left(a\right)\leq\delta\right\} 
\]
Since $h\left(\cdot\right)$ is linear, this is just a polytope in
$A=\left[0,1\right]^{\mathcal{X}}$. 

Fix some $\delta$ and define
\begin{alignat*}{1}
e_{r}^{\ast} & :=\min_{a\in\mathcal{A}_{\delta}}e_{r}\left(a\right)\\
e_{b}^{\ast} & :=\min_{a^{\prime}\in\arg\min_{a\in\mathcal{A}_{\delta}}e_{r}\left(a\right)}e_{b}\left(a^{\prime}\right)
\end{alignat*}
We will show that $e^{\ast}=r_{X}\left(\delta\right)$. First, let
$a^{\ast}$ be the corresponding algorithm for $e^{\ast}$ so $\left|h\left(a^{\ast}\right)\right|\leq\delta$.
This implies that $d\left(e^{\ast}\right)\leq\delta$ so $e^{\ast}\in\mathcal{E}_{\delta}$.
Thus, if we let $e=r_{X}\left(\delta\right)$, then $e_{r}\leq e_{r}^{\ast}$.
Suppose the inequality is strict. That means we can find some algorithm
$a^{\ast\ast}$ such that $e_{r}\left(a^{\ast\ast}\right)<e_{r}\left(a^{\ast}\right)$
and $\left|h\left(a^{\ast\ast}\right)\right|\leq\delta$. But that
implies $a^{\ast\ast}\in\mathcal{A}_{\delta}$ contradicting the definition
of $e_{r}^{\ast}$ so it must be $e_{r}=e_{r}^{\ast}$. This implies
that $e_{b}\leq e_{b}^{\ast}$. Suppose the inequaility is strict,
so again we can find some algorithm $a^{\ast\ast}$ such that $e_{r}\left(a^{\ast\ast}\right)=e_{r}\left(a^{\ast}\right)$,
$e_{b}\left(a^{\ast\ast}\right)<e_{b}\left(a^{\ast}\right)$ and $\left|h\left(a^{\ast\ast}\right)\right|\leq\delta$.
This contradicts the definition of $e_{b}^{\ast}$ so it must be that
$e=e^{\ast}$. We can thus write
\[
r_{X}\left(\delta\right)=\left(\min_{a\in\mathcal{A}_{\delta}}e_{r}\left(a\right),\min_{a^{\prime}\in\arg\min_{a\in\mathcal{A}_{\delta}}e_{r}\left(a\right)}e_{b}\left(a^{\prime}\right)\right)
\]
Continuity follows from the fact that $e_{r}\left(\cdot\right)$,
$e_{b}\left(\cdot\right)$ and $h\left(\cdot\right)$ are all linear.
That $b_{X}\left(\delta\right)$ is continuous follows symmetrically.
Since $\mathcal{F}=\bigcup_{\delta\geq0}\mathcal{P}\left(\mathcal{E}_{\delta}\right)$
and $\mathcal{P}\left(\mathcal{E}_{\delta}\right)$ is characterized
by $r_{X}\left(\delta\right)$ and $b_{X}\left(\delta\right)$, the
result follows.\\

\noindent \textbf{Completion of the proof of Proposition \ref{prop:gen_gbalance}}.
Recall 
\[
\mathcal{A}_{\delta}=\left\{ a\in A\text{ : }-\delta\leq h\left(a\right)\leq\delta\right\} 
\]
Now, for $\lambda\in\left(0,1\right)$, define
\[
e_{\lambda}:=\lambda e_{r}+\left(1-\lambda\right)e_{b}
\]
and

\[
\mathcal{A}_{\delta}^{\ast}\left(\lambda\right):=\arg\min_{a\in\mathcal{A}_{\delta}}e_{\lambda}\left(a\right)
\]
Define $\mathcal{A}_{\delta}^{\ast}\left(0\right)$ and $\mathcal{A}_{\delta}^{\ast}\left(1\right)$
similarily but with tie-breaking. For large enough $\delta$ where
$\mathcal{A}_{\delta}=\mathcal{A}$, we can just let $\mathcal{A}^{\ast}=\mathcal{A}_{\delta}^{\ast}$.
It is straightforward to show that
\[
\mathcal{P}\left(\mathcal{E}_{\delta}\right)=\bigcup_{\lambda\in\left[0,1\right]}\left\{ e\left(a\right)\text{ : }a\in\mathcal{A}_{\delta}^{\ast}\left(\lambda\right)\right\} 
\]

We will now prove that if $\mathcal{P}\left(\mathcal{E}_{\delta_{1}}\right)\subset\mathcal{P}\left(\mathcal{E}\right)$
and $\delta_{1}\leq\delta_{2}$, then $\mathcal{P}\left(\mathcal{E}_{\delta_{2}}\right)\subset\mathcal{P}\left(\mathcal{E}\right)$.
Consider some $e\in\mathcal{P}\left(\mathcal{E}_{\delta_{2}}\right)$
so we can find some $\lambda\in\left[0,1\right]$ such that $a_{2}\in\mathcal{A}_{\delta_{2}}^{\ast}\left(\lambda\right)$
and $e=e\left(a_{2}\right)$. Let $a_{1}\in\mathcal{A}_{\delta_{1}}^{\ast}\left(\lambda\right)$
and $\bar{a}\in\mathcal{A}^{\ast}\left(\lambda\right)$. Since $\mathcal{A}_{\delta}$
is increasing in $\delta$, it must be that
\[
e_{\lambda}\left(a_{1}\right)\leq e_{\lambda}\left(a_{2}\right)\leq e_{\lambda}\left(\bar{a}\right)
\]
Now, since $e\left(a_{1}\right)\in\mathcal{P}\left(\mathcal{E}_{\delta_{1}}\right)\subset\mathcal{P}\left(\mathcal{E}\right)$,
there must exist some $a_{1}^{\prime}\in\mathcal{A}^{\ast}\left(\lambda_{1}\right)$
for some $\lambda_{1}\in\left[0,1\right]$ such that $e\left(a_{1}^{\prime}\right)=e\left(a_{1}\right)$.
That implies that
\[
e_{\lambda_{1}}\left(a_{1}\right)=e_{\lambda_{1}}\left(a_{1}^{\prime}\right)\leq e_{\lambda_{1}}\left(a\right)
\]
for all $a\in\mathcal{A}$ so $a_{1}\in\mathcal{A}^{\ast}\left(\lambda_{1}\right)$.
Note that $a_{1}\in\mathcal{A}^{\ast}\left(\lambda_{1}\right)$ and
$\bar{a}\in\mathcal{A}^{\ast}\left(\lambda\right)$ are on the boundary
of $\mathcal{A}$. Since $a_{2}$ is also on the boundary of $\mathcal{A}$,
by continuity, we can find some $\lambda_{2}$ between $\lambda_{1}$
and $\lambda$ such that $e\left(a_{2}\right)\in\mathcal{A}_{\delta_{2}}^{\ast}\left(\lambda_{2}\right)$.
This implies $e\left(a_{2}\right)\in\mathcal{P}\left(\mathcal{E}\right)$
as desired. 

\subsection{Proof of Claim \ref{claim:BinaryX}}

Since $X$ is binary-valued, each algorithm can be identified with a pair $(p_0,p_1)$ denoting the respective probabilities with which $X=0$ and $X=1$ are mapped into $d=1$. From the proof of Lemma \ref{lemm:ConvexPolygon}, we know that the feasible set is a polygon whose vertices are the error rates derived from the deterministic algorithms $(0,0)$, $(0,1)$, $(1,0)$, and $(1,1)$. Moreover, 
\begin{align*}
    \vert \mathbb{E}(d=1 \mid G=r) - \mathbb{E}(d=1 \mid G=b) \vert & = \vert (\alpha_r p_0 + (1-\alpha_r) p_1) - (\alpha_b p_0 + (1-\alpha_b) p_1 \vert \\
    & = \vert (\alpha_r - \alpha_b)(p_0 - p_1)) \vert
\end{align*}
where $\alpha_g := \mathbb{P}(X=0 \mid G=g)$. So unfairness is minimized (and achieves the value zero) by setting $p_0=p_1$. Thus $F_{X}$ is the Pareto set of the line from the $(1,1)$ vertex to the $(0,0)$ vertex of the polygon.
\begin{figure}[h]
    \centering
    \includegraphics[scale=0.65]{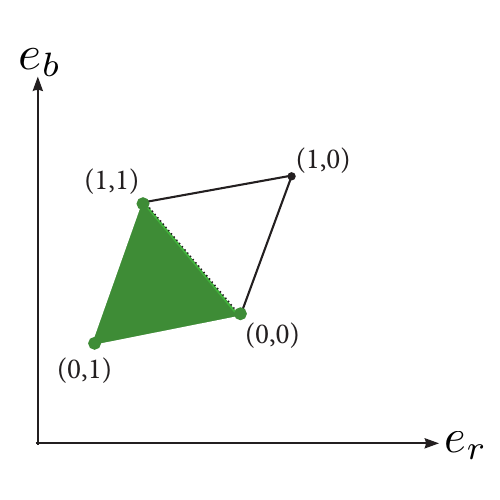}
    \caption{\footnotesize{$X$ fails generalized group balance. The fairness-accuracy set is the shaded green area. In this example, $r_X=b_X=(0,1)$ while $F_{X}$ is the line from $(0,0)$ to $(1,1)$.}} \label{fig:FnotinP}
\end{figure}

Suppose $a_{r0}=a_{b0}$ and $a_{r1} = a_{b1}$ with distinct values. Then the deterministic algorithm $(p_0,p_1)=(a_{r0},a_{r1}) \in \{(0,1),(1,0)\}$ maximizes accuracy for both groups. The corresponding vertex is simultaneously $r_X$ and $b_X$, so it is also the Pareto set $\mathcal{P}(\mathcal{E})$. But this point does not intersect the line from $(0,0)$ to $(1,1)$, so $F_{X}$ does not belong to $\mathcal{P}(\mathcal{E})$ (see Figure \ref{fig:FnotinP} for an example). We thus have the claim in one direction.
 
 \begin{figure}[h]
    \centering
    \includegraphics[scale=0.65]{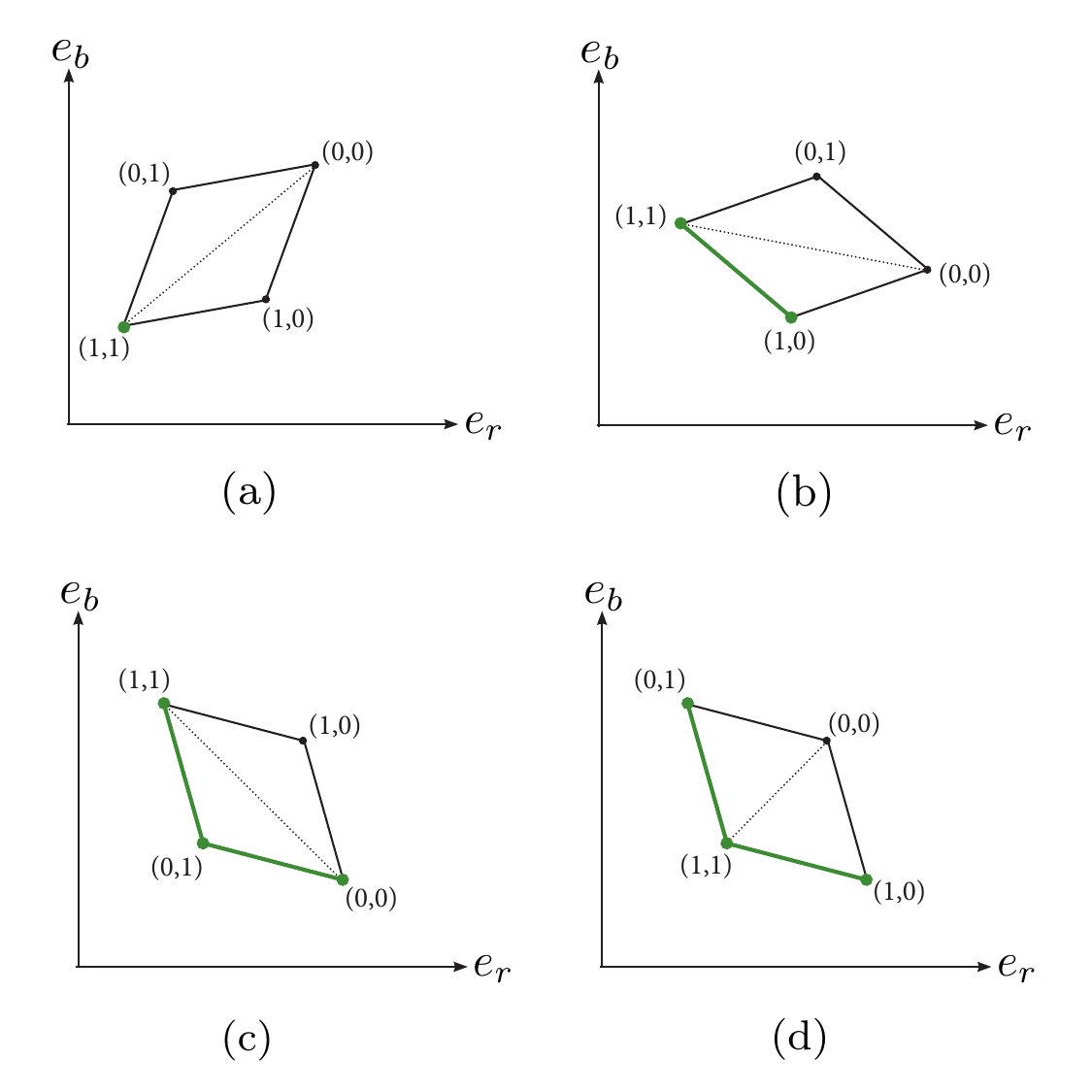}
    \caption{\footnotesize{$X$ satisfies generalized group balance. The fairness-accuracy set is the shaded green area. In all cases, $F_x=\{f_x\}$ is a singleton. In Panel (a), $r_X=b_X=f_{X}=(1,1)$. In Panels (b) and (c), $r_X=f_{X}=(1,1)$ while $b_X=(0,0)$. In Panel (d), $r_X=(0,1)$, $f_{X}=(1,1)$, and $b_X=(1,0)$.}} \label{fig:FinP}
\end{figure}
In the other direction, suppose first that  $a_{r0}=a_{b0} = a_{r1}=a_{b0}$. In this case, $(p_0,p_1)=(a_{r0},a_{r1}) \in \{(0,0),(1,1)\}$ is simultaneously $r_X$, $b_X$, and $F_{X}$, as depicted in Panel (a) of Figure \ref{fig:FinP}. Clearly $F_{X} \in \mathcal{P}(\mathcal{E})$. 

In all remaining cases, $r_X$ is different from $b_X$, so the Pareto set $\mathcal{P}(\mathcal{E}_X)$ includes at least one non-degenerate line segment.  This line segment must  include  at least one of the vertices $(0,0)$ and $(1,1)$. See Panels (b)-(d) of Figure \ref{fig:FinP} for the possible configurations. So the Pareto set intersects the line connecting $(1,1)$ and $(0,0)$, and $F_{X}$ is precisely this point of intersection. Thus $F_{X} \in \mathcal{P}(\mathcal{E})$, completing the argument.

\subsection{General Fairness Criteria} \label{sec:generalphi}

In this section, we consider the general case where fairness is evaluated using $\left|\phi\left(e_{r}\right)-\phi\left(e_{b}\right)\right|$ for some strictly increasing continuous function $\phi$. For instance, if $\phi$ is $\log$, then this reduces to using the ratio of error rates as a measure of fairness. The characterization of the fairness-accuracy frontier remains the same except the fairness optimal point $f_{X}$ may now be different. Whether it expands or contracts depends on the curvature of $\phi$ as the following proposition demonstrates.\footnote{We assume that the accuracy and fairness loss functions are the same but can generalize the results in this section via the same methodology as in Section \ref{app:DifferentLoss}.}

\begin{proposition}
Let $\mathcal{F}^{\prime}_{X}$ denote the fairness-accuracy frontier where
fairness is evaluated using 
\[
\left|\phi\left(e_{r}\right)-\phi\left(e_{b}\right)\right|
\]
for strictly increasing $\phi:\mathbb{R}\rightarrow\mathbb{R}$. Then
\begin{enumerate}
\item $\mathcal{F}_X=\mathcal{F}^{\prime}_{X}$ if
$X$ is group-balanced
\item $\mathcal{F}_X\subseteq\mathcal{F}^{\prime}_{X}$
if $X$ is group-skewed and $\phi$ in concave 
\item $\mathcal{F}_X\supseteq\mathcal{F}^{\prime}_{X}$
if $X$ is group-skewed and $\phi$ in convex
\end{enumerate}
\end{proposition}
\begin{proof}Let $\mathcal{E}_X$ and $\mathcal{E}^{\prime}_{X}$
denote the feasible sets where fairness is defined using $\left|e_{r}-e_{b}\right|$
and $\left|\phi\left(e_{r}\right)-\phi\left(e_{b}\right)\right|$
respectively. Let $f_{X}$ and $f_{X}^{\prime}$ denote the corresponding
fairness optimal points. First, note that if $X$
is group-balanced, then by the same argument as Theorem 1, $\mathcal{F}_X=\mathcal{F}^{\prime}_{X}$
is the lower boundary from $r_{X}=r_{X}^{\prime}$ to $b_{X}=b_{X}^{\prime}$.

Now, suppose $X$ is $r$-skewed without loss. Let $e$ and $e^{\prime}$
correspond to $f_{X}$ and $f_{X}^{\prime}$ so
\begin{alignat*}{1}
e_{b}-e_{r} & \leq e_{b}^{\prime}-e_{r}^{\prime}\\
\phi\left(e_{b}^{\prime}\right)-\phi\left(e_{r}^{\prime}\right) & \leq\phi\left(e_{b}\right)-\phi\left(e_{r}\right)
\end{alignat*}
First, suppose $\phi$ is concave. We will show that $e_{r}^{\prime}\geq e_{r}$.
Suppose by contradiction that $e_{r}^{\prime}<e_{r}$ so $\phi\left(e_{r}^{\prime}\right)<\phi\left(e_{r}\right)$.
Thus,
\[
\phi\left(e_{b}^{\prime}\right)-\phi\left(e_{b}\right)\leq\phi\left(e_{r}^{\prime}\right)-\phi\left(e_{r}\right)<0
\]
so $e_{b}^{\prime}<e_{b}$. Thus, we have $e_{r}^{\prime}\leq e_{b}^{\prime}<e_{b}$.
Note that
\[
e_{b}^{\prime}=\lambda e_{b}+\left(1-\lambda\right)e_{r}^{\prime}
\]
where
\[
\lambda:=\frac{e_{b}^{\prime}-e_{r}^{\prime}}{e_{b}-e_{r}^{\prime}}
\]
We thus have
\begin{alignat*}{1}
\phi\left(e_{b}\right)-\phi\left(e_{r}\right)+\phi\left(e_{r}^{\prime}\right) & \geq\phi\left(e_{b}^{\prime}\right)=\phi\left(\lambda e_{b}+\left(1-\lambda\right)e_{r}^{\prime}\right)\\
 & \geq\lambda\phi\left(e_{b}\right)+\left(1-\lambda\right)\phi\left(e_{r}^{\prime}\right)\\
\left(1-\lambda\right)\left(\phi\left(e_{b}\right)-\phi\left(e_{r}^{\prime}\right)\right) & \geq\phi\left(e_{r}\right)-\phi\left(e_{r}^{\prime}\right)\\
\left(e_{b}-e_{b}^{\prime}\right)\frac{\phi\left(e_{b}\right)-\phi\left(e_{r}^{\prime}\right)}{e_{b}-e_{r}^{\prime}} & \geq\phi\left(e_{r}\right)-\phi\left(e_{r}^{\prime}\right)
\end{alignat*}
where the second inequality follows from the fact that $\phi$ is
concave. Since $e_{r}-e_{r}^{\prime}\geq e_{b}-e_{b}^{\prime}$, this
implies
\[
\frac{\phi\left(e_{b}\right)-\phi\left(e_{r}^{\prime}\right)}{e_{b}-e_{r}^{\prime}}\geq\frac{\phi\left(e_{r}\right)-\phi\left(e_{r}^{\prime}\right)}{e_{r}-e_{r}^{\prime}}
\]
Since $X$ is $r$-skewed, $e_b \geq e_{r}> e_{r}^{\prime}$. Since $\phi$ is concave, the above inequality
must be satisfied with equality. This means that
\[
\left(e_{b}-e_{b}^{\prime}\right)\frac{\phi\left(e_{b}\right)-\phi\left(e_{r}^{\prime}\right)}{e_{b}-e_{r}^{\prime}}\geq \phi\left(e_{r}\right)-\phi\left(e_{r}^{\prime}\right) =\left(e_{r}-e_{r}^{\prime}\right)\frac{\phi\left(e_{b}\right)-\phi\left(e_{r}^{\prime}\right)}{e_{b}-e_{r}^{\prime}}
\]
so $e_{b}-e_{b}^{\prime}=e_{r}-e_{r}^{\prime}$ or $e_{b}-e_{r}=e_{b}^{\prime}-e_{r}^{\prime}$.
But $e$ corresponds to $f_{X}$ and since $e^{\prime}$ achieves the same fairness as $e$, it must be that $e_{r}\leq e_{r}^{\prime}$. This contradicts our assumption that $e_{r}^{\prime}< e_{r}$. Thus, $e_{r}^{\prime}\geq e_{r}$ and  by
the same argument characterizing the FA frontier as in Theorem 1, $\mathcal{F}_X\subseteq\mathcal{F}^{\prime}_{X}$.
The case for when $\phi$ is convex is symmetric.
\end{proof}

\subsection{Adversarial Agents } \label{app:ExtendBD}

We now consider the problem outlined in Section \ref{sec:DesignInputs}, when one of the weights $\alpha_r,\alpha_b$ is negative.\footnote{It is straightforward also to consider the case where both weights are negative, but we do not consider this setting to be practically relevant.} Without loss, let $\alpha_r > 0 > \alpha_b$, reflecting an adversarial agent who prefers for group $b$'s error to be higher.
 The first half of Lemma \ref{lemm:BayesDesign} extends fully.

\begin{lemma} \label{lemm:feasiblesetH-Adv}For every covariate vector $X$, $\mathcal{E}^{*}_{X} = \mathcal{E}_X \cap H$.
\end{lemma}

\noindent But the analogous equivalence for the FA frontier does not extend. Instead, similar to the development of $r_X$, $b_X$, and $f_{X}$,  define \[\displaystyle g^*_X := \argmin_{e \in \mathcal{E}^{*}_{X}} e_g\]
to be the feasible point in $\mathcal{E}^{*}_{X}$ that minimizes group $g$'s error (breaking ties by minimizing the other group's error), and define
\[\displaystyle f^*_X := \argmin_{e \in \mathcal{E}^{*}_{X}} \vert e_r - e_b \vert\]
to be the point that minimizes the absolute difference between group errors (breaking ties by minimizing either group's error).

\begin{definition} Covariate vector $X$ is:
\begin{itemize}
    \item \emph{input-design $r$-skewed} if $e_r< e_b$ at $r^*_{X}$ and $e_r\leq e_b$ at $b^*_{X}$
    \item \emph{input-design $b$-skewed} if $e_b< e_r$ at $b^*_{X}$ and $e_b\leq e_r$ at $r^*_{X}$
    \item \emph{input-design group-balanced} otherwise
\end{itemize}
\end{definition}

\noindent The proof for Theorem \ref{thm:FullDesignPareto} applies for any compact and convex feasible set, and so directly implies:

\begin{theorem} \label{thm:AdversarialPF} The input-design fairness-accuracy (FA) frontier $\mathcal{F}^{*}_{X}$ is the lower boundary of the input-design feasible set $\mathcal{E}^{*}_{X}$ between
\begin{itemize}
\item[(a)] $r^*_{X}$ and $b^*_{X}$ if $X$ is input-design group-balanced
\item[(b)] $g^*_X$ and $f^*_X$ if $X$ is input-design $g$-skewed
\end{itemize}
\end{theorem}

We can use this characterization to extend our result from Section \ref{sec:GNew}.
 
 \begin{definition} $X$ is \emph{strictly input-design-group-balanced} if $e_r<e_b$ at $r^*_{X}$ and $e_b<e_r$ at $b^*_{X}$.
\end{definition}

 \begin{proposition}
     Suppose $\alpha_r > 0 > \alpha_b$ and $X$ is strictly input-design group-balanced. Then excluding $G$ over $X$ uniformly worsens the frontier.
 \end{proposition}
 
This result says that, perhaps surprisingly, even if the agent choosing the algorithm has adversarial motives against one of the groups, the designer may still prefer to send information about group identity. The notion of group-balanced covariate vectors, suitably adapted to the input design setting, again serves as a sufficient condition for uniform worsening of the frontier when excluding $G$.  
\begin{proof}
By assumption that $X$ is strictly input-design group-balanced, the input-design FA frontier given $X$ is the lower boundary of $\mathcal{E}^{*}_{X}$ from $r^*_{X}$ to $b^*_{X}$, which consists of negatively sloped edges. We will show that every point on this frontier is FA-dominated by some point in $\mathcal{E}^{*}_{X,G}$. 

If this point $(e_r, e_b)$ is distinct from $b^*_{X}$ and $r^*_{X}$, then we claim that for sufficiently small positive $\epsilon$, the point $(e_r - \epsilon, e_b - \epsilon)$ belongs to $\mathcal{E}^{*}_{X,G}$. Indeed, $(e_r - \epsilon, e_b - \epsilon)$ belongs to the unconstrained feasible set $\mathcal{E}_{X,G}$ because this feasible set is a rectangle, and $e_r - \epsilon$, $e_b - \epsilon$ are within the minimal and maximal group errors achievable given $X$. Moreover, $(e_r, e_b)$ must have smaller group-$r$ error and larger group-$b$ error compared to $b^*_{X}$, which means the same is true for $(e_r - \epsilon, e_b - \epsilon)$. Since $\alpha_r >  0 > \alpha_b$, the point $(e_r - \epsilon, e_b - \epsilon)$ must belong to $H$
 given that $b^*_{X}$ does. 
Hence when $(e_r, e_b)$ differs from $b^*_{X}$ and $r^*_{X}$, it is FA-dominated by $(e_r - \epsilon, e_b - \epsilon) \in \mathcal{E}^{*}_{X,G}$. 

Suppose now that $(e_r, e_b) = b^*_{X}$. Then by similar argument it is FA-dominated by $(e_r - \epsilon, e_b) \in \mathcal{E}^{*}_{X,G}$. Finally if $(e_r, e_b) = r^*_{X}$, then it is FA-dominated by $(e_r, e_b - \epsilon) \in \mathcal{E}^{*}_{X,G}$. In all these cases the FA frontier uniformly worsens when excluding $G$, completing the proof.
\end{proof}

\subsection{Fairness Criteria in the Literature} \label{sec:FairCriteria}

We review here certain fairness criteria that have appeared in the literature, and explain how these criteria can be accommodated within our framework.

\subsubsection{Statistical Parity.} This criterion seeks equality in decisions, namely that the proportion of either group receiving the two decisions is the same \citep{DworkHardtPitassiReingoldZemel}. Formally, an algorithm $a$ satisfies statistical parity if
\[\mathbb{E}(a(X)=1 \mid G=r) - \mathbb{E}(a(X)=1 \mid G=b) =0\]
The loss function 
\[\ell(d,y) = \left\{\begin{array}{cc}
1 & \mbox{ if } d=1 \\
0 & \mbox{ otherwise} \end{array} \right.\]
returns a relaxed version of this criterion, since 
\[e_g(a)= 
 \mathbb{E}\left[\ell(a(X),Y) \mid G=g\right] =  \mathbb{E}\left[a(X)=1 \mid G=g\right]
\]
 so $\vert e_r(a) - e_b(a) \vert$ is the absolute difference in the probability that a group-$r$ individual and a group-$b$ individual receive the decision $d=1$.

\subsubsection{False Positives.} Another common fairness criterion is equality of false positives across two groups \citep{ProPublica,chouldechova,KMR}. For example, among borrowers who would not have defaulted on their loan if approved, prediction of default should be equal across the two groups. Formally, an algorithm $a$ satisfies equality of false positive rates if
\[\mathbb{P}(a(X)=1 \mid Y=0,  G=r) - \mathbb{P}(a(X)=1 \mid Y=0, G=b)=0\]
To see how we can accommodate this, consider the group-dependent loss function
\[
\ell_{g}\left(d,y\right)=\frac{\mathbf{1}\left\{ d=1,y=0\right\} }{\mathbb{P}\left(\left.Y=0\right|G=g\right)},
\]
In this case, we obtain
\begin{alignat*}{1}
e_{g}\left(a\right) & =\frac{\mathbb{P}\left(\left.a\left(X\right)=1,Y=0\right|G=g\right)}{\mathbb{P}\left(\left.Y=0\right|G=g\right)}\\
 & =\frac{\mathbb{P}\left(a\left(X\right)=1,Y=0,G=g\right)}{\mathbb{P}\left(Y=0,G=g\right)}\\
 & =\mathbb{P}\left(\left.a\left(X\right)=1\right|Y=0,G=g\right)
\end{alignat*}
is the false-positive rate for group $g$, 
and so $\vert e_r(a) - e_b(a) \vert$ is the absolute difference in false positive rates. A fairness criterion based on the difference in false negative rates can be accommodated similarly.

An alternative (unconditional) false positive rate would be equalizing
\[\mathbb{P}(a(X)=1, Y=0 \mid  G=r) - \mathbb{P}(a(X)=1,  Y=0 \mid  G=b)=0\]
This is achieved by setting the loss function 
\[\ell(d,y) = \left\{\begin{array}{cc}
1 & \mbox{ if } (d,y)=(1,0)\\
0 & \mbox{ otherwise} \end{array} \right.\]
since in this case
\[e_g(a)= 
 \mathbb{E}\left[\ell(a(X),Y) \mid G=g\right] =  \mathbb{P}\left[a(X)=1,Y=0 \mid G=g\right].
\]

\subsubsection{Equalized Odds.} \label{app:EqualizedOdds} Another popular fairness criterion asks for equalized odds \citep{HardtPriceSrebro}, which an algorithm $a$ satisfies if
\begin{equation} \label{eq:EqOdds}
\mathbb{E}_Y[ \mathbb{E}_X[a(X) \mid G=r, Y] -  \mathbb{E}_X[a(X) \mid G=b, Y]] =0
\end{equation}
The inner difference compares the average decision for group-$r$ and group-$b$ individuals who share the same type $Y$, and the outer expectation averages over those values of $Y$.
 
 The group-dependent loss function 
\[\ell(d,y,g) = \left\{ \begin{array}{cc}
    \frac{P(Y=y)}{P(Y=y \mid G=g)} & \mbox{ if } d=1 \\
    0 & \mbox{otherwise} 
\end{array}\right.\]
returns a relaxed version of this criterion, since
\begin{align*}
\mathbb{E}[\ell(d,y,g)  \mid G=r] &= 
P(Y=0 \mid G=r) \times \mathbb{E}\left[ \frac{P(Y=0)}{P(Y=0 \mid G=r) } \times \mathbbm{1}(d=1) \mid G=r, Y=0 \right] \\
& \quad + P(Y=1 \mid G=r) \times \mathbb{E}\left[ \frac{P(Y=1)}{P(Y=1 \mid G=r) } \times \mathbbm{1}(d=1) \mid G=r, Y=1 \right] \\
& = P(Y=0) \times \mathbb{E}[\mathbbm{1}(d=1) \mid G=r, Y=0] \\
& \hspace{20mm} + P(Y=1) \times \mathbb{E}[\mathbbm{1}(d=1) \mid G=r, Y=1]
\end{align*}
so $\vert \mathbb{E}[\ell(a(X),Y,G) \mid G=r]  - \mathbb{E}[\ell(a(X),Y,G) \mid G=b] \vert$ is exactly the LHS of (\ref{eq:EqOdds}). As discussed in footnote \ref{fn:GroupDependence}, all of our results hold also for this group-dependent loss function.

\end{document}